\providecommand{\pgfsyspdfmark}[3]{}
\providecommand{\pdfinfo}[1]{}

\documentclass[a4paper,aps, noarxiv,onecolumn,10pt,accepted=2023-07-17]{quantumarticle}

\input{preamble.tex}

\usepackage{xspace}

\usetikzlibrary{intersections,fit,shapes.misc, decorations.markings}

\newcommand\marktopleft[1]{%
    \tikz[overlay,remember picture] 
        \node (marker-#1-a) at (-.5ex,1.5ex) {};%
}
\newcommand\markbottomright[1]{%
    \tikz[overlay,remember picture] 
        \node (marker-#1-b) at (-.5ex,-0.1ex) {};%
    \tikz[overlay,remember picture,inner sep=3pt]
        \node[draw=none,fill=red!20,rectangle,fill opacity=.2,fit=(marker-#1-a.center) (marker-#1-b.center)] {};%
}

\newcommand\marktopleftb[1]{%
    \tikz[overlay,remember picture] 
        \node (marker-#1-a) at (-.5ex,1.5ex) {};%
}
\newcommand\markbottomrightb[1]{%
    \tikz[overlay,remember picture] 
        \node (marker-#1-b) at (-.5ex,-0.1ex) {};%
    \tikz[overlay,remember picture,inner sep=3pt]
        \node[draw=none,fill=blue!20,rectangle,fill opacity=.2,fit=(marker-#1-a.center) (marker-#1-b.center)] {};%
}

\makeatletter
\def\namedlabel#1#2{\begingroup
    #2%
    \def\@currentlabel{#2}%
    \phantomsection\label{#1}\endgroup
}
\makeatother

\renewcommand\subsectionautorefname{Sec.}
\renewcommand\subsubsectionautorefname{Sec.}

\newcommand\Output{\hspace*{\algorithmicindent} \textbf{output: }}

\makeatletter
\patchcmd{\ALG@step}{\addtocounter{ALG@line}{1}}{\refstepcounter{ALG@line}}{}{}
\newcommand{\ALG@lineautorefname}{Line}
\makeatother

\renewcommand\phi{\varphi}

\newcommand\defaccr[2]{\newcommand#1{#2\xspace}}
\newcommand\defmath[2]{\newcommand#1{\ensuremath{#2}\xspace}}
\newcommand\concept[1]{\textit{#1}}

\defmath{\img}{\mathtt{image}}
\defmath{\apre}{\mathtt{\forall preimage}}

\defmath{\Post}{\mathit{Postfix}}

\let\set\undefined

\providecommand{\tuple}[1]{\ensuremath{\left( #1 \right)}}
\providecommand{\set}[1]{\ensuremath{\left\lbrace #1 \right\rbrace}}

\providecommand{\sizeof}[1]{\ensuremath{\left\vert{#1}\right\vert}}
\providecommand{\vect}[1]{\ensuremath{( \begin{matrix} #1 \end{matrix} )}}

\providecommand{\gen}[1]{\ensuremath{\left\langle #1 \right\rangle}}

\newcommand\qcsim[1]{\ensuremath{\textsc{QSim}_C^{#1}}\xspace}
\newcolumntype{x}[1]{>{\centering\arraybackslash\hspace{0pt}}p{#1}}
\defmath\OO{\Omega^*}

\newcommand\cliffordt{\text{Clifford + $T$}\xspace}

\newcommand{\id}[1][]{\ensuremath{\mathbb{I}_{#1}}\xspace}

\defmath{\bool}{\ensuremath{\mathbb{B}}}
\defmath{\naturalNumber}{\ensuremath{\mathbb{N}}}
\defmath{\complex}{\ensuremath{\mathbb{C}}}
\defmath{\real}{\ensuremath{\mathbb{R}}}
\defmath{\integers}{\ensuremath{\mathbb{Z}}}
\defmath{\conditionalind}{\mathrel{\text{\scalebox{1.07}{$\perp\mkern-10mu\perp$}}}}
\defmath{\dx}{\partial x}
\defmath{\ddx}{\sfrac{\partial}{\partial x}}
\defmath{\half}{\textstyle{\frac{1}{2}}}

\newcommand{\floor}[1]{\left\lfloor #1 \right\rfloor}

\defmath\Exists{\mathit{Exists}}
\defmath\PlusExists{\mathit{PlusExists}}
\defmath\var{\mathit{var}}
\defmath\calciso{\mathsf{calciso}}

\newcommand{\defn}{\,\triangleq\,}

\tikzstyle{oval} = [state, ellipse, minimum size=4mm, inner sep=0.5mm, node distance=1cm]
\tikzset{every picture/.style={->,thick}}

\tikzstyle{leaf}=[draw, rectangle,minimum size=4.mm, inner sep=3pt]
\tikzstyle{var}=[circle,draw=black!70,solid,thick,minimum size=6mm]
\tikzstyle{bdd}=[regular polygon, regular polygon sides=3, draw=black!70,solid,thick,inner sep=0.5mm]
\tikzstyle{n}=[->,loosely dashed,thick]
\tikzstyle{p}=[->,solid,thick]
\tikzstyle{b}=[->,densely dashdotted,ultra thick]
\tikzset{every node/.style={initial text={}, inner sep=2pt, outer sep=0}}
\tikzstyle{e0}[0]=[dashed,thick,bend right=#1]
\tikzstyle{e1}[0]=[solid, bend left =#1]

\tikzstyle{lbl}=[draw,fill=white,inner sep=2pt, minimum size=0cm,line width=.5pt]

\defmath\before{\prec}
\defmath\beforeq{\preccurlyeq}

\newcommand{\dyad}[1]{| #1 \rangle \langle #1 |}
\newenvironment{smallmat}{\left[\begin{smallmatrix}}{\end{smallmatrix}\right]}
\newcommand{\Iso}{\ensuremath{\text{Iso}}}
\newcommand{\Stab}{\ensuremath{\text{Stab}}}
\newcommand{\Aut}{\ensuremath{\text{Stab}}}

\newcommand\antii[2]{\begin{smallmat}0 & #1 \\ #2 & 0\end{smallmat}\xspace}
\newcommand\diag[1]{\begin{smallmat}1 & 0 \\  0 & #1\end{smallmat}\xspace}
\newcommand\diagg[2]{\begin{smallmat}#1 & 0 \\  0 & #2\end{smallmat}\xspace}

\defaccr{\bdd}{\textsf{BDD}}
\defaccr{\bdds}{\textsf{BDD}s}
\defaccr{\qmdd}{\textsf{QMDD}}
\defaccr{\add}{\textsf{ADD}}
\defaccr{\isoqmdd}{\textsf{LIMDD}}
\defaccr{\limdd}{\textsf{LIMDD}}
\defaccr{\qmdds}{\textsf{QMDD}s}
\defaccr{\adds}{\textsf{ADD}s}
\defaccr{\isoqmdds}{\textsf{LIMDD}s}
\defaccr{\limdds}{\textsf{LIMDD}s}
\defaccr{\glimdd}{\ensuremath{G}-\limdd}
\defaccr{\glimdds}{\ensuremath{G}-\limdds}

\renewcommand\index{\textsf{idx}\xspace}
\newcommand\leaf{\textsf{Leaf}\xspace}
\newcommand\lbl{\textsf{label}\xspace}
\newcommand\highlabel{\textsf{HighLabel}\xspace}
\newcommand\rootlabel{\textsf{RootLabel}\xspace}
\newcommand\unique{\textsc{Unique}\xspace}
\newcommand\cache{\textsc{Cache}\xspace}
\newcommand\autocache{\textsc{StabCache}\xspace}

\newcommand\Edge{\textsc{Edge}\xspace}
\newcommand\Node{\textsc{Node}\xspace}

\defmath\LIM{\textsc{LIM}}
\defmath\glim{G\text -\LIM}
\def\paulilim{\textnormal{\sc PauliLIM}}
\newcommand\Pauli{\textsc{Pauli}\xspace}
\newcommand\pauli{\Pauli}
\newcommand\makeedge{\textsc{MakeEdge}\xspace}

\newcommand{\project}{\textsc{UpdatePostMeas}}

\defmath\oh{\mathcal O}

\defmath\rootlim{B_{\textnormal{root}}}
\defmath\lowlim{B_{\textnormal{low}}}
\defmath\highlim{B_{\textnormal{high}}}
\defmath\gmax{g}
\defmath\kmax{\kappa^{\textnormal{final}}}

\defmath\cast{\mathbb C^\ast}

\DeclareMathOperator*{\argmin}{arg\,min}

\newcommand\follow[2]{\ensuremath{\textsc{follow}_{#1}(#2)}}
\defmath\plus{+}

\newlength{\pgfcalcparm}
\newlength{\pgfcalcparmm}

\DeclareRobustCommand{\leafedge}[2][]{%
  \pgftext{\settowidth{\global\pgfcalcparm}{\scriptsize $\,\,\,#2\,\,\,$}}%
  \raisebox{-.8mm}{%
  \tikz{%
    \node[inner sep=0pt] (x){};%
    \node[right=\pgfcalcparm of x,leaf](v){\scriptsize 1};%
    \draw (x) to node[above,pos=.5]{\scriptsize $\,#2\,\,$} (v);%
  }%
  }%
}

\DeclareRobustCommand{\ledge}[3][]{%
  \pgftext{\settowidth{\global\pgfcalcparm}{\scriptsize $\,\,\,#2\,\,\,$}}%
  \raisebox{-.8mm}{%
  \tikz{%
    \node[inner sep=0pt] (x){$#1\,\,$};%
    \node[state,inner sep=0pt,minimum size=10pt,right=\pgfcalcparm of x](v){\scriptsize $#3$};%
    \draw (x) to node[above,pos=.5]{\scriptsize $\,#2\,\,$} (v);%
  }%
  }%
}

\DeclareRobustCommand{\lnode}[5][]{%
    \pgftext{\settowidth{\global\pgfcalcparm}{\scriptsize $\,\,\,#2\,\,\,$}}%
    \pgftext{\settowidth{\global\pgfcalcparmm}{\scriptsize $\,\,\,#4\,\,\,$}}%
  ~\raisebox{-.mm}{%
  \tikz{%
  \vspace{-1mm}%
    \node[state,inner sep=0pt,minimum size=10pt] (v){\scriptsize $#1$};%
    \node[state,inner sep=0pt,minimum size=10pt,left=\pgfcalcparm of v](v0){\scriptsize $#3$};%
    \draw[dotted] (v) to node[above,pos=.45]{\scriptsize $#2$} (v0);%
    \node[state,inner sep=0pt,minimum size=10pt,right=\pgfcalcparmm of v](v1){\scriptsize $#5$};%
    \draw (v) to node[above,pos=.45]{\scriptsize $#4$} (v1);%
  }%
  }%
}

\DeclareRobustCommand{\lonenode}[3][]{%
    \pgftext{\settowidth{\global\pgfcalcparm}{\scriptsize $\,\,\,#2\,\,\,$}}%
    \pgftext{\settowidth{\global\pgfcalcparmm}{\scriptsize $\,\,\,#3\,\,\,$}}%
  ~\raisebox{-.mm}{%
  \tikz{%
  \vspace{-1mm}%
    \node[state,inner sep=0pt,minimum size=10pt] (v){\scriptsize $#1$};%
    \node[leaf,inner sep=0pt,minimum size=10pt,left=\pgfcalcparm of v](v0){\scriptsize $1$};%
    \draw[dotted] (v) to node[above,pos=.45]{\scriptsize $#2$} (v0);%
    \node[leaf,inner sep=0pt,minimum size=10pt,right=\pgfcalcparmm of v](v1){\scriptsize $1$};%
    \draw (v) to node[above,pos=.45]{\scriptsize $#3$} (v1);%
  }%
  }%
}

\newcommand\low[1]{\ensuremath{\textsf{low}(#1)}}
\newcommand\high[1]{\ensuremath{\textsf{high}(#1)}}
\defmath\Low{\ensuremath{\textsf{low}}}
\defmath\High{\ensuremath{\textsf{high}}}

\def\getautomorphisms{{\textsc{GetStabilizerGenSet}}\xspace}

\def\getsingleisomorphism{\textnormal{{\textsc{GetIsomorphism}}}\xspace}

\def\findisomorphismsetintersection{{\textsc{IntersectIsomorphismSets}}\xspace}

\def\none{{\tt None}}

\def\Prest{P_{\textnormal{rest}}}

\defmath\yy{\begin{smallmat}
    0 & y^*\\
    y & 0\\
\end{smallmat}}

\defmath\ww{\begin{smallmat}
      0 & y   \\
      y^* & 0  \\
  \end{smallmat}
}

\defmath\hv{{\hat v}}

\newcommand\qmddc{\ensuremath{\qmdd{}~\cup~\text{Stab}}\xspace}

\begin{document}

\title{LIMDD: A Decision Diagram for Simulation of Quantum Computing Including Stabilizer States}

\author{Lieuwe Vinkhuijzen${}^{\ast,}$}
\affiliation{Leiden University, The Netherlands}
\email{l.t.vinkhuijzen@liacs.leidenuniv.nl}
\orcid{0000-0002-8199-0901}

\author{Tim Coopmans${}^{\ast,}$}
\affiliation{Leiden University, The Netherlands}
\affiliation{Delft University of Technology, The Netherlands}
\email{t.j.coopmans@liacs.leidenuniv.nl}
\homepage{https://timcp.github.io/}
\orcid{0000-0002-9780-0949}

\author{David Elkouss}
\affiliation{Delft University of Technology, The Netherlands}
\affiliation{Networked Quantum Devices Unit, Okinawa Institute of Science and Technology Graduate University, Okinawa, Japan}
\email{d.elkousscoronas@tudelft.nl}
\homepage{https://www.davidelkouss.com}
\orcid{0000-0003-2023-2768}

\author{Vedran Dunjko}
\affiliation{Leiden University, The Netherlands}
\email{v.dunjko@liacs.leidenuniv.nl}
\homepage{https://liacs.leidenuniv.nl/$\sim$dunjkov}
\orcid{0000-0002-2632-7955}

\author{Alfons Laarman}
\affiliation{Leiden University, The Netherlands}
\email{a.w.laarman@liacs.leidenuniv.nl}
\homepage{https://alfons.laarman.com}
\orcid{0000-0002-2433-4174}

\blfootnote{$^\ast$ These authors contributed equally.}

\maketitle

\begin{abstract}
Efficient methods for the representation and simulation of quantum states and quantum operations are crucial for the optimization of quantum circuits.
Decision diagrams (DDs), a well-studied data structure originally used to represent Boolean functions, have proven capable of capturing relevant aspects of quantum systems, but their limits are not well understood.
  In this work, we investigate and bridge the gap between existing DD-based structures and the stabilizer formalism, an important tool for simulating quantum circuits in the tractable regime.
     We first show that although DDs were suggested to succinctly represent important quantum states,
     they actually require exponential space for certain stabilizer states.
	To remedy this, we introduce a more powerful decision diagram variant, called Local Invertible Map-DD (\limdd).
	We prove that the set of quantum states represented by poly-sized \limdds strictly contains the union of stabilizer states and other decision diagram variants. 
Finally, there exist circuits which \limdds can efficiently simulate, while their output states cannot be succinctly represented by two state-of-the-art simulation paradigms:
		the stabilizer decomposition techniques for Clifford + $T$ circuits
		 and Matrix-Product States.
	By uniting two successful approaches, \limdds thus pave the way for fundamentally more powerful solutions for simulation and analysis of quantum computing.
\end{abstract}

\newpage
\tableofcontents
\renewcommand\subsectionautorefname{Sec.}
\renewcommand\subsubsectionautorefname{Sec.}

\section{Introduction~\label{sec:introduction}}
\vspace{-.8em}
Classical simulation of quantum computing is useful for circuit design~\cite{zulehner2017one,burgholzer2020improvedEquivalenceChecking}, verification~\cite{burgholzer2021random,burgholzer2020advanced} and studying noise resilience in the era of Noisy Intermediate-Scale Quantum (NISQ) computers~\cite{preskill2018quantum}.
Moreover, identifying classes of quantum circuits that are classically simulatable, helps in excluding regions where a quantum computational advantage cannot be obtained.
For example, circuits containing only Clifford gates
 (a non-universal quantum gate set), using an all-zero initial state, only compute the so-called `stabilizer states' and can be simulated in polynomial time
\cite{gottesman1998heisenberg,aaronson2008improved,gottesman1997stabilizer,nest2005local,englbrecht2020symmetries}.
Stabilizer states, and associated formalisms for expressing them, are fundamental to many quantum error correcting codes~\cite{gottesman1997stabilizer} and play a role in measurement-based quantum computation~\cite{raussendorf2001oneway}.
In fact, simulation of universal quantum circuits is fixed-parameter tractable in the number of non-Clifford gates~\cite{bravyi2016trading}, which is why
	 many modern simulators are based on \concept{stabilizer decomposition}~\cite{bravyi2016trading,bravyi2017improved,bravyi2019simulation,huang2019approximate,kocia2018stationary,kocia2020improved}.

\begin{wrapfigure}[16]{r}{4.37cm}
\vspace{-1em}
\begin{tikzpicture}\footnotesize
  \tikzset{venn circle/.style={circle,minimum width=2cm,fill=#1,opacity=0.4}}
  \node [venn circle=white,minimum width=4cm,draw] (A) at (0,0.3) {};
  \node  at (0,1.95) 			{State space};
  \node [venn circle = red!50!white, ellipse,minimum height=3cm, minimum width=3.6cm] (L) at (0,0.3) {};
  \node  at (0,1.2) 		{poly-size \limdd};
  \node [venn circle = blue!70!white,text width=1.3cm,align=center,rotate=90,ellipse,minimum height=2cm, minimum width=2.5cm] (B) at (-.6,-.2) {};
  \node[text width=1.3cm,align=center]  at (-.6,-1.)
  							{poly-size MPS};
  \node [venn circle = yellow!90!white,text width=1.3cm,align=center, minimum width=1.8cm] (B) at (-.6,.1) 
  							{poly-size \qmdd};
  \node [venn circle = green!40!white,text width=1cm,align=center, minimum width=1.5cm] (C) at (.8,.1)
  							{Stabilizer states};

  \node[,inner sep=0pt] (a) at (1,-2) 		{cluster states};
  \node[fill=black,circle,inner sep=1.5pt] (b) at (1,-.4) {};
	\draw[-,bend left=-20] (a) to (b);

  \node[fill=black,circle,inner sep=1.5pt] (b) at (.5,-.8) {};
  \node[,inner sep=0pt] (a) at (-.55,-2.02) 		{(pseudo)};
	\draw[-,bend left=-20] (a.north) to (b);

\end{tikzpicture}  
\caption{
        The set of stabilizer states and states representable as
        poly\hbox{-}sized: (Pauli-)\limdds (this work), QMDDs and MPS.
}
\label{fig:venn-diagram} 
\end{wrapfigure}

Another method for simulating universal quantum computation is based on decision diagrams (DDs)~\cite{akers1978binary,bryant86,bryant1995verification,viamontes2004high}, including Algebraic DDs \cite{580054,viamontes2003improving,fujita1997multi,10.1145/157485.164569}, Affine Algebraic DDs~\cite{sanner2005AffineADDs}, Quantum Multi-valued DDs \cite{miller2006qmdd,zulehner2018advanced}, and Tensor DDs \cite{hong2020tensor}.
A DD is a directed acyclic graph (DAG) in which each path represents a quantum amplitude, enabling the succinct (and exact) representation of many quantum states through the combinatorial nature of this structure.
A DD can also be thought of as a homomorphic (lossless) compression scheme, since
various \concept{manipulation operations for DDs} exist which implement any quantum gate operation, including measurement (without requiring decompression).
Strong simulation is therefore easily implemented using a DD data structure~\cite{miller2006qmdd,zulehner2018advanced,hong2020tensor}.
Indeed, DD-based simulation was empirically shown to be competitive with state-of-the-art simulators~\cite{viamontes2004high,zulehner2018advanced,DBLP:conf/date/HillmichKMW21} and is used in several simulator and circuit verification implementations~\cite{viamontes2009quantum,hong2021approximate}.
DDs and the stabilizer formalism are introduced in \autoref{sec:preliminaries}.

\qmdds are currently the most succinct DD supporting quantum simulation, but in this paper we show that they require exponential size to represent a type of stabilizer state called a cluster state~\cite{briegel2000persistent}.
In order to unite the strengths of DDs and the stabilizer formalism and inspired by SLOCC (Stochastic Local Operations and Classical Communication) equivalence of quantum states \cite{dur2000three,chitambar2014everything}, in \autoref{sec:main-results}, we propose \limdd: a new DD for quantum computing simulation using local invertible maps (LIMs).
Specifically, \limdds eliminate the need to store multiple states which are equivalent up to LIMs, allowing more succinct DD representations.
For the local operations in the LIMs, we choose Pauli operations, creating a Pauli-\limdd, which we will simply refer to as \limdd.
We prove that there is a family of quantum states ---called pseudo cluster states---
 that can be represented by poly-sized
(Pauli\nobreakdash-)\limdds 
but that require exponentially-sized \qmdds and cannot be expressed in the stabilizer formalism.
We also show the same separation for matrix product states (MPS)~\cite{white1992density,cirac2021MPSPEPS,vidal2003efficient}. 
\autoref{fig:venn-diagram} visualizes the resulting separations.

Further, we give algorithms for simulating quantum circuits using Pauli-\limdds. We continue by investigating the power of these algorithms compared to state-of-the-art simulation algorithms based on \qmdd, MPS and stabilizer decomposition. 
We find circuit families which Pauli-\limdd can efficiently simulate, which stands in stark contrast to the exponential space needed by \qmdd-based, MPS-based and a stabilizer-decomposition-based simulator
(the latter result is conditional on the exponential time hypothesis).
This is the first analytical comparison between decision diagrams and matrix product states.

Efficient decision diagram operations for both classical \cite{darwiche2002knowledge} and quantum \cite{burgholzer2020improvedEquivalenceChecking} applications crucially rely on \emph{dynamic programming} (storing the result of each intermediate computation) and \emph{canonicity} (each quantum state has a unique, smallest representative as a \limdd)~\cite{brace1991efficient,knuth4,somenzi2001efficient}.
We provide algorithms for both in \autoref{sec:canonicity}.
Indeed, the main technical contribution of this paper
 is the formulation of a canonical form for Pauli-\limdds together with
 an algorithm which brings a Pauli-\limdd into this canonical form.
By interleaving this algorithm with the circuit simulation algorithms,  we ensure that the algorithms act on LIMDDs that are canonical and as small as possible.

The canonicity algorithm effectively determines whether two $n$-qubit quantum states
$\ket\phi,\ket{\psi}$, each represented by a \limdd node $\phi,{\psi}$,
are equivalent up to a Pauli operator~$P$, i.e, $\ket\phi =  P \ket\psi$, which we call an isomorphism between $\ket\phi$ and $\ket\psi$.
Here $P = P_n \otimes \dots \otimes P_1$ consists of single qubit Pauli operators $P_i$
		(ignoring scalars for now).
In general, there are multiple choices for $P$,
so the goal is to make a deterministic selection among them,
to ensure canonicity of the resulting \limdd.
To do so, we first take out one qubit and write the states as, e.g., $\ket\phi = c_0 \ket{0}\ket{\phi_0} + c_1 \ket{1}\ket{\phi_1}$ for complex coefficients~$c_0, c_1$.
We then realize that $\Prest = P_{n-1} \dots \otimes P_1$ must map the pair $(\ket{\phi_0}, \ket{\phi_1})$ to either $(\ket{\psi_0},\ket{\psi_1})$ or $(\ket{\psi_1}, \ket{\psi_0})$ (in case $P_n$ is a diagonal or antidiagonal, respectively).
Hence $\Prest$ is a member of the intersection of the two sets of isomorphisms.
Next, we realize that the set of all isomorphisms, e.g. mapping $\ket{\phi_0}$ to $\ket{\psi_0}$, is a coset $\pi \cdot G$ of the stabilizer group $G$ of $\ket{\phi_0}$ (i.e. the set of isomorphisms mapping $\ket{\phi_0}$ to itself) where $\pi$ is a single isomorphism $\ket{\phi_0} \rightarrow \ket{\psi_0}$.
Thus, to find a (canonical) isomorphism between $n$-qubit states $\ket{\phi} \rightarrow \ket{\psi}$ (or determine no such isomorphism exists), we need three algorithms: to find (a) an isomorphism between $(n-1)$-qubit states, (b) the stabilizer group of an $(n-1)$-qubit state (in fact: the group generators, which form an efficient description), (c) the intersection of two cosets in the Pauli group (solving the \textit{Pauli coset intersection problem}).
Task (a) and (b) are naturally formulated as recursive algorithms on the number of qubits, which invoke each other in the recursion step.
For (c) we provide a separate algorithm which first rotates the two cosets such that one is a member of the Pauli $Z$ group, hence isomorphic to a binary vector space, followed by using existing algorithms for binary coset (hyperplane) intersection.
Having characterized all isomorphisms $\ket\phi \rightarrow \ket\psi$, we select the lexicographical minimum to ensure canonicity.
We emphasize that the algorithm works for arbitrary quantum states, not only stabilizer states.

\section{Preliminaries \label{sec:preliminaries}}

Here, we briefly introduce two methods to manipulate and succinctly represent quantum states: decision diagrams, which support universal quantum computing, and the stabilizer formalism, in which a subset of all quantum computations is supported which can however be efficiently classically simulated.
Both support \concept{strong simulation}, i.e. the probability distribution of measurement outcomes can be computed (through \concept{weak simulation} one only samples measurement outcomes).
\defmath\leafnode{\raisebox{-.7mm}{\tikz{\node[draw,minimum width=.3cm,minimum height=.3cm]{\scriptsize 1};}}}

\subsection{Decision diagrams}

An $n$-qubit quantum state $\ket{\phi}$ can be represented as a $2^n$-dimensional vector of complex numbers (modeling amplitudes) and
can thus be described by a pseudo-Boolean function $f: \{0, 1\}^n \rightarrow \mathbb{C}$ where
\begin{equation}
    \label{eq:quantum-state-expansion}
    \ket\phi = \sum_{x_1, \dots, x_n \in \set{0,1}} f(x_n,\dots, x_1) \ket{x_n}\otimes  \dots \otimes\ket{x_1}.\vspace{-.5em}
\end{equation}
The Quantum Multi-valued Decision Diagram (\qmdd) \cite{miller2006qmdd} is a data structure which can succinctly represent functions of the form $f\colon \{0,1\}^n\to\mathbb C$,
and thus can represent any quantum state per~\autoref{eq:quantum-state-expansion}.
A \qmdd is a rooted DAG with a unique leaf node \leafnode, representing the value $1$.
\autoref{fig:qmdd-small-example}~(d) shows an example (and its construction from a binary tree).
Each node has two outgoing edges, called its \emph{low edge} (dashed line) and its \emph{high edge} (solid line).
The diagram has \concept{levels} as each node is labeled with (the index of) a variable;
the root has index $n$, its children $n-1$, etc, until the leaf with index 0
(the set of nodes with index $k$ form level~$k$).
Hence each path from root to leaf visits nodes representing the variables $x_3,x_2,x_1$ in order.
The value $f(x_n,\ldots, x_1)= \braket{x_n \dots x_1 | \phi} $ is computed by traversing the diagram, starting at the root edge and then
for each node at level $i$  following the low edge (dashed line) when $x_i=0$, and the high edge (solid line) when $x_i=1$, while multiplying the edge weights (shown in boxes) along the path, e.g., $f(1,1,0)= \frac{1}{2} \cdot 1 \cdot -\sqrt{2} \cdot 1 = -\frac 1{\sqrt 2}$ in \autoref{fig:qmdd-small-example}.

\begin{figure}[h]
\centering
~~~
\begin{tikzpicture}[
    scale=0.3,
    every path/.style={>=latex},
    every node/.style={},
    inner sep=0pt,
    minimum size=0.5cm,
    line width=1pt,
    node distance=1cm,
    thick,
    font=\footnotesize
    ]
    
    \node[draw,circle] (a1) {};
    \node[draw,circle, below left=of a1] (a2) {};
    \node[draw,circle, below right=of a1] (a3) {};
    \node[draw,circle, below right=of a2] (a4) {$u$};
    \node[draw,circle,rectangle,minimum size=0.4cm, below=of a4] (w1) {$1$};
    
    \draw[<-] (a1) --++(90:3cm) node[lbl,right,pos=.7] {$\nicefrac 12$};
    \draw[e0] (a1) edge  node[,left,pos=.2] {} (a2);
    \draw[e1] (a1) edge  node[,right,pos=.2] {} (a3);

    \draw[e0= 25] (a2) edge  node[lbl] {$0$} (a4);
    \draw[e1= 25] (a2) edge  node[] {} (a4);
    \draw[e0= 25] (a3) edge  node[] {} (a4);
    \draw[e1= 25] (a3) edge  node[lbl,below,pos=.7,right] {$-\sqrt 2$} (a4);
    \draw[e0= 25] (a4) edge  node[] {} (w1);
    \draw[e1= 25] (a4) edge  node[lbl] {$0$} (w1);
    
\node[left =.5cm of a1,yshift=.7cm] {d)};

    \node[draw,circle, left= 3.cm of a1] (a1) {};
    \node[draw,circle, below = .55cm of a1, xshift=-.9cm] (a2) {};
    \node[draw,circle, below = .55cm of a1, xshift= .9cm] (a3) {};
    \node[draw,circle, below = .55cm of a2, xshift= .9cm] (a41) {$u$};
    \node[leaf, below=of a41] (w1) {$1$};
    
    \draw[<-] (a1) --++(90:3cm) node[right,pos=.7] {};
    \draw[e0 = 0] (a1) edge  node[] {} (a2);
    \draw[e1 = 0] (a1) edge  node[] {} (a3);

    \draw[e0= 30] (a2) edge  node[lbl,left,pos=.4] {$0$} (a41);
    \draw[e1= 30] (a2) edge  node[lbl,right,pos=.2] {$\frac 12$} (a41);
    \draw[e0= 30] (a3) edge  node[lbl,left,pos=.2] {$\frac 12$} (a41);
    \draw[e1= 30] (a3) edge  node[lbl,right,pos=.4] {-$\frac 1{\sqrt2}$} (a41);
    \draw[e0=25] (a41) edge  node[lbl] {$1$} (w1);
    \draw[e1=25] (a41) edge  node[lbl] {$0$} (w1);
    
\node[left =.5cm of a1,yshift=.7cm] {c)};

    \node[draw,circle, left= 2.7cm of a1] (a1) {};
    \node[draw,circle, below = .55cm of a1, xshift=-.9cm] (a2) {};
    \node[draw,circle, below = .55cm of a1, xshift= .9cm] (a3) {};
    \node[draw,circle, below = .55cm of a2, xshift=-.3cm] (a41) {$v$};
    \node[draw,circle, below = .55cm of a2, xshift= .3cm] (a42) {$w$};
    \node[draw,circle, below = .55cm of a3, xshift=-.3cm] (a43) {$s$};
    \node[draw,circle, below = .55cm of a3, xshift= .3cm] (a44) {$t$};
    \node[leaf, below=of a42, xshift=.6cm] (w1) {$1$};
    
    \draw[<-] (a1) --++(90:3cm) node[right,pos=.7] {};
    \draw[e0 = 0] (a1) edge  node[] {} (a2);
    \draw[e1 = 0] (a1) edge  node[] {} (a3);

    \draw[e0=  0] (a2) edge  node[] {} (a41);
    \draw[e1=  0] (a2) edge  node[] {} (a42);
    \draw[e0=  0] (a3) edge  node[] {} (a43);
    \draw[e1=  0] (a3) edge  node[] {} (a44);
    \draw[e0= 70] (a41) edge  node[lbl] {$0$} (w1);
    \draw[e1=-25] (a41) edge  node[lbl] {$0$} (w1);
    \draw[e0= 10] (a42) edge  node[lbl] {$\frac 12$} (w1);
    \draw[e1= 20] (a42) edge  node[lbl,pos=.2] {$0$} (w1);
    \draw[e0= 10] (a43) edge  node[lbl] {$\frac 12$} (w1);
    \draw[e1= 20] (a43) edge  node[lbl] {$0$} (w1);
    \draw[e0=-30] (a44) edge  node[lbl,pos=.3] {-$\frac 1{\sqrt2}$} (w1);
    \draw[e1= 70] (a44) edge  node[lbl] {$0$} (w1);
    
\node[left =.5cm of a1,yshift=.7cm] {b)};

%
%

    \node[draw,circle, left= 3.5cm of a1] (a1) {};
    \node[draw,circle, below = .55cm of a1, xshift=-.9cm] (a2) {};
    \node[draw,circle, below = .55cm of a1, xshift= .9cm] (a3) {};
    \node[draw,circle, below = .55cm of a2, xshift=-.3cm] (a41) {};
    \node[draw,circle, below = .55cm of a2, xshift= .3cm] (a42) {};
    \node[draw,circle, below = .55cm of a3, xshift=-.3cm] (a43) {};
    \node[draw,circle, below = .55cm of a3, xshift= .3cm] (a44) {};
    \node[leaf, below=of a41, xshift=-.5cm] (w1) {$0$};
    \node[leaf, right= .02cm of w1,inner sep=0pt] (w2) {$0$};
    \node[leaf, right= .02cm of w2,inner sep=0pt] (w3) {$\frac 12$};
    \node[leaf, right= .02cm of w3,inner sep=0pt] (w4) {$0$};
    \node[leaf, right= .02cm of w4,inner sep=0pt] (w5) {$\frac 12$};
    \node[leaf, right= .02cm of w5,inner sep=0pt] (w6) {$0$};
    \node[leaf, right= .02cm of w6,inner sep=0pt] (w7) {$-\frac 1{\sqrt 2}$};
    \node[leaf, right= .02cm of w7,inner sep=0pt] (w8) {$0$};
        
    \draw[<-] (a1) --++(90:3cm) node[right,pos=.7] {};
    \draw[e0 = 0] (a1) edge  node[] {} (a2);
    \draw[e1 = 0] (a1) edge  node[] {} (a3);

    \draw[e0=  0] (a2) edge  node[] {} (a41);
    \draw[e1=  0] (a2) edge  node[] {} (a42);
    \draw[e0=  0] (a3) edge  node[] {} (a43);
    \draw[e1=  0] (a3) edge  node[] {} (a44);
    \draw[e0=  0] (a41) edge  node[] {} (w1);
    \draw[e1=  0] (a41) edge  node[] {} (w2);
    \draw[e0=  0] (a42) edge  node[] {} (w3);
    \draw[e1=  0] (a42) edge  node[] {} (w4);
    \draw[e0=  0] (a43) edge  node[] {} (w5);
    \draw[e1=  0] (a43) edge  node[] {} (w6);
    \draw[e0=  0] (a44) edge  node[] {} (w7);
    \draw[e1=  0] (a44) edge  node[] {} (w8);

\node[left =.5cm of a1,yshift=.7cm] {a)};

    \node[left=1.6cm of a1] (x) {$3$};
    \node[above =-.1cm of x] (y) {\hspace{-.7cm}Level:};
    \node[below=.55cm of  x] (x) {$2$};
    \node[below=.52cm of  x] (x) {$1$};


\end{tikzpicture}
\vspace{-2ex}
\caption{
Different decision diagrams representing the $3$-qubit state $[0,0,\frac 12,0,\frac 
12,0, -\frac 1{\sqrt 2},0]^\top$, evolving into a \qmdd (right).
Left, (a) shows the exponential binary tree, where a node on level $i$ represents $x_i$ (see \autoref{eq:quantum-state-expansion}) and its outgoing arrows $x_i=0$ (dashed) and $x_i=1$ (solid). The leaf contains the complex amplitude $f(x_1, x_2, x_3)$ (see \autoref{eq:quantum-state-expansion}) for the assignment of $(x_1, x_2, x_3)$ corresponding to the path from the root, e.g. $f(1, 1, 0) = - \frac{1}{\sqrt{2}}$.
 Next (b), the leaves are merged by dividing out common factors, putting these as weights (shown in boxes) on the edges going out of level-1 nodes (note in particular that we can suppress a separate 0 leaf, as $0 = 0\cdot 1$).  
Then the same trick is applied to level-1 nodes in (c).
In this example, all level-1 nodes $v, w, s, t$ become \concept{isomorphic} and can be merged into a new node $u$, representing the vector $\ket{u} = [1, 0]^{\top}$.
This can be done because the level-1 nodes $v, w, s, t$ respectively represent the vectors $[0, 0]^{\top}, [\frac{1}{2}, 0]^{\top}, [\frac{1}{2}, 0]^{\top}, [\frac{1}{\sqrt{2}}, 0]^{\top}$, which can be written as $c \cdot \ket{u} = c \cdot [1, 0]^{\top}$ for respective weights $c = 0, \frac{1}{2}, \frac{1}{2}, \frac{1}{\sqrt{2}}$.
 		Finally, (d) shows the resulting \qmdd, applying the same tactic to nodes on levels 2 and 3. Note that a \qmdd requires a root edge.
 		Merging (\concept{isomorphic}) nodes makes \qmdds succinct.
By convention, unlabelled edges have label $1$.
}
\label{fig:qmdd-small-example}
\end{figure}

A path from the root to a node $v$  with index $k$  (on level $k$) thus corresponds to a \emph{partial assignment} $(x_n=a_n,\ldots, x_{k-1}=a_{k-1})$, which induces 
subfunction $f_{\vec a}(x_{k},\ldots, x_1)\defn f(a_n,\ldots, a_{k-1},x_k,\ldots, x_1)$.
The node $v$ represents this subfunction \emph{up to a complex factor $\gamma$}, which is stored on
the edge incoming to $v$ along that path.
This allows any two nodes which represent functions equal up to a complex factor to be \emph{merged}.
For instance, the node $u$ on level 1 in \autoref{fig:qmdd-small-example}
represents $f_{01} = f_{10} = \frac {-1}{\sqrt{2}} f_{11} = 0 \cdot f_{00}$.
When all eligible nodes have been merged, the \qmdd is \emph{reduced}.
A reduced \qmdd is a \emph{canonical} representation: a given function has a unique reduced \qmdd.

Canonicity ensures that the \qmdd is always as small as possible as redundant nodes are merged.
But more importantly, canonicity allows for quick equality checks: two diagrams represent the same state \emph{if and only if} their root edges are the same (i.e., have the same label and point to the same root node). This allows for efficient \qmdd manipulation algorithms (i.e. updating the QMDD upon performing a gate or measurement) through dynamic programming, which avoids traversing all paths (exponentially many in the size of the diagram in the worst case).
For all quantum gates, there are algorithms to update the \qmdd accordingly and
measurement is also supported (even efficiently).
Therefore, \qmdds can simulate any quantum circuit, although they may become exponentially large (in the number of qubits) already after applying part of the gates from the circuit.
The resulting simulator is \concept{strong}, as the complete final state is computed as \qmdd (and computing measurement outcome probabilities on \qmdd is tractable).

Finally, we can also define the semantics of a node $v$ recursively, overloading Dirac notation:
$\ket v$.
For convenience, we denote an edge to node $v$ labeled with $\ell$ pictographically as $\ledge{\ell}{v}$.
Now a node $v$ with low edge $\ledge \alpha {v_0}$ and high edge $\ledge \beta{v_1}$, 
represents the state:
$\ket v \defn \alpha\ket{0}\otimes\ket{v_0}+\beta\ket{1}\otimes\ket{v_1}$, where in the base case $\ket{\leafnode} \defn 1$ as defined above. We later define \limdd semantics similarly.

\subsection{Pauli operators and stabilizer states}
\label{sec:prelims-stabilizers}

In contrast to decision diagrams, the \concept{stabilizer formalism}~\cite{gottesman1998heisenberg} forms a classically simulatable subset of quantum computation.
Instead of explicitly representing the (exponential) amplitude vector, the stabilizer formalism 
describes the symmetries a quantum state using so-called \concept{stabilizers}.
A unitary operator $U$ \concept{stabilizes} a state $\ket\phi$ if $\ket\phi$ is a $+1$ eigenvector of $U$, i.e., $U\ket\phi = \ket\phi$.
The formalism considers stabilizers $U$ made up of the single-qubit
Pauli operators $\pauli \defn \set{\id, X, Y, Z}$ as defined below.
In fact, a stabilizer is taken from the $n$-qubit Pauli group, defined as $\pauli_n \defn \gen{ \pauli^{\otimes n} }$, i.e. it is the group generated by all $n$-qubit \concept{Pauli strings} $P_n\otimes \dots \otimes P_1$ with $P_i \in \pauli$.
Here we used the notation $\gen G = H$ to denote that $G\subseteq H$ is a generator set for a group $H$.
One can check that 
$\pauli_n = \set{ i^c P_n\otimes \dots \otimes P_1 \mid P_1,\dots,P_n\in \pauli, c \in \set{0,1,2,3} }$, so in particular we have $\pauli_1 = \set{\pm P, \pm i P \mid P \in \Pauli}$ (the Pauli set with a factor $\pm1$ or $\pm i$).

\begin{equation*}
\id \defn \begin{pmatrix} 1 & 0\\ 0 & 1 \end{pmatrix},
X \defn \begin{pmatrix} 0 & 1\\ 1 & 0 \end{pmatrix},
Y \defn \begin{pmatrix} 0 & \hspace{-1.5ex}-i\\ i & 0 \end{pmatrix},
Z \defn \begin{pmatrix} 1 & 0\\ 0 & \hspace{-1.5ex}-1 \end{pmatrix}
\label{eq:pauli-matrices}
\end{equation*}

The set of Pauli stabilizers $\Stab(\ket\phi) \subset \Pauli_n$ of an $n$-qubit quantum state $\ket\phi$ necessarily forms a subgroup of $\Pauli_n$, since the identity operator $\id^{\otimes n}$ is a stabilizer of any $n$-qubit state and moreover if $U$ and $V$ stabilize $\ket\phi$, then so do $UV, VU$ and $U^{-1}$.
Furthermore, any Pauli stabilizer group $G$ is abelian, i.e. $A, B \in G$ implies $AB = BA$.
The reason for this is that elements of $\Pauli_n$ either commute ($AB=BA$) or anticommute ($AB = -BA$) under multiplication and anticommuting elements can never be stabilizers of the same state $\ket{\phi}$, because if $A, B \in \Stab(\ket\phi)$ and $AB = -BA$ then $\ket\phi = AB\ket\phi = -(BA)\ket\phi = -\ket{\phi}$, a contradiction.
Finally, note that $-\id^{\otimes n}$ can never be a stabilizer.
In fact, these conditions are necessary and and sufficient: the class of abelian subgroups $G$ of $\Pauli_n$, not containing $-\id^{\otimes n}$, are precisely all $n$-qubit stabilizer groups.
For clarity, we adopt the convention that we denote Pauli strings \emph{without phase} using the symbols $P,Q,R,\ldots$ and we use the symbols $A,B,C,\ldots$ for Pauli operators including phase; e.g., we may write $A=\lambda P$.
The phase $\lambda$ of any stabilizer $\lambda P \in \pauli$ can only be $\lambda =\pm 1$, derived as
\begin{equation}
\label{eq:stabilizer-scalar-is-pmone}
\forall \lambda P \in \Stab(\ket{\phi}) : 
\quad 
\ket{\phi} = (\lambda P)\ket{\phi} = (\lambda P)^2 \ket{\phi} =  \lambda^2 \id \ket{\phi} = \lambda^2 \ket{\phi}
\quad
\implies
\quad
\lambda = \pm 1
.
\end{equation}

The number of generators $k$ for a $n$-qubit stabilizer group $S$ can range from $1$ to $n$, and $S$ has $2^k$ elements.
If $k=n$, then there is only a single quantum state $\ket\phi$ (a single vector up to complex scalar) which is stabilized by $S$; such a state is called a \concept{stabilizer state}.
Equivalently,
$\ket \phi = C\ket{0}^{\otimes n}$ where $C$ is a circuit composed of only Clifford unitaries, a group generated by the Clifford gates:
\begin{equation*}
\textnormal{(Hadamard gate) }\mbox{ } H \defn \frac 1{\sqrt 2}\begin{pmatrix} 1 & 1\\ 1 & \hspace{-1.5ex}-1 \end{pmatrix},~~~
\textnormal{(phase gate) }\mbox{ } S \defn \begin{pmatrix} 1 & 0\\ 0 & \hspace{-1.5ex}-i \end{pmatrix},~~~
\text{and }CNOT \defn \left( \begin{smallmatrix}
            1 & 0 & 0 & 0\\
            0 & 1 & 0 & 0\\
            0 & 0 & 0 & 1\\
            0 & 0 & 1 & 0
        \end{smallmatrix}\right).
\end{equation*}
In the stabilizer formalism, an $n$-qubit stabilizer state is succinctly represented through $n$ independent generators of its stabilizer group, each of which is represented by $\mathcal O(n)$ bits to encode the Pauli string (plus factor), yielding $\mathcal O(n^2)$ bits in total.
Examples of (generators of) stabilizer groups are $\gen Z$ for $\ket{0}$  and
$\gen{X \otimes X, Z \otimes Z}$ for $\frac 1{\sqrt 2}(\ket{00} + \ket{11})$.
Updating a stabilizer state's generators after application of a Clifford gate or a single-qubit computational-basis measurement can be done in polynomial time in $n$ \cite{gottesman1998heisenberg, aaronson2008improved}.
Various efficient algorithms exist for manipulating stabilizer (sub)groups $S$, including testing membership (is $A\in \Pauli_n$ a member of $S$?) and finding a generating set of the intersection of two stabilizer (sub)groups.
These algorithms predominantly use standard linear algebra, e.g., Gauss-Jordan elimination, as described in \autoref{app:prelims-linear-algebra} in detail.

In this work, we also consider states which are not stabilizer states and which therefore have a nonmaximal stabilizer group (i.e. $< n$ generators).
To emphasize that a stabilizer group need not be maximal, i.e. it is a subgroup of maximal stabilizer groups, we will use the term \emph{stabilizer subgroup}.
Such objects are also studied in the context of simulating mixed states \cite{audenaert2005entanglement} and quantum error correction~\cite{gottesman1997stabilizer}.
Examples of stabilizer subgroups are $\{\id\}$ for $\frac1{\sqrt{2}}(\ket{0} + e^{i\pi/4}\ket{1})$, $\langle -Z\rangle$ for $\ket{1}$ and $\langle X \otimes X\rangle$ for $\frac1{\sqrt{6}}(\ket{00} + \ket{11}) + \frac{1}{\sqrt 3}(\ket{01} + \ket{10})$.
In contrast to stabilizer states, in general a state is not uniquely identified by its stabilizer subgroup.

\concept{Graph states} on $n$ qubits are the output states of circuits with input state $\frac 1{2^{n/2}} (\ket{0} + \ket{1})^ {\otimes n}$ followed by only $\textnormal{CZ} \defn \dyad{00} + \dyad{01} + \dyad{10} - \dyad{11}$ gates, and form a strict subset of all stabilizer states that is also important in error correction and measurement-based quantum computing~\cite{hein2006entanglement}.
By the (two-dimensional) cluster state on $n^2$ qubits, we mean the graph state whose graph is the $n \times n$ grid.

Given a vector space $V \subseteq \{0, 1\}^n$ and a length-$n$ bitstring $s$, the corresponding \concept{coset state} is $\frac{1}{\sqrt{|V|}} \sum_{x \in V} \ket{x + s}$ where `$+$' denotes bitwise xor-ing~\cite{aaronson2004multilinear}.
Each coset state is a stabilizer state.

\concept{Stabilizer decomposition-based methods}~\cite{bravyi2016trading,bravyi2017improved,bravyi2019simulation,huang2019approximate,kocia2018stationary,kocia2020improved} extend the stabilizer formalism to families of Clifford circuits with arbitrary input states $\ket{\phi_n}$, enabling the simulation of universal quantum computation~\cite{bravyi2005universal}.
By decomposing the $n$-qubit state $\ket{\phi_n}$ as linear combination of $\chi$ stabilizer states followed by simulating the circuit on each of the $\chi$ stabilizer states, the measurement outcomes can be computed in time $\oh(\chi^2 \cdot \textnormal{poly}(n))$, where the least $\chi$ is referred to as the \concept{stabilizer rank} of $\ket{\phi_n}$.
Therefore, stabilizer-rank based methods are efficient for any family of input states $\ket{\phi_n}$ whose stabilizer rank grows polynomially in $n$.

A specific method for obtaining a stabilizer decomposition of the output state of an $n$-qubit circuit is by rewriting the circuit into Clifford gates and $T = \dyad{0} + e^{i\pi/4}\dyad{1}$ gates (a universal gate set).
Next, each of the $T$ gates can be converted into an ancilla qubit initialized to the state $T\ket{+}$ where $\ket{+} = \frac 1{\sqrt 2}(\ket{0} + \ket{1})$; thus, an $n$-qubit circuit containing $t$ $T$ gates will be converted into an $n+t$-qubit Clifford circuit with input state $\ket{\phi} = \ket{0}^{\otimes n} \otimes \left(T\ket{+}\right)^{\otimes t}$ \cite{bravyi2016trading}.
We will refer to the resulting \textit{specific} stabilizer-rank based simulation method as the `Clifford + $T$ simulator,' whose simulation runtime scales with $\chi_t=\chi(\left(T\ket+\right)^{\otimes t})$, the number of stabilizer states in the decomposition of $\ket{\phi}$.
Trivially, we have $\chi_t \leq 2^t$, and although recent work~\cite{bravyi2016trading, bravyi2017improved} has found decompositions smaller than $2^t$ terms based on weak simulation methods, the scaling of $\chi_t$ remains exponential in $t$.
We emphasize that the Clifford + T decomposition is not necessarily optimal, in the sense that the intermediate states of the circuit might have lower stabilizer rank than $\ket{T}^{\otimes t}$ does.
Consequently, if a given circuit contains $t=\Omega(n)$ $T$-gates, then the Clifford + T simulator requires exponential time (in $n$) for simulating this $n$-qubit circuit, even if there exist polynomially-large stabilizer decompositions of each of the circuit's intermediate and output states (i.e., in principle, there might exist another stabilizer rank-based simulator that can simulate this circuit efficiently).

\subsection{Matrix product states}

Representing quantum states as \concept{matrix product states} (MPS) has proven  successful for solving a large range of many-body physics problems \cite{white1992density,perez2006matrixProductStateRepresentations}.
For qubits, an $n$-qubit MPS $M$ is formally defined as a series of $2n$ matrices $A_k^x \in \mathbb{C}^{D_k \times D_{k+1}}$ where $k\in [n], x\in \{0, 1\}, D_k \in \mathbb{N}_{\geq 1}$ and $D_1 = D_{n+1} = 1$.
Here, $D_{k+1}$ is the matrix dimension over the $k$-th \concept{bond}.
The interpretation $\ket{M}$ is determined as $\langle x_1 x_2 \dots x_n| M \rangle = A_1^{x_1} A_2^{x_2}\cdots A_n^{x_n}$ for $x_1, \dots, x_n \in \{0, 1\}$.
If the bond dimension may scale exponentially in the number of qubits, any family of quantum states can be represented exactly by an MPS.

The \concept{Schmidt rank} of a state $\ket\phi$ on $n$ qubits, relative to a bipartition of the qubits into two sets $A$ and $B$, is the smallest integer $m\geq 1$ such that $\ket\phi$ can be expressed as the superposition $\ket\phi=\sum_{j=1}^m c_j\ket{a_j}_A\ket{b_j}_B$ for complex coefficients $c_j$, where the states $\ket{a_j}_A$ ($\ket{b_j}_B$) form an orthonormal basis for the Hilbert space of the $A$ register ($B$ register).
The relation with MPS is that the maximum Schmidt rank with respect to any bipartition $A=\{x_1,\ldots, x_k\},B=\{x_k+1,\ldots, x_n\}$  is precisely the smallest possible bond dimension $D_{k+1}$ required to exactly express a state in MPS.

Vidal~\cite{vidal2003efficient} showed that MPS-based circuit simulation is possible in time $\oh(n\cdot\poly(\chi))$ per elementary operation, where $n$ is number of qubits and $\chi$ the maximum Schmidt rank for all intermediate states computed.

\section{Local Invertible Map Decision Diagrams}
\label{sec:main-results}

\autoref{sec:isomorphism-qmdd} introduces a \limdd definition parameterized with different local operations. We mainly consider the \pauli-\limdd and refer to it simply as \limdd.
We show how \limdds generalize \qmdds
and can represent arbitrary quantum states, normalized or not.
We then use this definition in \autoref{sec:exponential-separations} to show how \limdds
succinctly ---i.e., in polynomial space--- represent  graph states (in particular cluster states), coset states and, more generally, stabilizer states.
On the other hand, \qmdds and MPS require exponential size to represent two-dimensional cluster states.

We translate this exponential advantage in quantum state representation to (universal and strong)  quantum circuit simulation in \autoref{sec:quantum-simulation} by giving 	
	algorithms to update and query the  \limdd data structure.
These \concept{\limdd manipulation algorithms} take a \limdd $\phi$, representing some state $\ket\phi$, and return another \limdd $\psi$ that represents the state $\ket\psi = U\ket \phi$ for standard gates $U$ and also for arbitrary unitaries $U$
(by preparing $U$ in \limdd form first; we show how).
The measurement algorithm we give returns the outcome in linear time in size of the \limdd representation of the quantum state.

For many quantum operations, we show that our manipulation algorithms are efficient on all quantum states, i.e., take polynomial time in the size of the \limdd representation of the state.
Algorithms for certain other operations are efficient for certain classes of states, e.g., all Clifford gates can be applied in polynomial time to a \limdd representing a stabilizer state.
We show that \limdds can be exponentially faster than \qmdds, while they are never slower by more than a multiplicative factor $\oh(n^3)$. 
These algorithms use a \emph{canonical} form of \limdds, such that for each state there is a \emph{unique} \limdd. 
We defer this subject to \autoref{sec:canonicity}, which introduces \concept{reduced} \limdds and efficient algorithms to compute them.

With these algorithms, a quantum circuit simulator can be engineered by applying the circuit's gates one by one on the representation of the state as \limdd.
\autoref{prop:simulation} provides the bottom line of this section by comparing simulator runtimes.
In \autoref{sec:simulation}, we prove \autoref{prop:simulation}.

\begin{proposition}\label{prop:simulation}
Let \qcsim{\cliffordt} denote the runtime of the \cliffordt simulator on \nobreak{circuit~$C$}
(allowing for weak simulation as in \cite{bravyi2017improved}).
Let $\qcsim D$ denote the runtime of strong simulation
of circuit $C$ using method $D=(\pauli-)\limdd$, \qmdd, \qmddc, \text{MPS}, \qmddc.\footnote{We are not aware of any (potentially better) weak $D$-based simulation approaches and do not consider them.} %
Here, the latter is an (imaginary) ideal combination of \qmdd (not tractable for all Clifford circuits) and the stabilizer formalism (tractable for Clifford circuits), i.e., one that always inherits the best worst-case runtime from either method.
\\
The following holds, where $\OO$ discards polynomial factors, i.e., $\OO(f(n)) \defn \Omega(n^{\mathcal O(1)} f(n))$.
\begin{enumerate}
	\item[] There is a family of circuits $C$ such that:
    \item \quad \limdd is exponentially faster than \cliffordt: $ \qcsim{\cliffordt} = \Omega^*(2^n \cdot \qcsim \limdd)$,\footnote{Assuming the exponential time hypothesis (ETH). See \autoref{sec:simulation:cliffordt} for details.}
    \vspace{-\baselineskip}
\label{sim:cliffordt}
	\item \quad \limdd is exponentially faster than MPS: $\qcsim{\text{MPS}} = \Omega^*(2^n \cdot \qcsim \limdd)$, and\label{sim:mps}
	\item \quad \limdd is exponentially faster than \qmdd: $ \qcsim \qmdd = \Omega^*(2^n \cdot \qcsim \limdd)$.\label{sim:qmdd}
	\item For all $C$, \limdd is at worst cubically slower than \qmdd: $\qcsim \limdd = \mathcal O(n^3 \cdot \qcsim \qmdd)$.\label{sim:qmdd-limdd}
	\item \autoref{sim:qmdd} and \ref{sim:qmdd-limdd} hold when replacing \qmdd with \qmddc.\label{sim:qmddc}
\end{enumerate}
\end{proposition}

\subsection{The \limdd data structure}
\label{sec:isomorphism-qmdd}

Where \qmdds only merge nodes representing the same complex vector up to a constant factor, the \limdd data structure goes further by also merging nodes that are equivalent up to local operations, called Local Invertible Maps (LIMs) (see \autoref{def:isomorphism}).
As a result, \limdds can be exponentially more succinct than \qmdds, for example in the case of stabilizer states (see \autoref{sec:exponential-separations}).
    We will call nodes which are equivalent under LIMs, \emph{(LIM-) isomorphic}.
This definition generalizes SLOCC equivalence (Stochastic Local Operations and Classical Communication); if we choose the parameter $G$ to be the linear group, then the two notions coincide (see \cite[App. A]{dur2000three} and ~\cite{bennett1996concentrating,chitambar2014everything}).

\begin{definition}[$G$-LIM, $G$-Isomorphism]
	\label{def:isomorphism}
    An $n$-qubit $G$-Local Invertible Map (LIM) is an operator~$\mathcal{O}$ of the form $\mathcal{O}=\lambda \mathcal{O}_n\otimes\cdots\otimes \mathcal{O}_1$, where $G$ is a group of invertible $2\times 2$ matrices, $\mathcal{O}_i\in G$ and $\lambda\in\mathbb{C}\setminus \{0\}$.
    A $G$-\emph{isomorphism} between two $n$-qubit quantum states $\ket{\phi},\ket{\psi}$ is a LIM $\mathcal{O}$ such that $\mathcal{O}\ket{\phi} = \ket{\psi}$, denoted $\ket{\phi}\simeq_G\ket{\psi}$.
    Note that $G$-isomorphism is an equivalence relation.
\end{definition}

\begin{figure}[t!]
\centering
	\begin{tikzpicture}[-{>[scale=0.5]},>=stealth',shorten >=1pt,auto,node distance=.6cm,
    thick, state/.style={circle,draw,minimum size=14pt},font=\footnotesize]

\node[state, 
] (n1) {};

\node[state](n2)[below = of n1, xshift=-1.1cm]{$q_3$};
\node[state](n3)[below = of n1, xshift= 1.1cm, inner sep = 0pt]{$q_3'$};

\node[state](n21)[below = of n2, xshift=-.6cm, inner sep = 0pt]{$q_2$};
\node[state](n22)[below = of n2, xshift= .6cm, inner sep = 0pt]{$p_2$};
\node[state](n31)[below = of n3, xshift=-.6cm, inner sep = 0pt]{$q_2'$};
\node[state](n32)[below = of n3, xshift= .6cm, inner sep = 0pt]{$p_2'$};

\node[state](n41)[below = of n21, xshift= .6cm, inner sep = 0pt]{$q_1$};
\node[state](n42)[below = of n31, xshift= .6cm, inner sep = 0pt]{$q_1'$};

\node[draw, leaf,below = of n41, xshift=1.1cm] (e) {$1$};

\draw[<-] (n1) --++(90:1cm) node[lbl,pos=.6] {$\frac 14$};

\path[]
(n1) edge[e0] node[left,pos=.5] {} (n2)
(n1) edge[e1] node[lbl]			{-$1$} (n3)
(n2) edge[e0] node[left,pos=.5] {} (n21)
(n2) edge[e1] node[,lbl] 		{-$\omega$} (n22)
(n3) edge[e0] node[left,pos=.5] {} (n31)
(n3) edge[e1] node[lbl] {$\omega$} (n32)

(n21) edge[e0=20] node[left,pos=.5] {}  (n41)
(n21) edge[e1=20] node[lbl		  ] {$i$}  (n41)
(n22) edge[e0=0] node[left] {}  (n42)
(n22) edge[e1] node[lbl,right,yshift=.cm] {-$\omega$}  (n41)
(n31) edge[e0=10] node[left,pos=.5] {} (n42)
(n31) edge[e1=20] node[lbl,pos=.7	 ] {$i$} (n42)
(n32) edge[e0=0] node[left,pos=.5] {} (n41)
(n32) edge[e1=20] node[lbl       ] {-$\omega$} (n42)

(n41) edge[e0=20] node[pos=.5] {} (e)
(n41) edge[e1=20] node[pos=.5] {} (e)
(n42) edge[e0=20] node[pos=.5] {} (e)
(n42) edge[e1=20] node[lbl] {$-1$} (e);

\node[left =.5cm of n1,yshift=.7cm] {a)};

\node[state, right = 3.7cm of n1, label=right:{}] (n1) {};

\node[state](n2)[below = of n1, xshift=-1.1cm]{$q_3$};
\node[state](n3)[below = of n1, xshift= 1.1cm, inner sep = 0pt]{$q_3'$};

\node[state](n21)[below = of n2, xshift=-.6cm, inner sep = 0pt]{$q_2$};
\node[state](n22)[below = of n2, xshift= .6cm, inner sep = 0pt]{$p_2$};
\node[state](n31)[below = of n3, xshift=-.6cm, inner sep = 0pt]{$q_2'$};
\node[state](n32)[below = of n3, xshift= .6cm, inner sep = 0pt]{$p_2'$};

\node[state](n4)[below = of n22, xshift= .6cm, inner sep = 0pt]{$\ell_1$};

\node[draw, leaf,below = of n4] (e) {$1$};

\draw[<-] (n1) --++(90:1cm) node[lbl,pos=.6] {$\frac 14$};

\path[]
(n1) edge[e0] node[left,pos=.5] {} (n2)
(n1) edge[e1] node[lbl]			{$-1$} (n3)
(n2) edge[e0] node[left,pos=.5] {} (n21)
(n2) edge[e1] node[,lbl] 		{-$\omega$} (n22)
(n3) edge[e0] node[left,pos=.5] {} (n31)
(n3) edge[e1] node[lbl] {$\omega$} (n32)

(n21) edge[e0= 60] node[left,pos=.5] {}  (n4)
(n21) edge[e1=-30] node[lbl,pos=.1,below right] {$i$}  (n4)
(n22) edge[e0= 40] node[lbl,left,pos=.2] {$Z$}  (n4)
(n22) edge[e1= 0] node[lbl,pos=.6,xshift=-.3cm] {-$\omega$}  (n4)
(n31) edge[e0=  0] node[left,pos=.2,lbl] {$Z$} (n4)
(n31) edge[e1= 40] node[lbl,pos=.22,	 ] {$iZ$} (n4)
(n32) edge[e0=-20] node[left,pos=.5] {} (n4)
(n32) edge[e1= 60] node[lbl       ] {-$\omega Z$} (n4)

(n4) edge[e0=20] node[pos=.5] {} (e)
(n4) edge[e1=20] node[pos=.5] {} (e);

\node[left =.5cm of n1,yshift=.7cm] {b)};

\node[state, right = 3.1cm of n1, label=right:{}] (n1) {};

\node[state](n2)[below = .3cm of n1, xshift=-1.1cm, inner sep = 0pt]{$q_3$};
\node[state](n3)[below = .3cm of n1, xshift= 1.1cm, inner sep = 0pt]{$q_3'$};

\node[state](n21)[below = .9cm of n2, inner sep = 0pt]{$\ell_2$};
\node[state](n31)[below = .9cm of n3, inner sep = 0pt]{$\ell_2'$};

\node[state](n4)[below = of n21, xshift= 1.1cm, inner sep = 0pt]{$\ell_1$};

\node[draw, leaf,below = of n4] (e) {$1$};

\draw[<-] (n1) --++(90:1cm) node[lbl,pos=.6] {$\frac 14$};

\path[]
(n1) edge[e0] node[left,pos=.5] {} (n2)
(n1) edge[e1] node[lbl]			{$-1$} (n3)
(n2) edge[e0] node[below,pos=.2,lbl] {$\id\otimes  \id$} (n21)
(n2) edge[e1=10] node[pos=.0,yshift=-.2cm,lbl]	{-$\omega TZX\otimes \id$} (n31)
(n3) edge[e0=-20] node[lbl,below,pos=.6] {$\id\otimes Z$} (n21)
(n3) edge[e1] node[lbl,pos=.4,xshift=-.2cm] {$\omega T\otimes X$} (n31)

(n21) edge[e0= 60] node[left,pos=.5] {}  (n4)
(n21) edge[e1=-30] node[lbl,pos=.1,below right] {$i\cdot \id$}  (n4)
(n31) edge[e0=  0] node[left,pos=.2] {} (n4)
(n31) edge[e1= 40] node[lbl,pos=.3,	 ] {$Z$} (n4)

(n4) edge[e0=20] node[pos=.5] {} (e)
(n4) edge[e1=20] node[pos=.5] {} (e);

\node[left =.5cm of n1,yshift=.7cm] {c)};

\node[state, right = 2.5cm of n1, label=right:{}] (n1) {};

\node[state](n2)[below = of n1]{$\ell_3$};

\node[state](n21)[below = of n2, xshift=-.9cm, inner sep = 0pt]{$\ell_2$};
\node[state](n22)[below = of n2, xshift= .9cm, inner sep = 0pt]{$\ell_2'$};

\node[state](n41)[below = of n21, xshift= .9cm]{$\ell_1$};

\node[draw, leaf,below = of n41] (e) {$1$};

\draw[<-] (n1) --++(90:1cm) node[lbl,pos=.6,right] {$\frac 14 \cdot \id \otimes\id \otimes\id \otimes X$};

\path[]
(n1) edge[e0=50] node[    ] {} (n2)
(n1) edge[e1=30] node[pos=.4767,lbl,xshift=-.4cm] {$Z\otimes \id \otimes Z$} (n2)
(n2) edge[e0   ] node[    ] {} (n21)
(n2) edge[e1   ] node[pos=.6,lbl,xshift=-.6cm] 		{$\omega \cdot T \otimes Z$} (n22)

(n21) edge[e0=20] node[   ] {}  (n41)
(n21) edge[e1=20] node[lbl,pos=.3] {$i\cdot \id$}  (n41)
(n22) edge[e0=20] node[   ] {}  (n41)
(n22) edge[e1=20] node[lbl] {$Z$}  (n41)

(n41) edge[e0=30] node[pos=.5] {} (e)
(n41) edge[e1=30] node[pos=.5] {} (e);

\node[left =.5cm of n1,yshift=.7cm] {d)};

    \node[right=1.9cm of n1,yshift=.5mm] (x) {$4$};
    \node[left =-.cm of x,yshift=.5mm] (y) {\hspace{-.8cm}Level:};
    \node[below=.8cm of  x] (x) {$3$};
    \node[below=.8cm of  x] (x) {$2$};
    \node[below=.8cm of  x] (x) {$1$};

\end{tikzpicture}
	\vspace{-1em}
	\caption{A \qmdd (a) representing the state 
	$\frac 14 [1,1,i,i, -\omega,\omega,i,i, -1,1,-i,i,-\omega,-\omega, i, -i]^\top$
	with $\omega=e^{i\pi/4}$, evolving into a \limdd (d).
	As in \autoref{fig:qmdd-small-example}, diagram nodes are horizontally ordered in `levels' with qubit indices ${4,3,2,1}$. Low edges are dashed, high edges solid. See the text for an explanation.
	\\By convention, unlabelled edges have label $1$ (for \qmdd) or $\id^{\otimes k}$ (for \limdd nodes at level~$k$).
    }
	\label{fig:qmdd-isoqmdd-exposition}
\end{figure}

We define $\paulilim_n \defn  \gen{\pauli}$-\LIM, %
i.e., the group of Pauli operators $P\in\pauli_n$ with arbitrary complex factor~$\lambda \in \mathbb{C}\setminus\{0\}$ ($\lambda$ can absorb the factor $\gamma=\pm1,\pm i$ in $P = \gamma P_n \otimes\dots \otimes P_1$. Note $\lambda=\pm 1$ still for $\paulilim_n$ operators which are stabilizers, by eq.~\eqref{eq:stabilizer-scalar-is-pmone}).

Before we give the formal definition of \limdds in \autoref{def:limdd}, we give a motivating example in \autoref{fig:qmdd-isoqmdd-exposition}, which uses $\gen{X,Y,Z,T}$-\LIM{s} to demonstrate how the use of isomorphisms can yield small diagrams for a four-qubit state.
This figure shows how to merge nodes in four steps, shown in subfigures (a)-(d), starting with a large \qmdd (a) and ending with a small \limdd (d).
In the \qmdd (a), the nodes labeled $q_1$ and $q_1'$ represent the single-qubit states $\ket{q_1}=\left[1, 1\right]^\top$ and $\ket{q_1'}=\left[1, -1\right]^\top$, respectively.
By noticing that these two vectors are related via $\ket{q_1'}=Z\ket{q_1}$, we merge nodes $q_1,q_1'$ into node $\ell_1$ in (b), storing the isomorphism $Z$ on all incoming edges that previously pointed to $q_1'$.
From step (b) to (c), we first merge $q_2,q_2'$ into $\ell_2$,
	observing that $\ket{q_2'}= \id \otimes Z \ket{q_2}$.
Second, we create a node $\ell_2'$ such that
	$\ket{p_2}=  TZX\otimes I\ket{\ell_2'}$ and
	$\ket{p_2'}= T  \otimes X\ket{\ell_2'}$.
	So we can merge nodes $p_2,p_2'$ into $\ell_2'$, placing these isomorphisms on the respective edges.
To go from (c) to (d), we merge nodes $q_3,q_3'$ into node $\ell_3$ by noticing that $\ket{q_3'}=(Z\otimes \mathbb I\otimes Z)\ket{q_3}$.
This isomorphism $Z\otimes\id\otimes Z$ is stored on the high edge out of the root node.
We have $\ket{q_3}=\id\otimes\id\otimes X\ket{\ell_3}$, so we propagate the isomorphism $\id\otimes\id\otimes X$ upward, and store it on the root edge.
Therefore, the  final \limdd has the LIM $\frac 14 \id \otimes\id \otimes\id \otimes X$ on its root edge.

The resulting data structure in \autoref{fig:qmdd-isoqmdd-exposition} is a \limdd of only six nodes instead of ten, but requires additional storage for the LIMs.
\autoref{sec:exponential-separations} shows that merging isomorphic nodes sometimes leads to exponentially smaller diagrams, while the additional cost of storing the isomorphisms results only costs a linear factor of space (linear in the number of qubits).

The transformation presented above (for \autoref{fig:qmdd-isoqmdd-exposition}) only considers particular choices for LIMs. For instance, it would be equally valid to select LIM $\id \otimes Z$ instead of $-\id \otimes XZ$ for mapping $q_2'$ onto $q_2$. In fact, efficient algorithms to select LIMs in such a way that a canonical \limdd is obtained are a cornerstone for the \limdd manipulation algorithms presented in \autoref{sec:quantum-simulation}.
\autoref{sec:canonicity} provides a solution for \gen\pauli-LIMs (the basis for all results presented in the current article), which is based on using the stabilizers of each node, e.g., the group generated by $\set{\id \otimes X, Y\otimes \id}$~for~$q_2$.

\tikzset{every node/.style={initial text={}, inner sep=2pt, outer sep=0}}

\begin{definition}
	\label{def:limdd}
    An $n$-\glimdd is a rooted, directed acyclic graph (DAG)
     representing an $n$-qubit quantum state. %
    Formally, it is a $6$-tuple $(\Node\cup \{\leaf\}, \index,\Low,\High,\lbl,e_{\text{root}})$,~where:
\begin{itemize}
    \item $\leaf$ (a sink) is a unique leaf node with qubit index $\index(\leaf) = 0$;
	\item $\Node$ is a set of nodes with qubit indices $\index(v) \in \{1,\ldots, n\}$ for $v\in \Node$;
	\item $e_{\text{root}}$ is a root edge without source pointing to the root node $r\in \Node$ with $\index(r) = n$;
	\item $\Low,\High \colon \Node \to \Node\cup \{\leaf\}$ indicate the low  and high edge functions, respectively.  We write $\Low_v$ (or $\High_v$) to obtain the edge $(v,w)$ with $w = \low v$ (or $w = \high v$).
		For all $v \in \Node$ it holds that $\index(\low v) = \index(\high v) = \index(v)  - 1$ (no qubits are skipped\footnote{\label{fn:skipping}Decision diagram definitions~\cite{akers1978binary,bryant86,feinstein2011skipped} often allow to skip (qubit) variables, interpreting them as `don't cares.' We disallow this here, since it complicates definitions and proofs, while at best it yields linear size reductions~\cite{knuth4}.
});
	\item $\lbl\colon \Low\cup \High\cup\set{e_{\text{root}}}\to k-G$-$\LIM\cup\set 0$ is a function labeling edges $(\,.\,, w)$ with $k$-$G$-LIMs or $0$, where $k = \index(w)$ 
\end{itemize}
\vspace{-1em}
We will find it convenient to write $\lnode[u]{A}{v}{B}{w}$ for a node $u$ with low and high edges to nodes $v$ and $w$ labeled with $A$ and $B$, respectively.
We will also denote \ledge Av for a (root) edge to $v$ labeled with~$A$. When omitting $A$ or $B$, e.g., $\ledge {}v$, the LIM should be interpreted as $\id^{\otimes \index(v)}$.

We define the semantics of a leaf, node $v$ and an edge $e$ to node $v$ by overloading the Dirac notation:
\begin{align*}
\ket \leaf &\defn 1\\
\ket{e} & \defn 
      \lbl(e) \cdot \ket v
\\
\ket{v}     & \defn \ket{0}\otimes\ket{\Low_v}+ \ket{1}\otimes \ket{\High_v}
\end{align*}
\end{definition}
It follows from this definition that a node $v$ with $\index(v)=k$ represents a quantum state on $k$ qubits.
\emph{This state is however not necessarily normalized:}
For instance, a normalized state $\alpha\ket0 + \beta \ket1$,
can be represented as a \limdd \lonenode[v] \alpha\beta
or a \limdd \lonenode[v] {}{\nicefrac\beta\alpha} with root edge
\ledge \alpha{v}. So the node $v$ represents a state up to global scalar.
But, in general, any scalar can be applied to the root edge, or any other edge for that matter. So \limdds can represent any complex vector.

The tensor product $\ket{e_0} \otimes \ket{e_1}$ of the $G$-\limdds with root edges
$\ledge[e_0] {A}{v}$ and $\ledge[e_1] {B}{w}$   can be computed  just like for \qmdds~\cite{miller2006qmdd}: Take all edges $\leafedge \alpha$ pointing to the leaf in the \limdd $e_0$ and replace them with edges $\ledge {\alpha \cdot B}{w}$ pointing to the $e_1$ root node $w$. The result is an $n +m$ level \limdd if $e_0$ has $n$ levels and  $e_1$ has $m$.
 In addition, the LIMs $C$ on the other edges in the \limdd~$e_0$ should be extended to $C \otimes \id^{\otimes m}$.

We can now consider various instantiations of the above general \limdd definition  for different LIM groups $G$.
A \glimdd with $G=\set{\mathbb I}$ yields precisely all \qmdds by definition, i.e., all edges labels effectively only contain scalars.
As all groups $G$ contain the identity operator $\mathbb I$, the universality of \glimdds (i.e., all quantum states can be represented) follows from the universality of \qmdds.
It also follows that any state that can efficiently be represented by \qmdd, can be efficiently represented by a \glimdd for any $G$.
Similarly, we can consider $\gen Z$ and $\gen X$, which are subgroups of the Pauli group, and define a \gen Z-\limdd and a \gen X-\limdd; instances that we will study for their relation to graph states and coset states
 in \autoref{sec:exponential-separations}.
 Finally, and most importantly, \gen\Pauli-\limdds can represent all stabilizer states in polynomial space, which is a feature that neither \qmdds nor matrix product states (MPS) posses, as shown in \autoref{sec:exponential-separations}.
 
In what follows, we only consider \gen X-, \gen Z-, and \gen\Pauli-\limdds, or \pauli-\limdd for short. For \pauli-\limdds, we now illustrate how to find the amplitude of a computational basis state $\langle x | \psi \rangle$ for a bitstring $x \in \{0, 1\}^n$ by traversing the \limdd
of the state $\ket\psi$ from root to leaf, as follows.
Suppose that this diagram's root edge $e_{\text{root}}$ points to node $r$ and is labeled with the LIM $\lbl(e_{\text{root}})=A=\lambda P_{n}\otimes \cdots\otimes P_1 \in \paulilim_n$.
First, we substitute $\ket r = \ket 0\ket{\Low_r}+\ket 1\ket{\High_r}$, where $\Low_r,\High_r$ are the low and high edges going out of $r$, thus obtaining $\braket{x|\psi}=\braket{x|e_{\text{root}}} = \bra{x} A\left( \ket 0\ket{\Low_r}+\ket1\ket{\High_r} \right)$.
Next, we notice that $\bra x A=\lambda (\bra{x_n}P_n)\otimes \cdots\otimes (\bra{x_1}P_1)=\gamma\bra y$ for some $\gamma\in \mathbb C$ and a computational basis state $\bra y$.
Therefore, letting $y'=y_{n-1}\ldots y_1$, it suffices to compute $\bra{y_n}\bra {y'}(\ket 0\ket{\Low_r}+\ket 1\ket{\High_r})$, which reduces to computing either $\braket{y'|\Low_r}$ if $y_n=0$, or $\braket{y'|\High_r}$ if $y_n=1$.
By applying this simple rule repeatedly, one walks from the root to the leaf, encountering one node on each level.
The amplitude $\braket{x|\psi}$ is then found by multiplying together the scalars $\gamma$ found along this path. \autoref{alg:amplitude} formalizes this. Its runtime is $\oh(n^2)$. %

\begin{algorithm}
	\begin{algorithmic}[1]
		\Procedure{ReadAmplitude}{\Edge $\ledge[e]{A}{v}$, $x_n, \dots, x_1 \in \set{0,1}$ \textbf{with} $n = \index(v)$}
        \If{$n = 0$} \Return $\lambda$
        \EndIf
		\Stateh $\gamma \bra{y_n \dots y_1} := \bra{x_n \dots x_1} \lambda P_n \otimes \dots \otimes P_1$ \Comment{$\oh(n)$-computable LIM operation}
		\label{algline:amplitude:lim}
		\If{$y_n = 0$}\Comment{$y_n = 0$}
		\State \Return $\gamma \cdot \textsc{ReadAmplitude}(\Low_v, y_{n-1}, \dots, y_1)$
		\Else\Comment{$y_n = 1$}
		\State \Return $\gamma \cdot \textsc{ReadAmplitude}(\High_v, y_{n-1}, \dots, y_1)$
		\EndIf
		\EndProcedure
	\end{algorithmic}
	\caption{Read the amplitude for basis state $\ket{x_n \dots x_1}$ from $n$-qubit state $\ket e = A\cdot\ket { v}$ with $A = \lambda P_n \otimes \dots \otimes P_1 \in \paulilim_n$.}
	\label{alg:amplitude}
\end{algorithm}

\subsection{Succinctness of \limdds}
\label{sec:exponential-separations}

Succinctness is crucial for efficient simulation, as we show later.
In this section, we show exponential advantages for representing states with \limdds over two other state-of-the-art data structures: \qmdds and Matrix Product States (MPS)~\cite{white1992density,perez2006matrixProductStateRepresentations}.
Specifically, \qmdds and MPS require exponential space in the number of qubits to represent specific stabilizer states called (two-dimensional) cluster states.
We also show that an ad-hoc combination of \qmdd with the stabilizer formalism
still requires exponential space for `pseudo-cluster states.'
These results are visualized in \autoref{fig:venn-diagram}.

\subsubsection{\limdds are exponentially more succinct than \qmdds (union stabilizer states)}

\begin{wrapfigure}[20]{r}{7cm}\footnotesize
\vspace{-1em}
	\begin{tikzpicture}[node distance=.7cm,minimum height=.5cm]
		\node[draw] (pl) 			   {Pauli-\limdd};
		\node[draw, below =of pl, yshift=.1cm] (pq) {\qmddc};
		\node[draw, below = .3cm of pq, xshift=-2.5cm] (xl) {\gen X-\limdd};
		\node[draw, below = .3cm of pq, xshift= 2.5cm] (zl) {\gen Z-\limdd};
		\node[draw, below =of xl, xshift= 2.5cm, yshift=-.6cm] (qm) {\qmdd};

		\draw[-stealth] (xl.north) to (pl.south west);
		\draw[-stealth] (zl.north) to (pl.south east);
		\draw[-stealth] (qm) to (zl.south west);
		\draw[-stealth] (qm) to (xl.south east);

		\node[draw, fill=lightgray, below =of pq, yshift=.4cm, xshift=0cm] (ss) {stabilizer states};
		\node[draw, fill=lightgray, below = .4cm of ss, xshift= 2.5cm] (gs) {graph states};
		\node[draw, fill=lightgray, below = .4cm of ss, xshift=-2.5cm] (co) {coset states};
		\node[draw, fill=lightgray, below = .7cm of gs, yshift=.3cm] (cs) {2D cluster states};

		\draw[-stealth] (gs.north west) to (ss);
		\draw[-stealth] (co.north east) to (ss);
		\draw[-stealth] (cs) to (gs);

		\draw[-stealth,bend left=20] (pq) to (pl);
		\draw[-stealth] (ss) to (pq);
		\draw[-stealth] (co) to (xl);
		\draw[-stealth] (gs) to (zl);

		\draw[-stealth] (gs) to (zl);
		\draw[-stealth] (gs) to (zl);
		
		\draw[dashed, bend left=20] (pl) to  %
					(pq);
		\draw[dashed] (cs) to (qm.east);
		\draw[dashed] (co.south) to (qm.west);
	\end{tikzpicture}
	\caption{
	 Relations between non-universal classes of quantum states (gray) and decision diagrams, where we consider a diagram as the set of states that it can represent in  polynomial size.
	 Solid arrows denote set inclusion.
	 Dashed arrows $D_1\dashrightarrow D_2$ signify an exponential separation between two classes,
	 i.e., some quantum states have polynomial-size representation in $D_1$, but only exponential-size in $D_2$.\\
	 By transitivity, \qmdd is exponentially separated from all representations (not drawn for clarity).
    }
    \label{fig:complexity-classes-inclusion-diagram}
    \vspace{-1em}
\end{wrapfigure}

\autoref{fig:complexity-classes-inclusion-diagram} visualizes succinctness relations between different quantum state representations, as proved in \autoref{prop:separation}.
In particular, $G$-\limdds with $G = \gen\pauli$ can be exponentially more succinct than \qmdds, and retain this exponential advantage even with $G=\braket{Z},\braket{X}$.
In \autoref{thm:limdd-superset-qmdd-plus-stabilizers}, we show the strongest result, namely that \limdds are also more succinct than the union of \qmdds and stabilizer states, written \qmddc, which can be thought of a structure that switches between \qmdd and the stabilizer formalism depending on its content (stabilizer or non-stabilizer state). This demonstrates that ad-hoc combinations of existing formalisms do not make \limdds obsolete.

\begin{proposition}\label{prop:separation}
The inclusions and separations in \autoref{fig:complexity-classes-inclusion-diagram} hold.
\end{proposition}
\begin{proof}
The inclusions between the sets of states shown in gray are well known~\cite{aaronson2004multilinear,hein2006entanglement}.
The inclusions between decision diagrams hold because, e.g., a \qmdd is a $G$-\limdd with $G = \set{\id}$, i.e., each label is of the form $\lambda\id[n]$ with $\lambda\in\mathbb C$, as discussed in \autoref{sec:isomorphism-qmdd}.
The relations between coset, graph, stabilizer states and $G$-\limdd with $G=\gen X,\gen Z, \gen\pauli$ are proven in \autoref{thm:pauli-limdd-is-stabilizer} and \autoref{sec:graph-states-limdds}
(which also shows that poly-sized \limdd includes \qmddc).
\autoref{thm:limdd-superset-qmdd-plus-stabilizers} shows that there is family of a non-stabilizer states (with small \limdd) for which \qmdd is exponential, hence the separation between \qmddc.
\autoref{thm:graph-state-qmdd-lower-bound} shows the separation with \qmdds by demonstrating that
the so-called (two-dimensional) cluster state, requires $2^{\Omega(\sqrt{n})}$ nodes as \qmdd.
Finally, \autoref{sec:graph-states-limdds} proves the same for coset states.
\end{proof}

\autoref{thm:pauli-limdd-is-stabilizer} shows that any stabilizer state can be represented as a
$\gen\Pauli$-Tower \limdd (\autoref{def:tower}).

\begin{definition}\label{def:tower}
	A $n$-qubit $G$-Tower-\limdd, is a $G$-\limdd with exactly one node on each level. 
	Edges to nodes on level $k$ are labeled as follows:
	low edges are labeled with $\id^{\otimes k}$, high edges with $P\in G^{\otimes k} \cup \{0\}$
	and the root edge is labeled with $\lambda \cdot P$ with $P\in G^{\otimes k}$ and $\lambda \in \mathbb C \setminus \set 0$ (i.e., in contrast to high edges, the root edge can have an arbitrary scalar).
	\autoref{fig:tower} depicts a $n$-qubit $G$-Tower \limdd.
\end{definition}

\begin{restatable}{theorem}{thmpaulitower}%
	\label{thm:pauli-limdd-is-stabilizer}
	\label{thm:pauli-tower-limdds-are-stabilizer-states}
    Let $n>0$.
    Each $n$-qubit stabilizer state is represented up to normalization by a \gen\Pauli-Tower \limdds of \autoref{def:tower}, e.g., where the scalars $\lambda$ of the {\paulilim}s $\lambda P$ on high edges are restricted as $\lambda \in \{0, \pm 1, \pm i\}$.
	Conversely, every such \limdd represents a stabilizer state.
\end{restatable}

\begin{proof}[Proof sketch of \autoref{thm:pauli-limdd-is-stabilizer}] (Full proof in \autoref{sec:graph-states-limdds})
    The $n=1$ case: the six single-qubit states $\ket{0}, \ket{1}, \ket{0} \pm \ket{1}$ and $\ket{0} \pm i \ket{1}$ are all represented by a \gen\Pauli-Tower \limdd with a single node on top of the leaf.  The induction step: Let $\ket{\psi}$ be an $n$-qubit stabilizer state.
    First, consider the case that $\ket{\psi} = \ket{a}\ket{\psi'}$ where $\ket{a} = \alpha \ket{0} + \beta \ket{1}$ (with $\alpha, \beta \in \{0, \pm 1, \pm i\}$) and $\ket{\psi'}$ are stabilizer states on respectively $1$ and $n-1$ qubits.
    Then $\ket{\psi}$ is represented by the \gen\Pauli-Tower-\limdd~~\lnode{\alpha \id}{\psi'}{\beta \id}{\psi'}.
	In the remaining case, $\ket{\psi} = \frac{1}{\sqrt{2}}\left( \ket{0}\ket{\psi_0}+\ket{1}\ket{\psi_1} \right)$, where both $\ket{\psi_0}$ and $\ket{\psi_1}$ are stabilizer states.
    Moreover, since $\ket{\psi}$ is a stabilizer state, there is always a set of single-qubit Pauli gates $P_1,\ldots, P_n$ and a $\lambda \in \{\pm 1, \pm i\}$ such that $\ket{\psi_1}=\lambda P_n\otimes\cdots\otimes P_1\ket{\psi_0}$.
	That is, in our terminology, the states $\ket{\psi_0}$ and $\ket{\psi_1}$ are \emph{isomorphic}.
	Hence $\ket{\psi}$ can be written as
	\begin{align}
		\ket{\psi}=\frac{1}{\sqrt{2}} \left[\ket{0}\ket{\psi_0} + \lambda \ket{1}\otimes \left(P_n\otimes\cdots\otimes P_1\ket{\psi_0}\right)\right]
	\end{align}
    Hence $\ket{\psi}$ is represented by the Tower \Pauli-\limdd~$\lnode{\id}{\psi_0}{\lambda P_n \otimes \cdots\otimes P_1}{\psi_0}$.
    In both cases, $\ket{\psi'}$ is represented by a Tower \Pauli-\limdds (up to normalization) by the induction hypothesis.
\end{proof}

\begin{tabular}{p{10.3cm}p{4cm}}
	\begin{tikzpicture}[->,>=stealth',shorten >=1pt,auto,node distance=1cm,
        thick, state/.style={circle,draw,minimum size=0.5cm},font=\footnotesize, scale=0.3,
    		inner sep=0pt,]
    \node[state] (a1) {};
    \node[state, below =.7cm of a1] (a3) {$v$};
    \node[state, below =.7cm of a3] (a4) {};
    \node[draw,rectangle,minimum size=0.4cm, below= .5cm of a4] (w4) {1};

    \draw[<-] (a1) --++(90:3cm) node[pos=.7,lbl] {$\nicefrac 1{\sqrt 2} \cdot \id^{\otimes3}$} node[right,pos=.8] {$e_{GHZ}$};
    \draw[e0=25] (a1) edge  node[] {} (a3);
    \draw[e1=25] (a1) edge  node[lbl,right] {$X \otimes X$} (a3);
    \draw[e0=25] (a1) edge  node[] {} (a3);
    \draw[e1=25] (a3) edge  node[lbl,right] {$0$} (a4);
    \draw[e0=25] (a3) edge  node[] {} (a4);
    \draw[e0=25] (a4) edge  node[] {} (w4);
    \draw[e1=25] (a4) edge  node[lbl,right] {0} (w4);

    \node[state,right= 3cm of a1] (a1) {};
    \node[state, below =.7cm of a1] (a3) {};
    \node[state, below =.7cm of a3] (a4) {};
    \node[draw,rectangle,minimum size=0.4cm, below= .5cm of a4] (w4) {1};

    \draw[<-] (a1) --++(90:3cm) node[pos=.7,lbl] {$\nicefrac 1{\sqrt 8} \cdot \id^{\otimes3}$} node[right,pos=.8] {$e_{+++}$};
    \draw[e0=25] (a1) edge  node[] {} (a3);
    \draw[e1=25] (a1) edge  node[right] {} (a3);
    \draw[e0=25] (a1) edge  node[] {} (a3);
    \draw[e1=25] (a3) edge  node[right] {} (a4);
    \draw[e0=25] (a3) edge  node[] {} (a4);
    \draw[e0=25] (a4) edge  node[] {} (w4);
    \draw[e1=25] (a4) edge  node[right] {} (w4);

    \node[state,right= 3cm of a1] (a1) {};
    \node[state, below =.7cm of a1] (a3) {};
    \node[state, below =.7cm of a3] (a4) {};
    \node[draw,rectangle,minimum size=0.4cm, below= .5cm of a4] (w4) {1};

    \draw[<-] (a1) --++(90:3cm) node[pos=.7, lbl] 
    		{$\nicefrac 1{\sqrt 8} \cdot \id^{\otimes3}$} node[right,pos=.8] {$e_{\phi}$};
    \draw[e0=25] (a1) edge  node[] {} (a3);
    \draw[e1=25] (a1) edge  node[lbl,right] {$X \otimes X$} (a3);
    \draw[e0=25] (a1) edge  node[] {} (a3);
    \draw[e1=25] (a3) edge  node[lbl,right] {$Y$} (a4);
    \draw[e0=25] (a3) edge  node[] {} (a4);
    \draw[e0=25] (a4) edge  node[] {} (w4);
    \draw[e1=25] (a4) edge  node[lbl,right] {$-1$} (w4);
 
	\end{tikzpicture}
&
~\begin{tikzpicture}[
    scale=0.3,
    every path/.style={>=latex},
    every node/.style={},
    inner sep=0pt,
    minimum size=14pt,
    line width=1pt,
    node distance=.7cm,
    thick,
    font=\footnotesize
    ]

    \node[draw,circle] (a1) {};
    \node[draw,circle, below =of a1] (a2) {};
    \node[ below =-.12cm of a2] (a) {$\pmb{\vdots}$};
    \node[draw,circle, below =.5cm of a2] (a3) {};
    \node[draw,circle,rectangle,minimum size=0.4cm, below=of a3] (w1) {$1$};

    \draw[<-] (a1) --++(90:2cm) node[lbl,pos=1.4] {$\lambda \cdot L_n$};
    \draw[e0=25] (a1) edge  node[] {} (a2);
    \draw[e1=25] (a1) edge  node[lbl,right] {$L_{n-1}$} (a2);
    \draw[e0=25] (a3) edge  node[] {} (w1);
    \draw[e1=25] (a3) edge  node[lbl,right] {$L_0$} (w1);

    \node[left=.5cm of a1] (x) {Level $n$};
    \node[below= of  x,xshift=0cm] (x) {Level ${n-1}$};
    \node[below= of  x,xshift=0cm] (x) {Level $1$};
\end{tikzpicture}
\\
	\refstepcounter{figure}%
      \small Figure~\thefigure: Example \gen\pauli-Tower \limdds for three stabilizer states:
      		the GHZ state $\ket{e_{GHZ}} = \frac{1}{\sqrt{2}} \left(\ket{000} + \ket{111}\right)$,
      		for $\ket{e_{+++}} = \ket{+++}$ where $\ket{+} = \frac{1}{\sqrt{2}} \left(\ket{0} + \ket{1}\right)$, and 
      		the state $\ket{e_\phi} = \frac{1}{\sqrt{8}} \left( \ket{000} - \ket{001} + i \ket{010} + i\ket{011} + i \ket{100} + i \ket{101} - \ket{110} + \ket{111} \right)$ with stabilizer group generators \set{X\otimes X\otimes X, -Z\otimes Z\otimes X, \id \otimes Y\otimes Z}.
	    \label{fig:tower-examples}
& 
	\refstepcounter{figure}%
      \small Figure~\thefigure: An $n$-qubit $G$-Tower \limdd.
	We let $L_i\in G^{\otimes i}\cup\set 0$ and $\lambda \in \mathbb C\setminus \set 0$ (only root edges have an arbitrary scalar). 
	    \label{fig:tower}
\\
\end{tabular}

We stress that obtaining the LIMs for the Pauli Tower-\limdd of a stabilizer state is not immediate from the stabilizer generators; specifically, the edge labels in the Pauli-\limdd are not directly the stabilizers of the state.
For example, the $GHZ$ state $\frac 1{\sqrt 2} (\ket{000} + \ket{111})$ is represented by $\ket{e_{GHZ}} = \frac 1{\sqrt 2} \lnode{\id}{v}{X \otimes X}{v}$ with $\ket{v} = \ket{00}$ in \autoref{fig:tower-examples}, but $X\otimes X$ is not a stabilizer of $\ket{00}$.
Nonetheless, \autoref{thm:pauli-limdd-is-stabilizer} implicitly contains an algorithm that constructs a \gen{\pauli}-Tower \limdd stabilizer state. \autoref{sec:canonicity} also provides the inverse construction,
which we use to make \limdds (representing any quantum state) canonical in time $\oh(m n^3)$ (using \autoref{alg:canonical}).

We also note that \autoref{thm:pauli-limdd-is-stabilizer} demonstrates that for any $n$-qubit stabilizer state $\ket{\phi}$, the $(n-1)$-qubit states $(\bra{0} \otimes \id[2^{n-1}]) \ket{\phi}$ and $(\bra{1} \otimes \id[2^{n-1}]) \ket{\phi}$ are not only stabilizer states, but also \paulilim-isomorphic.
While we believe this fact is known in the community,%
\footnote{For instance, this fact can be observed (excluding global scalars) by executing the original algorithm for simulating single-qubit computational-basis measurement on the first qubit, as observed in~\cite{gottesman1998heisenberg}.
Similarly, the characterization  in \autoref{prop:separation} of \gen Z-Tower-\limdds as representing precisely the graph states, is immediate by defining graph states recursively (see \autoref{sec:graph-states-limdds}).
The fact that $\langle X\rangle$-Tower \limdds represent coset states is less evident and requires a separate proof, which we also give in \autoref{sec:graph-states-limdds}.
}
we have not found this statement written down explicitly in the literature.
More importantly for this work, to the best of our knowledge, the resulting recursive structure (which DDs are) has not yet been exploited in the context of classical simulation.

Next, \autoref{thm:graph-state-qmdd-lower-bound} shows the separation with \qmdds by demonstrating that
the so-called (two-dimensional) cluster state, requires $2^{\Omega(\sqrt{n})}$ nodes as \qmdd. \autoref{thm:limdd-superset-qmdd-plus-stabilizers} shows that a trivial combination with stabilizer formalism does not solve this issue.

\begin{restatable}{theorem}{thmgraphqmddlower}
	\label{thm:graph-state-qmdd-lower-bound}
	Denote by $\ket{G_n}$ the two-dimensional cluster state, defined as a graph state on the $n \times n$ lattice. Each \qmdd representing $\ket{G_n}$ has at least $2^{\floor{n/12}}$ nodes.
\end{restatable}
\begin{proof}[Proof sketch]
	Consider a partition of the vertices of the $n\times n$ lattice into two sets $S$ and $T$ of size $\frac{1}{2}n^2$, corresponding to the first $\frac{1}{2}n^2$ qubits under some variable order.
	Then there are at least $\lfloor n/3 \rfloor$ vertices in $S$ that are adjacent to a vertex in $T$ \cite[Th. 11]{lipton1979generalized}.
	Because the degree of the vertices is small, many vertices on this boundary are not connected and therefore influence the amplitude function independently of one another.
	From this independence, it follows that, for any variable order,
	the partial assignments $\vec a \in \set{0,1}^{\frac{1}{2}n^2}$ induce
	$2^{\lfloor n/12 \rfloor}$ different subfunctions $f_{\vec a}$, where
	$f\colon \set{0,1}^{n^2} \rightarrow \mathbb C$
	is the amplitude function of $\ket{G_n}$.
	The lemma follows by noting that a \qmdd has a single node per unique subfunction modulo phase.
	For details see \autoref{sec:graph-state-lower-bound}.
\end{proof}

\begin{corollary}[Exponential separation between Pauli-\limdd versus QMDD union stabilizer states]
	\label{thm:limdd-superset-qmdd-plus-stabilizers}
	There is a family of non-stabilizer states, which we call \concept{pseudo cluster states}, that  have
	polynomial-size Pauli-\limdd but exponential-size \qmdds  representation.
\end{corollary}
\begin{proof}
    Consider the pseudo cluster state $\ket{\phi} = \frac1{\sqrt{2}}(\ket{0} + e^{i\pi/4}\ket{1}) \otimes \ket{G_n}$ where $\ket{G_n}$ is the graph state on the $n\times n$ grid.
    This is not a stabilizer state, because each computational-basis coefficient of a stabilizer state is of the form $z\cdot \frac 1{\sqrt{2}^k}$ for $z\in \{\pm 1, \pm i\}$ and some integer $k\geq 1$ \cite{nest2005local}, while $\bra{1}\otimes \bra{0}^{\otimes n^2} \ket{\phi} = e^{i\pi/4} \cdot \left(\frac{1}{\sqrt{2}}\right)^{n^2+1}$ is not of this form.
    Its canonical \qmdd and Pauli-\limdd have root nodes $\lnode{1}{G_n}{e^{i\pi/4}}{G_n}$
    and $\lnode{\id}{G_n}{e^{i\pi/4} \id}{G_n}$, where the respective diagram for $G_n$ is exponentially large (\autoref{thm:graph-state-qmdd-lower-bound}) and polynomially small (\autoref{thm:pauli-limdd-is-stabilizer}).
\end{proof}

\subsubsection{\limdds are exponentially more succinct than matrix product states}
\label{sec:separation-with-mps}

\autoref{thm:mps} states that matrix product states (MPS) require large bond dimension for representing the two-dimensional cluster states, which follows directly from the well-known results that these states have large Schmidt rank.

\begin{restatable}{theorem}{thmmps}
\label{thm:mps}
    To represent the graph state on the $n \times n$ grid (the two-dimensional cluster state on $n^2$ qubits), an MPS requires bond dimension $2^{\Omega(n)}$.
\end{restatable}
\begin{proof}
	Van den Nest et al.~\cite{nest2007classical} consider spanning trees over the complete graph where each node corresponds to a qubit and define the Schmidt-rank width: the largest encountered base-$2$ logarithm of the Schmidt rank between the two connected components resulting from removing an edge from the spanning tree, minimized over all possible spanning trees.
    It then follows from the relation between bond dimension and Schmidt rank (see \autoref{sec:preliminaries}) that any quantum state with Schmidt-rank width $w$ requires bond dimension $2^{w}$ for representation by an MPS.
    Van den Nest et al. also showed that for graph states, the Schmidt-rank width equals the so-called rank width of the graph, which for $n \times n$ grid graphs was shown to equal $n-1$ by Jelinek~\cite{jelinek2010rankwidth}.
    This proves the theorem.
\end{proof}

In contrast, the Pauli-\limdd efficiently represents cluster states, and more generally all stabilizer states (\autoref{thm:pauli-limdd-is-stabilizer}).

\subsection{Pauli-\limdd manipulation algorithms for simulation of quantum computing}
\label{sec:quantum-simulation}

In this section, we give all algorithms that are necessary to simulate a quantum circuit with Pauli-\limdds (referred to simply as \limdd from now on).
We provide algorithms which update the \limdd after an arbitrary gate and after a single-qubit measurement in the computational basis.
In addition, we give efficient specialized algorithms for applying a Clifford gate to a stabilizer state (represented by a \gen\pauli-Tower \limdd) and computing a measurement outcome.
We also show that many (Clifford) gates can in fact be applied to an arbitrary state in polynomial time.
\autoref{tab:complexity} provides an overview of the \limdd algorithms and their 
complexities compared to \qmdds.

\begin{table}[b!]
    \caption{Worst-case complexity of currently \emph{best-known algorithms} for applying specific operations, in terms of the size of the input diagram size $m$
	(i.e., the number of nodes in the DD) and the number of qubits $n$.
   	Although addition (\textsc{Add}) of quantum states is not, strictly speaking, a quantum operation, we include it because it is a subroutine of gate application.
	Note that several of the \limdd algorithms invoke \textsc{MakeEdge}
    and therefore inherit its cubic complexity (as a factor).}
\label{tab:complexity}
\centering
\begin{tabular}{l|ll|l}
    \bf Operation $\backslash$ input:
                    & \bf \qmdd & \bf \limdd & \bf Section \\
\hline
    Single $\ket{0}/\ket{1}$-basis measurement &  $\oh(m)$ & $\oh(m)$ & \autoref{sec:measurement} \\
Single Pauli gate  & $\oh(m)$ & $\oh(1)$ & \autoref{sec:simple-gates} \\
Single Hadamard gate / \textsc{Add()}  &  $\oh(2^n)$ \footnote{The worst-case of \qmdds and \limdds is caused by the vector addition introduced by the Hadamard gate~\cite[Table~2, +BC, +SLDD]{fargier2014knowledge}. See \autoref{fig:explosion} for an example. \label{fn1}}  & $\oh(n^3 2^n)$ $^{\ref{fn1}}$ & \autoref{sec:simple-gates} \\
Clifford gate on stabilizer state & $\oh(2^n)$ & $\oh(n^4)$ & \autoref{sec:clifford-polytime} \\
Multi-qubit gate  &  $\oh(4^n)$  & $\oh(n^34^n)$    & \autoref{sec:applygate} \\
\makeedge & $\oh(1)$ & $\oh(n^3)$ & \autoref{sec:makeedge} \\
Checking state equality &  $\oh(1)$ & $\oh(n^3)$
& \autoref{sec:pauli-equivalence-checking} 
\end{tabular}
\end{table}

Central to the speed of many DD algorithms is keeping the diagram canonical throughout the computation.
Recall from \autoref{sec:isomorphism-qmdd}, that a $G$-\limdd can merge isomorphic nodes
$v \simeq_G w$, i.e., if there exists a $G$-LIM $C$ such that $\ket w = C \ket v$.
To achieve this, we require a `\makeedge' subroutine which, given the node \lnode[w]{A}{v_0}{B}{v_1}, returns \ledge{C}{v} with $C \ket v = \ket w$, where $v$ is the unique, canonical node in the diagram that is $G$-isomorphic to node $w$.
\autoref{sec:makeedge} provides a
$\mathcal O(n^3)$ \makeedge algorithm for \gen\pauli-\limdds satisfying this specification.
For now, the reader may assume the provisional implementation in \autoref{alg:make-edge-prov}, which does not yet merge LIM-isomorphic nodes and hence does not yield canonical diagrams.

In line with other existing efficient decision-diagram algorithms, we use dynamic programming in our algorithms to avoid traversing all paths  (possibly exponentially many) in the \limdd.
In this approach, the decision diagram is manipulated and queried using recursive algorithms,
which store intermediate results for each recursive call to avoid unnecessary recomputations.
For instance, \autoref{alg:canonical} makes any \limdd canonical using dynamic programming and the (real) $\oh(n^3)$ \makeedge algorithm from \autoref{sec:makeedge}.
It recursively traverses child nodes at \autoref{can:recurse},
reconstructing the diagram bottom up in the backtrack at \autoref{can:makeedge}.
By virtue of dynamic programming it visits each node only once:
The table  $\textsc{Canonical}\cache \colon \Node \to \Edge$ stores for each node its canonical counterpart as soon as it is computed at \autoref{can:makeedge}. 
The algorithm therefore runs in time $\mathcal O(n^3 m)$ where $m$ is the number of nodes in the original diagram.

\begin{algorithm}
	\begin{algorithmic}[1]
		\Procedure{\makeedge}{\Edge $\ledge {A}{v}$, \Edge $\ledge Bw$}
			\State  $u:=\lnode[u] AvBw$
			\State \Return \Edge $\ledge {\id^{\otimes k}}u$  \Comment{Where $k = \index(u)$}
		\EndProcedure
	\end{algorithmic}
	\caption{Provisionary algorithm \makeedge for creating a new node/edge. Given two edges representing states $A\ket v,B\ket w$, it returns an edge representing the state $\ket 0A\ket v + \ket 1B\ket w$.
	The real \makeedge algorithm (\autoref{sec:makeedge}) returns a canonical node, assuming $v,w$ are already canonical.
}
	\label{alg:make-edge-prov}
\end{algorithm}
\begin{algorithm}
	\begin{algorithmic}[1]
		\Procedure{MakeCanonical}{\Edge $\ledge {A}{v}$}
		\If{$v \notin \textsc{Canonical}\cache$}\label{can:cache1} \Comment{Compute result once for $v$ and store in cache:}
		\State $e_0,~ e_1 := \textsc{MakeCanonical}(\low v), \textsc{MakeCanonical}(\high v)$
		\label{can:recurse}
		\State $\textsc{Canonical}\cache[v] := \makeedge(e_0, e_1)$
		\label{can:makeedge}
		\EndIf
		\State \Return $A  \cdot \textsc{Canonical}\cache[v]$
							\Comment{Retrieve result from cache}
		\label{can:rebuild2}
		\EndProcedure
	\end{algorithmic}
	\caption{Make any \limdd canonical using \makeedge.}
	\label{alg:canonical}
\end{algorithm}

This recursive algorithmic structure that uses dynamic programming and reconstructs the diagram in the backtrack, is typical for all DD manipulation algorithms.
Note that constant-time cache lookups (using a hash table) therefore require the canonical nodes produced by \makeedge.
\limdds additionally require the addition of LIMs to the caches; 
\autoref{sec:applygate} shows how we do this.

Finally, in this section, we often decompose LIMS using  $A = \lambda  P_n \otimes P' $.
Here $\lambda \in \mathbb C$ is a non-zero scalar, $P'$ a Pauli string and $P_n\in \set{\mathbb I,X,Z, Y}=\pauli$.
Our algorithms will use the \textsf{Follow} procedure from \autoref{alg:follow} to easily navigate diagrams according to edge semantics. Provided with a bit string $x_n\dots x_1$, the procedure is the same as \textsc{ReadAmplitude}. If however fewer bits are supplied, it returns a \limdd root edge representing a subvector.
For instance, the subvector for $\ket{10}$ of the \limdd root edge $e_r$ in \autoref{fig:qmdd-isoqmdd-exposition} (d) is computed by taking $\ket{\follow {10}{e_r}} = \ket{\ledge{\frac 14 \id\otimes XZ}{\ell_2}} =\frac 14\cdot [-1, 1, -i, i]$.
So, we can specify it as $\ket{\follow{b}{e}} = (\bra{b} \otimes \id^{n - \ell}) \ket{e}$, i.e., select the $b$th block of size $2^{n - \ell}$ from the vector $\ket e$ (or rather, return a \limdd edge representing that block).

\begin{algorithm}
	\begin{algorithmic}[1]
		\Procedure{Follow}{\Edge $\ledge[e]{\lambda P_n \otimes \dots \otimes P_1}{v}$, $x_n, \dots, x_k \in \set{0,1}$ \textbf{with} $n = \index(v)$ and $k \geq 1$}
        \If{$k > n$} \Return $\ledge{\lambda P_n \otimes \dots \otimes P_1}{v} $
        \Comment{End of bit string}
        \EndIf
		\Stateh $\gamma \bra{y_n \dots y_k} := \bra{x_n \dots x_k} \lambda P_n \otimes \dots \otimes P_k$ \Comment{$\oh(n)$-computable LIM operation}
		\label{algline:amplitude:lim}
		\If{$y_n = 0$}\Comment{$y_n = 0$}
		\State \Return $\gamma \cdot \textsc{Follow}(\Low_v, y_{n-1}, \dots, y_k)$
		\Else\Comment{$y_n = 1$}
		\State \Return $\gamma \cdot \textsc{Follow}(\High_v, y_{n-1}, \dots, y_k)$
		\EndIf
		\EndProcedure
	\end{algorithmic}
	\caption{\textsc{Follow}: a generalization of \textsc{ReadAmplitude}, returning edges.}
	\label{alg:follow}
\end{algorithm}

\subsubsection{Performing a measurement in the computational basis}
\label{sec:measurement}

\begin{algorithm}[t!]
	\begin{algorithmic}[1]
		\Procedure{MeasurementProbability}{\Edge $e$}
		\State $s_0 := \textsc{SquaredNorm}(\follow 0e)$
		\State $s_1 := \textsc{SquaredNorm}(\follow 1e)$
		\State \Return $s_0/(s_0+s_1)$
		\EndProcedure
        \Procedure{SquaredNorm}{$\Edge \ledge{\lambda P}{v}$ \textbf{with} $\lambda \in \mathbb C, P\in \Pauli^{\index(v)}$}
		\If{$\index(v)=0$}
		\Return $|\lambda|^2$
		\EndIf
		\If{$v \notin \textsc{SNorm}\cache$}\label{sn:cache1} \Comment{Compute result once for $v$ and store in cache:}
		\State $\footnotesize\textsc{SNorm}\cache[v] :=  \textsc{SquaredNorm}(\follow 0{\ledge {\mathbb I}v}) + \textsc{SquaredNorm}(\follow 1{\ledge {\mathbb I}v})$\label{sn:cache2}
		\EndIf
		\State \Return $|\lambda|^2 \cdot \textsc{SNorm}\cache[v]$
							\Comment{Retrieve result for $v$ from cache and multiply with $|\lambda|^2$}
		\EndProcedure
        \Procedure{\project}{\Edge $\ledge[e]{\lambda P}{v}$, measurement outcome $m\in\set{0,1}$}
        \If{$m=0$}
	        \State $e_r := \makeedge(\follow 0e, ~~0 \cdot \follow 0e)$
        \Else
	        \State $e_r := \makeedge(0 \cdot \follow 0e,~~ \follow 1e)$
	     \label{l:project-diag}
         \EndIf
		\State \Return $1/\sqrt{\textsc{SquaredNorm}(e_r)} \cdot e_r$
				  \label{l:project-project}
		\EndProcedure
	\end{algorithmic}
    \caption{Algorithms \textsc{MeasurementProbability} and \textsc{\project} for respectively computing the probability of observing outcome $\ket 0$ when measuring the top qubit of a Pauli \limdd in the computational basis and converting the \limdd to the post-measurement state after outcome $m\in \{0, 1\}$.
		The subroutine \textsc{SquaredNorm} takes as input a Pauli \limdd edge $e$, and returns $\braket{e|e}$.
		It uses a cache to store the value $s$ of a node $v$.}
	\label{alg:measurement-top-qubit}
\end{algorithm}

We discuss algorithms for measuring, sampling and updating after measurement of the top qubit.
\autoref{sec:advanced-algorithms} gives general algorithms with the same worst-case runtimes.

The procedure \textsc{MeasurementProbability} in \autoref{alg:measurement-top-qubit} computes the probability $p$ of observing the outcome $\ket 0$ for state $\ket e$.
If the quantum state can be written as $\ket e=\ket 0\ket{e_0} + \ket 1\ket{e_1}$, then the probability is $p=\braket{e_0|e_0}/\braket{e|e}$, where we have $\braket{e|e}=\braket{e_0|e_0}+\braket{e_1|e_1}$.
Hence we compute the squared norms of $e_x = \follow xe$ using the \textsc{SquaredNorm} subroutine. 
The total runtime is dominated by the subroutine \textsc{SquaredNorm}, which computes the quantity $\braket{e|e}$ given a \limdd edge $e=\ledge {\lambda P}v$ by traversing the entire \limdd.
We have $\braket{e|e}=|\lambda|^2\bra{v}P^\dagger P\ket{v}=|\lambda|^2\braket{v|v}$, because $P^{\dagger}P=\mathbb I$ for Pauli matrices.
Therefore, to this end, it computes the squared norm of $\ket{v}$.
Since $\braket{v|v}=\braket{\Low_v|\Low_v}+\braket{\High_v|\High_v}$, this is accomplished by recursively computing the squared norm of the node's low and high edges.
This subroutine visits each node at most once by virtue of dynamic programming, which stores
intermediate results in a cache 
$\textsc{SNorm}\cache \colon \Node \to \mathbb R$
 for all recursive calls (\autoref{sn:cache1}, \ref{sn:cache2}).
Therefore, it runs in time $\oh(m)$ for a diagram with $m$ nodes.

The outcome $m\in \{0,1\}$ can then be chosen by flipping a $p$-biased coin.
The corresponding state update is implemented by the procedure \project.
In order to update the state $\ket{e}=\ket 0\ket{e_0} + \ket 1\ket{e_1}$ after the top qubit is measured to be $m$, we simply construct an edge $\ket{m}\ket{e_m}$ using the \makeedge subroutine. This state is finally normalized by multiplying (the scalar on) the resulting root edge with a normalization constant computed using squared norm.

To sample from a quantum state in the computational basis, simply repeat the measurement procedure for edge \ledge{}{v} with $k=\index(v)$, throw a $p$-biased coin to determine $x_k$, use $\follow{x_k}{\ledge{}{v}}$ to go to level $k-1$ and repeat the process.

\subsubsection{Gates with simple \limdd algorithms}
\label{sec:simple-gates}
As a warm up, before we give the algorithm for arbitrary gates and Clifford gates,
we first give algorithms for several gates that have
a relatively simple and efficient \limdd manipulation operation.
In the case of a controlled gate, we distinguish two cases, depending whether the control or the target qubit comes first; we call these a \emph{downward} and an \emph{upward} controlled gate, respectively.

Here, we let \nobreak{$L_k$} denote the unitary applying local gate $L$ on qubit $k$, i.e.,
$L_k \defn\id^{\otimes n-k} \otimes  L \otimes \id^{\otimes k-1}$.

\paragraph{Applying a \textbf{single-qubit Pauli gate} $Q$} to qubit $k$ of a \limdd, by updating the diagram's root edge from $A$ to $Q_k A$, i.e., change 
$A =\lambda P_n\otimes\cdots\otimes P_1$ to $\lambda P_n\otimes\cdots\otimes P_{k+1}\otimes QP_k\otimes P_{k-1}\otimes\cdots\otimes P_1$.
		Since only nodes ---and not root edges--- need be canonical, this can be done in constant time, provided that the \limdd is stored in the natural way (uncompressed with objects and pointers).

\paragraph{Applying any diagonal or antidiagonal single-qubit gate} to the top qubit can be done efficiently, e.g., applying the \textbf{$T$-gate to the top qubit}.
	For root edge $e=\ledge{\rootlim}{v}$, we can construct $e_x = \follow xe$,
		which propagates the root edge's LIM to the root's two children.
	Then, for a diagonal node $\diagg \alpha\beta$, we construct a new root node
	$\makeedge(\alpha\cdot e_0, \beta \cdot e_1)$. For the anti-diagonal gate $\antii \beta\alpha$, it is sufficient to note that $\antii \beta\alpha=X\cdot \diagg \alpha\beta$; thus, we can first apply a diagonal gate, and then an $X$ gate, as described above.

\begin{wrapfigure}[6]{r}{2.7cm}
\vspace{-.5em}
\begin{tikzpicture}[->,>=stealth',shorten >=1pt,auto,node distance=1cm,
        thick, state/.style={circle,draw,minimum size=14pt},font=\footnotesize]
        
    \node[state](r) {$v$};
    \node[state](1a)[below = .7cm of r, xshift=-.29cm]{$v_0$};
    \node[state](1b)[below = .7cm of r, xshift= .29cm]{$v_1$};
    
    \node[above = .6cm of r] (x) {};

	\node[right = .2cm of r,yshift=-.0cm,rotate=-0] {$\rightsquigarrow$}; 
    
    \path[]
    (x)  edge         node[lbl,xshift=-.4cm,pos=.1] {$\rootlim$} (r)
    (r)  edge[e1]     node[right, pos=.3,lbl] {$B$} (1b)
    (r)  edge[e0]     node[left,pos=.3,lbl] {$A$} (1a)
    ;

    \node[state,right = .8cm  of r, inner sep = 0pt](r) {$v'$};
    \node[state](1a)[below = .7cm of r, xshift=-.29cm]{$v_0$};
    \node[state](1b)[below = .7cm of r, xshift= .29cm]{$v_1$};
    
    \node[above = .6cm of r] (x) {};
    
    \path[]
    (x)  edge         node[lbl,pos=.18,xshift=-.7cm] {$S_k \rootlim S_k^\dagger$} (r)
    (r)  edge[e0]     node[left, pos=.3,lbl] {$A$} (1a)
    (r)  edge[e1]     node[right,pos=.3,lbl] {$iB$} (1b)
    ;
    \end{tikzpicture}
\end{wrapfigure}
\paragraph{Applying a \textbf{phase gate}} ($S=\begin{smallmat}1& 0\\ 0& i\end{smallmat}$) to qubit with index $k$ on $\ledge{\rootlim}{v}$ is also efficient. \autoref{alg:S-gate} gives a recursive procedure.
            If $k < n = \index(v)$ (top qubit), then note $S_k \rootlim \ket v = (S_k \rootlim S_k^{\dagger}) S_k \ket{v}$ where $S_k \rootlim S_k^{\dagger}$ is the new ($\oh(n)$-computable) root \gen\pauli-LIM because $S_k$ is a Clifford gate.
            Hence, we can `push' $S_k$ through the LIMs down the recursion,
            rebuilding the \limdd in the backtrack with \makeedge on \autoref{sg:rebuild} and \ref{sg:rebuild2}.
            To apply $S_k$ to $v$ when $k = n = \index(v)$, we finally multiply the high edge label with~$i$ on \autoref{sg:sk}.
            Dynamic programming, using table \textsc{SGate}\cache, ensures a linear amount of recursive calls in the number of nodes $m$.
            The total runtime is therefore $\oh(m n^3)$, as \makeedge's is cubic (see \autoref{sec:canonicity}).

\begin{algorithm}[h!]
	\begin{algorithmic}[1]
        \Procedure{SGate}{$\Edge~ \ledge{A}{v}$ \textbf{with} $A\in \pauli$-LIM,
        	$k \in \set{1,\dots, \index(v)}$}
		\If{$v \notin \textsc{SGate}\cache$}\label{sg:cache1} \Comment{Compute result once for $v$ and store in cache:}
		\If{$\index(v)=k$}
		\State $\textsc{SGate}\cache[v] := \makeedge(\Low_v, i \cdot \High_v)$
		\label{sg:sk}
		\Else
		\State $\textsc{SGate}\cache[v] := \makeedge(\textsf{SGate}(\Low_v,k), \textsf{SGate}(\High_v,k))$
		\label{sg:rebuild}
		\EndIf
		\EndIf
		\State \Return $S_k A S_k^\dagger  \cdot \textsc{SGate}\cache[v]$
							\Comment{Retrieve result from cache}
		\label{sg:rebuild2}
		\EndProcedure
	\end{algorithmic}
    \caption{Apply gate $S$ to qubit $k$ for \pauli-\limdd \ledge Av. We let $n = \index(v)$.}
	\label{alg:S-gate}
\end{algorithm}

\begin{wrapfigure}[6]{r}{2.7cm}
\vspace{-2em}
\begin{tikzpicture}[->,>=stealth',shorten >=1pt,auto,node distance=1cm,
        thick, state/.style={circle,draw,minimum size=14pt},font=\footnotesize]
        
    \node[state](r) {$v$};
    \node[state](1a)[below = .7cm of r, xshift=-.29cm]{$v_0$};
    \node[state](1b)[below = .7cm of r, xshift= .29cm]{$v_1$};    

    \node[above = .6cm of r] (x) {};
    
    \node[above = .1cm of r,xshift=.2cm] (l) {$\index(v) = c:$};
        
	\node[right = .2cm of r,yshift=-.0cm,rotate=-0] {$\rightsquigarrow$}; 
    
    \path[]
    (r)  edge[e1]     node[right, pos=.3,lbl] {$B$} (1b)
    (r)  edge[e0]     node[left,pos=.3,lbl] {$A$} (1a)
    ;

    \node[state,right = .8cm  of r, inner sep = 0pt](r) {$v'$};
    \node[state](1a)[below = .7cm of r, xshift=-.29cm]{$v_0$};
    \node[state](1b)[below = .7cm of r, xshift= .29cm]{$v_1$};
    
    \node[above = .6cm of r] (x) {};
    
    \path[]
    (r)  edge[e0]     node[left, pos=.3,lbl] {$A$} (1a)
    (r)  edge[e1]     node[right,pos=.3,lbl] {$Q_tB$} (1b)
    ;
    \end{tikzpicture}
\end{wrapfigure}
\paragraph{Applying a \textbf{Downward Controlled-Pauli gate}} $CQ_{t}^c$, where $Q$ is a single-qubit Pauli gate, $c$ the control qubit and $t$ the target qubit with $t<c$, to a node $v$ can also be done recursively.
        If $\index(v)>c$, then since $CQ_t^c$ is a Clifford gate, we may push it through the node's root label, and apply it to the children $\low v$ and $\high v$, similar to the $S$ gate.
        Otherwise, if $\index(v)=c$, then update $v$'s high edge label as $B \mapsto Q_t B$, and do not recurse.
\autoref{alg:cpgate} shows the recursive procedure, which is similar to
\autoref{alg:S-gate} and also has $\oh(m n^3)$ runtime.

\begin{algorithm}[h!]
	\begin{algorithmic}[1]
        \Procedure{CPauliGate}{$\Edge~ \ledge{A}{v}$ \textbf{with} $A\in \pauli$-LIM,
        	$c$, $t$ \textbf{with} $1 \leq c < t \leq n$}
		\If{$v \notin \textsc{CPauli}\cache$}\label{sg:cache1} \Comment{Compute result once for $v$ and store in cache:}
		\If{$\index(v)=k$}
		\State $\textsc{CPauli}\cache[v] := \makeedge(\Low_v, X_t \cdot \High_v)$
		\label{cp:cp}
		\Else
		\State $\textsc{CPauli}\cache[v] := \makeedge(\textsf{CPauliGate}(\Low_v,c ,t), \textsf{CPauliGate}(\High_v,c ,t))$
		\label{cp:rebuild}
		\EndIf
		\EndIf
		\State \Return $CX^c_t \cdot A \cdot {CX^c_t}^\dagger  \cdot \textsc{CPauli}\cache[v]$
							\Comment{Retrieve result from cache}
		\label{cp:rebuild2}
		\EndProcedure
	\end{algorithmic}
    \caption{Apply gate $CX$  with control qubit $c$ and target qubit $t$ for \pauli-\limdd \ledge Av. We let $n = \index(v)$. We can replace $CX$, with $CY, CZ$. modifying \autoref{cp:cp} accordingly (i.e. to $Y_t, Z_t$).}
	\label{alg:cpgate}
\end{algorithm}

\autoref{sec:clifford-polytime} shows that all Clifford gates (including Hadamard and upward CNOT) have runtime  $\oh(n^4)$ when applied to a stabilizer state represented as a \limdd. We first show how to apply general gates, in \autoref{sec:applygate}, as this yields some machinery required for Hadamards
(specifically, a pointwise addition operation).

\subsubsection{Applying a generic multi-qubit gate to a state}
\label{sec:applygate}

We use a standard approach \cite{fujita1997multi} to represent quantum gates ($2^n\times 2^n$ unitary matrices) as \limdds.
Here a matrix $U$ is interpreted as a function $u(r_1,c_1,\ldots, r_n,c_n)\defn \bra rU\ket c$ on $2n$ variables, which returns the entry of $U$ on row $r$ and column $c$.
The function $u$ is then represented using a \limdd of $2n$ levels.
The bits of $r$ and $c$ are interleaved to facilitate recursive descent on the structure.
In particular, for $x,y\in \{0, 1\}$, the subfunction $u_{xy}$ represents a quadrant of the matrix, namely the submatrix
$u_{xy}(r_2, c_2,  \dots, r_n, c_n) \defn u(x, y, r_2, c_2,  \dots, r_n, c_n) $, as follows:
\begin{align}
u=\overbrace{
\left.
\begin{bmatrix}
\marktopleft{aa} u_{00} & \marktopleftb{bb} u_{01} \markbottomright{aa}  \\
~u_{10} & u_{11}\markbottomrightb{bb} \\
\end{bmatrix}
\right\rbrace
}^{u_{0*}}
u_{* 1}
\end{align}
\autoref{def:limdd-as-matrix} formalizes this idea.
\autoref{fig:gates-examples} shows a few examples of gates represented as \limdds.

\begin{definition}[\limdd{}s for gates]
	\label{def:limdd-as-matrix}
A \limdd edge $e=\ledge Au$ can represent a (unitary) $\nobreak{2^n\times 2^n}$ matrix $U$
iff $\index(u)=2n$.
The value of the matrix cell $U_{r,c}$ is defined as $\follow{r_1 c_1 r_2 c_2 \dots r_n c_n}{\ledge Au}$
where $r,c$ are the row and column index, respectively, with binary representation $r_1,\dots,r_n$ and $c_1,\dots,c_n$.
The semantics of a \limdd edge $e$ as a matrix is denoted $[e]\defn U$ (as opposed to its semantics $\ket{e}$ as a vector).
\end{definition}

\vspace{-1em}
\paragraph{The procedure \textsc{ApplyGate}}
 (\autoref{alg:apply-gate-limdd-limdd}) applies a gate $U$ to a state $\ket{\phi}$, represented by \limdds $e_U$ and~$e_\phi$.
It outputs a \limdd edge representing $U\ket{\phi}$.
It works similar to well-known matrix-vector product algorithms for decision diagrams \cite{fujita1997multi,miller2006qmdd}, except that we also handle edge weights with LIMs (see \autoref{fig:apply-gate} for an illustration).
Using the $\follow{x}{e}$ procedure, we write $\ket{\phi}$ and $U$ as
\begin{align}
	\ket{\phi} = & \ket{0}\ket{\phi_0}+\ket{1}\ket{\phi_1} \\
	U = & \ket{0}\bra{0}\otimes U_{00} + \ket{0}\bra{1}\otimes U_{01} + \ket{1}\bra{0}\otimes U_{10} + \ket{1}\bra{1}\otimes U_{11}
\end{align}

\begin{wrapfigure}[14]{r}{7cm}
\vspace{-2.4em}
\begin{tikzpicture}[
    scale=0.3,
    every path/.style={>=latex},
    every node/.style={},
    inner sep=0pt,
    minimum size=14pt,
    line width=1pt,
    node distance=1cm,
    thick,
    font=\footnotesize
    ]

    \node[draw,circle] (a1) {};
    \node[draw,circle, below =of a1] (a3) {};
    \node[draw,circle,rectangle,minimum size=0.4cm, below=of a3] (w1) {1};

    \draw[<-] (a1) --++(90:2cm) node[pos=1.4] {};
    \draw[e0=25] (a1) edge  node[] {} (a3);
    \draw[e1=25] (a1) edge  node[lbl,right] {$X$} (a3);
    \draw[e0=25] (a3) edge  node[] {} (w1);
    \draw[e1=25] (a3) edge  node[lbl,right] {$0$} (w1);

    \node[draw,circle,right= 1cm of a1] (a1) {};
    \node[draw,circle, below =of a1] (a3) {};
    \node[draw,circle,rectangle,minimum size=0.4cm, below=of a3] (w2) {1};

    \draw[<-] (a1) --++(90:2cm) node[pos=1.4] {};
    \draw[e0=25] (a1) edge  node[] {} (a3);
    \draw[e1=25] (a1) edge  node[lbl,right] {$-X$} (a3);
    \draw[e0=25] (a3) edge  node[] {} (w2);
    \draw[e1=25] (a3) edge  node[lbl,right] {$0$} (w2);

    \node[draw,circle,right= 1cm of a1] (a1) {};
    \node[draw,circle, below =of a1] (a3) {};
    \node[draw,circle,rectangle,minimum size=0.4cm, below=of a3] (w3) {1};

    \draw[<-] (a1) --++(90:2cm) node[pos=1.4] {};
    \draw[e0=25] (a1) edge  node[] {} (a3);
    \draw[e1=25] (a1) edge  node[lbl,right] {$Z$} (a3);
    \draw[e0=25] (a3) edge  node[] {} (w3);
    \draw[e1=25] (a3) edge  node[right] {} (w3);

    \node[draw,circle,right= 1cm of a1] (a1) {};
    \node[draw,circle, below =.5cm of a1] (a3) {};
    \node[draw,circle, below =.5cm of a3] (a4) {};
    \node[draw,circle, below =.5cm of a4] (a5) {};
    \node[draw,circle,rectangle,minimum size=0.4cm, below= .5cm of a5] (w4) {1};

    \draw[<-] (a1) --++(90:2cm) node[pos=1.4] {};
    \draw[e0=25] (a1) edge  node[] {} (a3);
    \draw[e1=25] (a1) edge  node[lbl,right] {$X\otimes \id \otimes X$} (a3);
    \draw[e0=25] (a1) edge  node[] {} (a3);
    \draw[e1=25] (a3) edge  node[lbl,right] {$0$} (a4);
    \draw[e0=25] (a3) edge  node[] {} (a4);
    \draw[e1=25] (a4) edge  node[lbl,right] {$X$} (a5);
    \draw[e0=25] (a4) edge  node[] {} (a5);
    \draw[e0=25] (a5) edge  node[] {} (w4);
    \draw[e1=25] (a5) edge  node[lbl,right] {0} (w4);
    
    \node[below = .1cm of w1] {\id gate};
    \node[below = .1cm of w2] {$Z$ gate};
    \node[below = .1cm of w3] {$H$ gate};
    \node[below = -.4cm of w4, xshift=-1.2cm] {CNOT gate:};
\end{tikzpicture}
    \vspace{-1ex}
	\caption{\limdds representing various gates.}
	\label{fig:gates-examples} 
    \vspace{-1em}
\end{wrapfigure}
Then, on \autoref{algline:ag:apply-gate-compute-term}, we compute each of the four terms $U_{rc}\ket{\phi_c}$ for row/column bits $r,c \in \{0,1\}$.
We do this by constructing four \limdds $f_{r,c}$ representing the states
$\ket{f_{r,c}}=U_{r,c}\ket{\phi_c}$, using four recursive calls to the \textsc{ApplyGate} algorithm.
Next, on \autoref{algline:ag:apply-gate-add-0} and \ref{algline:ag:apply-gate-add-1}, the appropriate states are added, using \textsc{Add} (\autoref{alg:add-limdds}), producing \limdds $e_0$ and $e_1$ for the states $\ket{e_0}=U_{00}\ket{\phi_0}+U_{01}\ket{\phi_1}$ and for $\ket{e_1}=U_{10}\ket{\phi_0}+U_{11}\ket{\phi_1}$.
The base case of \textsc{ApplyGate} is the case where $n=0$, which means $U$ and $\ket{v}$ are simply scalars, in which case both $e_U$ and $e_\phi$ are edges that point to the leaf.

\begin{algorithm}[h!]
	\caption{Applies the gate $[e_U]$ to the state $\ket{e_{\phi}}$. Here $e_U$ and $e_{\phi}$ are \limdd edges. 
		The output is a \limdd edge $\psi$ satisfying $\ket{\psi}=[e_U]\ket{e_\phi}$.}
	\label{alg:apply-gate-limdd-limdd}
	\label{alg:apply-gate-limdd-limdd-cache}
	\begin{algorithmic}[1]
	\Procedure{ApplyGate}{\Edge $\ledge[e_U=]{\lambda P}u$, \Edge $\ledge[e_\phi=]{\gamma Q}v$
		 \textbf{with} $\index(u) = 2\cdot\index(v)$}
		\If{$\index(v)=0$}  \Return $\leafedge{\lambda\cdot \gamma}$ \Comment{$P = Q = 1$}
        \EndIf
        \State $P', Q':=\rootlabel(\ledge Pu),\rootlabel(\ledge Qv)$ \Comment{Get canonical root labels}
				\label{algline:ag:apply-gate-root-label}
		\If{$({P'}, u, {Q'}, v)\notin \textsc{Apply-cache}$} 
			\Comment{Compute result for the first time:}
		\For{$r,c\in\{0,1\}$} 
		\State \Edge $f_{r,c}:=\textsc{ApplyGate}(\follow{rc}{\ledge {P'}u},\follow{c}{\ledge{Q'}v})$
			\label{algline:ag:apply-gate-compute-term}
		\EndFor
		\Stateh \Edge $e_0:=\textsc{Add}(f_{0,0}, f_{0,1})$
			\label{algline:ag:apply-gate-add-0}
		\State \Edge $e_1:=\textsc{Add}(f_{1,0}, f_{1,1})$
			\label{algline:ag:apply-gate-add-1}
		\State $\textsc{ApplyCache}[({P'},{u},{Q'},{v})]:=  \makeedge(e_0,e_1)$\Comment{Store in cache}
			\label{algline:ag:apply-gate-store-cache1}
	\EndIf
	\Stateh $e_\psi' := \textsc{Apply-cache}[(P', u, Q', v)]$\Comment{Retrieve from cache}
	\State\Return $\lambda \gamma \cdot e_\psi'$
		\label{algline:ag:apply-gate-retrieve-cache}
			\label{algline:ag:apply-gate-return}
		\EndProcedure
	\end{algorithmic}
\end{algorithm}

\paragraph{Caching in ApplyGate.}
A straightforward way to implement dynamic programming would be to simply store all results of \textsc{ApplyGate} in the cache, i.e., when $\textsc{ApplyGate}(\ledge {\lambda P}u, \ledge {\gamma Q}v)$ is called, store an entry with key $(P,u,Q,v)$ in the cache.
This would allow us to retrieve the result the next time \textsc{ApplyGate} is called with the same parameters.
However, we can do much better, in such a way that we can retrieve the result from the cache also when the procedure is called with parameters $\textsc{ApplyGate}(\ledge Ax, \ledge By)$ satisfying $[\ledge {\lambda P}u]=[\ledge Ax]$ and $\ket{\ledge {\gamma Q}v}=\ket{\ledge By}$.
This can happen even when $\lambda P\ne A$ or $\gamma Q\ne B$; therefore this may prevent many recursive calls.

To this end, we store not just an edge-edge tuple from the procedure's parameters, but a \emph{canonical} edge-edge tuple.
To obtain canonical edge labels, our algorithms use the function \rootlabel which returns a \emph{canonically chosen} LIM, i.e., it holds that $\rootlabel(\ledge Av)=\rootlabel(\ledge Bv)$ whenever $A\ket v=B\ket v$.
A specific choice for \rootlabel is the lexicographic minimum of all possible root labels.
In \autoref{alg:lexmin}, we give an $O(n^3)$-time algorithm for computing the lexicographically minimal root label, following the same strategy as the \makeedge procedure in \autoref{sec:makeedge}.
As a last optimization, we opt to not store the scalars $\lambda, \gamma$ in the cache (they are ``factored out''), so that we can retrieve this result also when \textsc{ApplyGate} is called with inputs that are equal up to a complex phase.
These scalars are then factored back in on \autoref{algline:ag:apply-gate-retrieve-cache} and \ref{algline:ag:apply-gate-store-cache1}.

\begin{figure}
	\centering
	\begin{tikzpicture}[
    scale=0.3,
    every path/.style={>=latex},
    every node/.style={},
    inner sep=0pt,
    node distance=.5cm,
    minimum size=14pt,
    line width=1pt,
    thick,
    font=\footnotesize
    ]

    \node[draw,circle] (r) {$u$};
    \node[draw,circle,below = of r,xshift=-.65cm] (a1) {};
    \node[draw,circle,below = of r,xshift= .65cm] (a2) {};
  
    \node[draw,circle,below = of a1,xshift=-.3cm] (a11) {$u_{00}$};
    \node[draw,circle,below = of a1,xshift= .3cm] (a12) {$u_{01}$};
    \node[draw,circle,below = of a2,xshift=-.3cm] (a21) {$u_{10}$};
    \node[draw,circle,below = of a2,xshift= .3cm] (a22) {$u_{11}$};

    \draw[<-] (r) --++(-4cm,0) node[lbl,pos=.5] {$\lambda A$} node[pos=1.2] {$e_U$};
    
    \draw[e0] (r) edge  node[left]  {} (a1);
    \draw[e1] (r) edge  node[right] {} (a2);
    \draw[e0] (a1) edge  node[left]  {} (a11);
    \draw[e1] (a1) edge  node[right] {} (a12);
    \draw[e0] (a2) edge  node[left]  {} (a21);
    \draw[e1] (a2) edge  node[right] {} (a22);

\draw [-, decorate,
    decoration = {calligraphic brace,mirror},ultra thick] ($(a11.south west) + (270:.4)$) --  ($(a22.south east) + (270:.4)$) node[pos=.5,below=.1cm] {$U$} 
    node[below=.1cm,xshift=.5cm] {$\times$};

    \node[draw,circle,right= 2.5cm of r] (r) {$v$};
    \node[below = of r,xshift=-.65cm] (a1h) {};
    \node[below = of r,xshift= .65cm] (a2h) {};
  
    \node[draw,circle,below = of a1h,xshift= .3cm] (a1) {$v_{0}$};
    \node[draw,circle,below = of a2h,xshift=-.3cm] (a2) {$v_{1}$};

    \draw[<-] (r) --++(-4cm,0) node[lbl,pos=.5] {$\gamma B$} node[pos=1.2] {$e_\phi$};
    
    \draw[e0] (r) edge  node[left]  {} (a1);
    \draw[e1] (r) edge  node[right] {} (a2);

\draw [-, decorate,
    decoration = {calligraphic brace,mirror},ultra thick] ($(a1.south west) + (270:.4)$) --  ($(a2.south east) + (270:.4)$) node[pos=.5,below=.1cm] {$\ket \phi$}
        node[below=.1cm,xshift=.4cm] {$=$};

    \node[draw,circle,right= 2.5cm of r] (r) {};
    \node[below = of r,xshift=-.65cm] (a1h) {};
    \node[below = of r,xshift= .65cm] (a2h) {};
    \node[draw, dotted, below = .6cm of r,xshift=-.8cm,rounded corners,minimum width=1.4cm,minimum height=1.4cm] (a1hh) {};
    \node[draw, dotted, below = .6cm of r,xshift= .8cm,rounded corners,minimum width=1.4cm,minimum height=1.4cm] (a2hh) {};

    \node[draw,circle,below = 1.4cm of r,xshift=-1.2cm] (a11) {};
    \draw[<-] (a11) --++(-.cm,2)  node[pos=1.6] {$f_{0,0}$} node[pos=1.6,xshift=.4cm] {$+$};
    \node[draw,circle,below = 1.4cm of r,xshift=-.4cm] (a12) {};
    \draw[<-] (a12) --++(-.cm,2)  node[pos=1.6] {$f_{0,1}$};
    \node[draw,circle,below = 1.4cm of r,xshift= .4cm] (a21) {};
    \draw[<-] (a21) --++(-.cm,2)  node[pos=1.6] {$f_{1,0}$} node[pos=1.6,xshift=.4cm] {$+$};
    \node[draw,circle,below = 1.4cm of r,xshift= 1.2cm] (a22) {};
    \draw[<-] (a22) --++(-.cm,2)  node[pos=1.6] {$f_{1,1}$};

    \draw[<-] (r) --++(-4cm,0) node[pos=.7] {} node[pos=1.2] {$e_\psi'$};
    
    \draw[e0] (r) edge  node[left,lbl]  {$e_0$} (a1hh.north);
    \draw[e1] (r) edge  node[right,lbl] {$e_1$} (a2hh.north);

\draw [-, decorate,
    decoration = {calligraphic brace,mirror},ultra thick] ($(a1hh.south west) + (270:.4)$) --  ($(a2hh.south east) + (270:.4)$) node[pos=.5,below=.1cm] {$U\ket \phi = \lambda \gamma \cdot \ket{e_\psi '}  = \ket {e_\psi} $}
    	;
 
\end{tikzpicture}
    \caption{An illustration of \textsc{ApplyGate} (\autoref{alg:apply-gate-limdd-limdd}), where matrix $U$ is applied to state $B\ket{v}$, both represented as Pauli-\limdds.
		The edges $f_{0,0}$, $f_{0,1}$, etc. are the edges made on \autoref{algline:ag:apply-gate-compute-term}.
		The dotted box indicates that these states are added, using \textsc{Add}, producing edges $e_0,e_1$, which are then passed to \makeedge, producing the result edge.
		For readability, not all edge labels are shown.
}
	\label{fig:apply-gate}
\end{figure}

\paragraph{The subroutine \textsc{Add}} (\autoref{alg:add-limdds}) adds two quantum states, i.e., given two \limdds representing $\ket{e}$ and $\ket{f}$, it returns a \limdd representing $\ket e + \ket f$.
It proceeds by simple recursive descent on the children of $e$ and $f$.
The base case is when both edges point to the diagram's leaf.
In this case, these edges are labeled with scalars $A,B\in\mathbb C$, so we return the edge $\ledge{A+B}{1}$.

\begin{algorithm}[h!]
	\begin{algorithmic}[1]
		\Procedure{Add}{\Edge $\ledge[e=] Av$, \Edge $\ledge[f=] Bw$
		 \textbf{with} $\index(v) = \index(w)$}
		\If{$\index(v)=0$}
			 \Return $\leafedge{A+B}$ \Comment{$A,B \in \mathbb C$}
		\EndIf
			\If{$v\not\beforeq w$} \Return $\textsc{Add}(\ledge Bw, \ledge Av)$\Comment{Normalize for cache lookup}
			\EndIf
			\label{algline:add:add-swap-cache-lookup}
			\State $C:=\rootlabel(\ledge {A^{-1}B}w)$
			\label{algline:add:add-cache-factor-C}
			\If{$(v, C, w)\notin \textsc{Add-Cache} $}\Comment{Compute result for the first time:}
			\label{algline:add:lookup-cache-factor-C}
		\Stateh \Edge $a_0:=\textsc{Add}(\follow 0{\ledge{}{v}}, \follow 0{\ledge{C}{w}})$
			\label{algline:add:add-0}
		\State \Edge $a_1:=\textsc{Add}(\follow 1{\ledge{}{v}}, \follow 1{\ledge{C}{w}})$
			\label{algline:add:add-1}
		\State $\textsc{Add-Cache}[(v, C, w)] :=  \makeedge(a_0, a_1)$\Comment{Store in cache}
			\label{algline:add:add-makeedge}
			\label{algline:add:add-store}
		\EndIf
		\State\Return $A\cdot \textsc{Add-Cache}[(v,C,w)]$\Comment{Retrieve from cache}
		\EndProcedure
	\end{algorithmic}
	\caption{Given two $n$-LIMDD edges $e,f$, constructs a new LIMDD edge $a$ with $\ket{a}=\ket{e}+\ket{f}$.
	}
	\label{alg:add-limdds}
	\label{alg:add-limdds-cache}
\end{algorithm}

\paragraph{Caching in Add.}
A straightforward way to implement the cache would be to store a tuple with key $(A,v,B,w)$ in the call $\textsc{Add}(\ledge Av,\ledge Bw)$.
However, we can do much better; namely,
we remark that we are looking to construct the state $A\ket{v} + B\ket{w}$, and that this is equal to $A\cdot (\ket{v}+A^{-1}B\ket{w})$.
This gives us the opportunity to ``factor out'' the LIM $A$, and only store the tuple $(v,A^{-1}B,w)$.
We can do even better by finding a canonically chosen LIM $C=\rootlabel(\ledge {A^{-1}B}w)$ (on \autoref{algline:add:add-cache-factor-C}) and storing $(v,C,w)$ (on line \autoref{algline:add:add-store}).
This way, we get a cache hit at \autoref{algline:add:lookup-cache-factor-C} upon the call $\textsc{Add}(\ledge Dv, \ledge Ew)$ whenever $A^{-1}B\ket{w}=D^{-1}E\ket{w}$.
This happens of course in particular when $(A,v,B,w)=(D,v,E,w)$, but can happen in exponentially more cases; therefore, this technique works at least as well as the ``straightforward'' way outlined above.
Finally, on \autoref{algline:add:add-swap-cache-lookup}, we take advantage of the fact that addition is commutative; therefore it allows us to pick a preferred order in which we store the nodes, thus improving possible cache hits by a factor two.
We also use $C$ in the recursive call at \autoref{algline:add:add-0} and \ref{algline:add:add-1}.

\begin{wrapfigure}[11]{r}{6cm}
\vspace{-1.3em}
\begin{tikzpicture}[
    scale=0.3,
    every path/.style={>=latex},
    every node/.style={},
    inner sep=0pt,
    minimum size=14pt,
    line width=1pt,
    thick,
    font=\footnotesize
    ]
    
    
    \node[draw,circle] (a1) at ( 0,-0)     {};
    \node[draw,circle] (a2) at (0,-3) {};

    \node[draw,circle,rectangle,minimum size=0.4cm] (w1) at  (-0,-6) {$1$};
    
    \draw[<-] (a1) --++(-2cm,0);
    \draw[e0=20] (a1) edge  node[left] {$1$} (a2);
    \draw[e1=20](a1) edge  node[right] {$4$} (a2);
    
    \draw[e0=20] (a2) edge node[left] {$0$} (w1);
    \draw[e1=20]       (a2) edge  node[right] {$1$} (w1);

    \node[draw,circle] (a1) at (5,-0)     {};
    \node[draw,circle] (a2) at (5,-3) {};
    
    \node[draw,circle,rectangle,minimum size=0.4cm] (w1) at ( 5,-6) {$1$};
    
    \draw[<-] (a1) --++(-2cm,0);
    \draw[e0= 20] (a1) edge  node[left] {$1$} (a2);
    \draw[e1= 20](a1) edge  node[right] {$2$} (a2);
    
    \draw[e0= 20] (a2) edge node[left] {$1$} (w1);
    \draw[e1= 20]       (a2) edge  node[right] {$2$} (w1);

    \node[draw,circle] (a1p) at ( 13,-0)     {};
    \node[draw,circle] (a20) at ( 10,-3) {};
    \node[draw,circle] (a21) at ( 16,-3) {};
    
    \node[draw,circle,rectangle,minimum size=0.4cm] (w1) at ( 13,-6) {$1$};
    
    \draw[<-] (a1p) --++(-2cm,0);
    \draw[e0] (a1p) edge  node[left] {1} (a20);
    \draw[e1]     (a1p) edge  node[right] {2} (a21);
    
    \draw[e0= 20] (a20) edge  node[left] {1} (w1);
    \draw[e1= 20]       (a20) edge  node[above right, xshift=-4pt] {3} (w1);
    
    \draw[e0= 20] (a21) edge node[above left, xshift=4pt] {1} (w1);
    \draw[e1= 20]       (a21) edge  node[right] {4} (w1);

    \node[] (a2) at (2.5,-3) {$+$};
    \node[] (a2) at (7.5,-3) {$=$};
    
\end{tikzpicture}
	    \caption{Adding two states \protect\vect{0,1,0,4} and \protect\vect{1,2,2,4} as \qmdds can cause an exponentially larger result \qmdd  \protect\vect{1,3,2,8} due to the loss of common factors.}
	    \label{fig:explosion}
\end{wrapfigure}
The worst-case runtime of \textsc{Add} is $\oh(n^32^n)$ (exponential as expected),
where $n$ is the number of qubits.
This can happen when the resulting \limdd is exponential in the input sizes (bounded by $2^n$), as identified for \qmdds in \cite[Table 2]{fargier2014knowledge}.
The reason for this is that addition may remove any common factors, as illustrated in \autoref{fig:explosion}.
However, the \textsc{Add} algorithm is polynomial-time when $v=w$ and $v$ is a stabilizer state, which is sufficient to show that the Hadamard gate can be efficiently applied to stabilizers represented as \limdd, as we demonstrate next in 
\autoref{sec:clifford-polytime}.

\subsubsection{\limdd operations for Clifford gates are polynomial time on stabilizer states}
\label{sec:clifford-polytime}

We give an algorithm for the Hadamard gate and then show that it can be applied to a stabilizer state in polynomial time.
Together with the results of \autoref{sec:simple-gates}, this shows that all Clifford gates can be applied to stabilizer states in polynomial time.
The key ingredient is \autoref{thm:recursive-add-calls}, which describes situations in which the \textsc{Add} procedure looks up the same tuples in the cache in both its recursive calls (modulo $\pm1$).
\autoref{thm:clifford-polytime} gives the final result.

\begin{theorem}
	\label{thm:clifford-polytime}
	Any Clifford gate ($H,S$, CNOT) can be applied in $\oh(n^4)$ time to any (combination of) qubits to a \limdd representing a stabilizer state.
\end{theorem}
\begin{proof}
	Let $\ket\psi$ be an $n$ qubit stabilizer state, represented by a \limdd with root edge $\ledge Av$.
	By \autoref{thm:pauli-tower-limdds-are-stabilizer-states}, this \limdd is a
	\gen\pauli-Tower-\limdd with $m=n$ nodes apart from the leaf.

	\autoref{sec:simple-gates} shows that any $S$-gate can be applied in time $\oh(n^3 m)$, so we get $\oh(n^4)$.

	\autoref{thm:hadamard-stabilizer-polytime} shows that any Hadamard gate can be applied on any qubit in time $\oh(n^4)$.
	
	\autoref{sec:simple-gates} shows that any downward CNOT-gate can be applied in time $\oh(n^3 m)$, so in this case $\oh(n^4)$. By applying Hadamard to the target and control qubits, 
	before and after the downward CNOT, we obtain an upward CNOT, i.e., $CX_c^t = (H\otimes H) CX_t^c (H \otimes H)$, still in time $\oh(n^4)$.
\end{proof}

\begin{wrapfigure}[6]{r}{3.8cm}
\vspace{-1.3em}
~\hspace{-1.8em}
\begin{tikzpicture}[->,>=stealth',shorten >=1pt,auto,node distance=1cm,
        thick, state/.style={circle,draw,minimum size=14pt},font=\footnotesize]

    \node[state](r) {$v$};
    \draw[<-] (r) --++(-.6cm,0);
    \node[](1a)[below = .7cm of r, xshift=-.29cm]{};
    \node[](1b)[below = .7cm of r, xshift= .29cm]{};

	\node[right = .5cm of r,yshift=-.0cm,rotate=-0] {$\rightsquigarrow$}; 
    
    \path[]
    (r)  edge[e1]     node[right, pos=.3] {$e_0$} (1b)
    (r)  edge[e0]     node[left,pos=.3] {$e_1$} (1a)
    ;

    \node[state,right = 1.5cm of r] (r) {$v'$};
    \node[below = of r,xshift=-.65cm] (a1h) {};
    \node[below = of r,xshift= .65cm] (a2h) {};
    \node[draw, dotted, below = .4cm of r,xshift=-.7cm,rounded corners,minimum width=1.3cm,minimum height=1.4cm] (a1hh) {};
    \node[draw, dotted, below = .4cm of r,xshift= .7cm,rounded corners,minimum width=1.3cm,minimum height=1.4cm] (a2hh) {};

    \node[state,below = 1.2cm of r,xshift=-1.cm] (a11) {};
    \draw[<-] (a11) --++(90:.6)  node[pos=1.6,xshift=.2cm] {$e_{0}$} node[pos=1.6,xshift=.5cm] {$+$};
    \node[state,below = 1.2cm of r,xshift=-.4cm] (a12) {};
    \draw[<-] (a12) --++(90:.6)  node[pos=1.6,xshift=.25cm] {$e_{1}$};
    \node[state,below = 1.2cm of r,xshift= .4cm] (a21) {};
    \draw[<-] (a21) --++(90:.6)  node[pos=1.6,xshift=.25cm] {$e_{0}$} node[pos=1.6,xshift=.55cm] {$-$};
    \node[state,below = 1.2cm of r,xshift= 1.cm] (a22) {};
    \draw[<-] (a22) --++(90:.6)  node[pos=1.6,xshift=.3cm] {$e_{1}$};

    \draw[<-] (r) --++(1cm,0) node[pos=1,lbl] {$\nicefrac 1{\sqrt 2}$};
    
    \draw[e0] (r) edge  node[above left]  {$a_0$} (a1hh.north);
    \draw[e1] (r) edge  node[above right] {$a_1$} (a2hh.north);
\end{tikzpicture}
\end{wrapfigure}
\paragraph{To apply a \textbf{Hadamard gate}} ($H=\frac 1{\sqrt 2}\begin{smallmat}1 & 1 \\ 1 & -1\end{smallmat}$) to the first qubit, we first construct edges representing the states $\ket{a_0}=\ket{e_0}+\ket{e_1}$ and $\ket{a_1}=\ket{e_0} - \ket{e_1}$, using the \textsc{Add} procedure (\autoref{alg:add-limdds} and multiplying the root edge with $-1$).
		Then we construct an edge representing the state $\ket 0\ket{a_0} + \ket 1\ket{a_1}$ using \makeedge.
		Lastly, the complex factor on the new edge's root label is multiplied by $\frac {1}{\sqrt 2}$. Since the Hadamard is also a Clifford gate, we can apply this operation to any qubit in the \limdd by pushing it through the LIMs, as we saw in \autoref{sec:simple-gates}.
\autoref{alg:H-gate} shows the complete algorithm.

\begin{algorithm}[h!]
	\begin{algorithmic}[1]
        \Procedure{HGate}{$\Edge~ \ledge{A}{v}$ \textbf{with} $A\in \pauli$-LIM,
        	$k \in \set{1,\dots, \index(v)}$}
		\If{$v \notin \textsc{HGate}\cache$}\label{sg:cache1} \Comment{Compute result once for $v$ and store in cache:}
		\If{$\index(v)=k$}
		\State $\textsc{HGate}\cache[v] := \nicefrac 1{\sqrt 2} \cdot\makeedge(\mathsf{Add}(\low v,  \high v), \mathsf{Add}(\low v,  -\high v))$\hspace{-1em}
		\label{hg:hk}
		\Else
		\State $\textsc{HGate}\cache[v] := \makeedge(\textsf{HGate}(\low v,k), \textsf{HGate}(\high v,k))$
		\label{hg:rebuild}
		\EndIf
		\EndIf
		\State \Return $H_k A H_k^\dagger  \cdot \textsc{HGate}\cache[v]$
							\Comment{Retrieve result from cache}
		\label{hg:rebuild2}
		\EndProcedure
	\end{algorithmic}
    \caption{Apply gate $H$ to qubit $k$ for \pauli-\limdd \ledge Av. We let $n = \index(v)$.}
	\label{alg:H-gate}
\end{algorithm}

\begin{theorem}
	\label{thm:hadamard-stabilizer-polytime}
	Let $e$ be an $n$-qubit \gen\pauli-Tower-\limdd.
	$\textsc{HGate}(e, k)$ of \autoref{alg:H-gate} takes $\oh(n^4)$ time.
\end{theorem}
\begin{proof}
	By virtue of the cache, \textsc{HGate} is called at most once per node. 
	Since the LIMDD is a Tower, there are only $n$ nodes; so \textsc{HGate} is called at most $n$ times.
	For the node at level  $k$, \textsc{HGate} 
		makes two calls to \textsc{Add} on \autoref{hg:hk}.
	By applying induction over the qubits $1,\dots, k$, using \autoref{thm:recursive-add-calls}, it is easy to see that at each level, the cache in \autoref{alg:add-limdds-cache} is consulted at \autoref{algline:add:lookup-cache-factor-C} with a tuple $(v_k, P^{(k)}, v_k)$ or $(v_k, -P^{(k)}, v_k)$.
	This tells us that \textsc{Add} performs at most $2k$ recursive calls.
	Each recursive call to \textsc{Add} may invoke the \textsc{MakeEdge} procedure, which runs in time $\mathcal O(n^3)$, yielding a total worst-case running time of $\mathcal O(n^4)$, when $k = \Omega(n)$.
\end{proof}

\begin{theorem}
	\label{thm:recursive-add-calls}
If \autoref{alg:add-limdds-cache} is called on two edges pointing to the
same \gen\pauli-Tower-\limdd node $v$ with $\low v = \high v$,
then the recursive \textsc{Add} calls at
\autoref{algline:add:add-0}, \ref{algline:add:add-1}
both lookup the same LIM in cache up to a factor $\pm 1$.
\end{theorem}
\begin{proof}
Assume the algorithm is at \autoref{algline:add:add-0}.
Let \lnode[v]{}w{Q}w be a node on which the algorithm was called.
Let $C = P_n \otimes P$ be the $n$ qubit \pauli-LIM computed at \autoref{algline:add:add-cache-factor-C} with $P_n \in \pauli$ and $P$ an $n-1$ qubit \pauli-LIM.
At \autoref{algline:add:add-0} and \ref{algline:add:add-1}, \textsc{Add} makes two recursive calls computing $a_x$ for $x\in \set{0,1}$,
as listed in the header of the below table.
The \follow{x}e semantics yield four cases cases for the parameters in a recursive \textsc{Add} calls, depending on $x$ and $P_n$.
The following table shows the tuples computed for cache normalization at \autoref{algline:add:add-cache-factor-C} \emph{in the recursive call}, 
ignoring the \rootlabel{}() function for now.
E.g., if $\rightsquigarrow$ denotes cache normalization,
then $\ledge{\gamma R} v, \ledge{\gamma  Q^{-1}R}w \rightsquigarrow (v, \pm Q, w)$  since $QR = \pm RQ$ for $P,Q \in\paulilim$:

\begin{tabular}{l|ll|ll|}
&
\multicolumn{2}{c|}{\footnotesize$a_0 := \textsc{Add}(\follow 0{\ledge {}v},\follow 0{\ledge Cv})$} & 
\multicolumn{2}{c|}{\footnotesize$a_1 := \textsc{Add}(\follow 1{\ledge {}v},\follow 1{\ledge Cv})$}
\\\hline
$P_n=\id$:&$\ledge{} w, \ledge{P}w$ 
			& $\rightsquigarrow (w, P, w)$
			& $\ledge{Q} w, \ledge{ PQ}w$ 
			& $\rightsquigarrow (w, \pm{P}, w)$
\\
$P_n = X$:	& $\ledge{} w, \ledge{ P Q}w $ 
			& $\rightsquigarrow (w,  PQ, w)$
		 	& $\ledge{ Q} w, \ledge{ P}w$ 
		 	& $\rightsquigarrow (w, \pm PQ, w) $
\\
$P_n = Y$: 	& $\ledge{} w, \ledge{-i PQ}w$ 
			& $\rightsquigarrow  (w, -i PQ, w)$
			& $\ledge{Q} w, \ledge{i P }w$
			& $\rightsquigarrow (w, \pm iPQ, w)$
\\
$P_n = Z$: 	& $\ledge{}w, \ledge{ P}w $
			& $\rightsquigarrow (w, { P}, w) $
			& $\ledge{Q} w, \ledge{-PQ}w$ 
			& $\rightsquigarrow (w,  {\pm P}, w)$
\\
\end{tabular}

In all four cases, the cache-normalized LIMs computed in both recursive calls
are equivalent up to a factor $\pm 1$.
Finally, our \rootlabel{}() function from \autoref{sec:choose-canonical-isomorphism-pauli}, which selects the lexicographic smallest label, satisfies $\rootlabel{}(\ledge Qw) = - \rootlabel{}(\ledge {-Q}w)$ for  any \paulilim~$Q$.
This completes the proof.
\end{proof}

Since stabilizer states are closed under Clifford gates, \gen\pauli-Tower-\limdds should also be closed under the respective \limdd manipulation operations.
We show this in \autoref{thm:clifford-gate-stabilizer-limdd-general} (\autoref{sec:graph-states-limdds}).

\subsection{Comparing \limdd-based simulation with other methods}
\label{sec:simulation}

\autoref{prop:simulation} shows exponential advantages of (\pauli-)\limdds over three state-of-the-art classical quantum circuit simulators: those based on \qmdds and MPS~\cite{white1992density,perez2006matrixProductStateRepresentations}, and the Clifford + $T$ simulator.
In this section we prove the proposition, mainly using results from the current section:
To show the separation between simulation with \limdds and Clifford + $T$, we present \autoref{thm:limdds-are-fast-for-wn}.

Our proofs often rely on the fact that \limdds are exponentially more succinct representations of a certain class of quantum states $S$ that are generated by circuits with a certain (non-universal) gate set $G$.
For instance, the stabilizer states that are generated by the Clifford gate set.
	\limdd-based simulation ---similar to MPS~\cite{vidal2003efficient} and \qmdd-based~\cite{zulehner2018advanced} simulation--- proceeds by representing a state
	$\ket{\phi_t}$ at time step $t$ as a \limdd $\phi_{t}$. It then applies the gate  $U_t \in G$ in the circuit corresponding to this time step to obtain a \limdd $\phi_{t+1}$ with
	$\ket{\phi_{t+1}} = U_t \ket{\phi_t}$, thus yielding strong simulation at the final time step as reading amplitudes from the final \limdd is easy (see \autoref{sec:isomorphism-qmdd}). 

It follows that \limdd-based simulation is efficient provided that it can execute all gates $U_t$ in polynomial time (in the size of the \limdd representation), at least for the states in $S$.
Note in particular that since the execution stays in $S$, i.e., $\ket{\phi_{t}} \in S \implies \ket{\phi_{t+1}} \in S$, the representation size can not grow to exponential size in multiple steps ($S$ can be considered an inductive invariant in the style of Floyd~\cite{floyd1993assigning} and de Bakker \& Meertens~\cite{de1975completeness}).
On the other hand, since MPS and \qmdd are exponentially sized for cluster states, they necessarily require exponential time on circuits computing this family of states.

\subsubsection{\limdd is exponentially faster than \qmdd-based simulation}
\label{sec:simulation:qmdd}

As state set $S$, we select the stabilizer states and for $G$ the Clifford gates.
\autoref{thm:pauli-limdd-is-stabilizer} shows that \limdds for stabilizers are always quadratic in size in the number of qubits $n$, as the diagram contains $n$ nodes and $n+1$ LIMs, each of size  at most $n$ (see \autoref{def:tower}).
\autoref{sec:simple-gates} shows that \limdd can execute all Clifford gates on stabilizer states in time $\mathcal O(n^4)$.

On the other hand, \autoref{thm:graph-state-qmdd-lower-bound} shows that \qmdds for cluster states are exponentially sized.
It follows that in simulation also, there is an exponential separation between \qmdd and \limdd, proving that  $ \qcsim \qmdd = \Omega^*(2^n \cdot \qcsim \limdd)$ (\autoref{prop:simulation} \autoref{sim:qmdd}).

For the other direction, we now show that \limdds are at most a factor $\mathcal O(n^3)$ slower than \qmdds on any given circuit.
First, a \limdd never contains more nodes than a \qmdd representing the same state (because \qmdd is by definition a specialization of \limdd, see \autoref{sec:isomorphism-qmdd}).
The \limdd additionally uses $\mathcal O(n)$ memory per node to store two Pauli LIMs; thus, the total memory usage is at most a factor $\mathcal O(n)$ worse than \qmdds for any given state.
The \textsf{ApplyGate} and \textsf{Add} algorithms introduced in \autoref{sec:applygate} are very similar to the ones used for \qmdds in~\cite{zulehner2017one,fujita1997multi}.
In particular, our \textsf{ApplyGate} and \textsf{Add} algorithms never make more recursive calls than those for \qmdds.
However, one difference is that our \makeedge algorithm runs in time $\mathcal O(n^3)$ instead of $\mathcal O(1)$.
Therefore, in the worst case these \limdd algorithms make the same number of recursive calls to \textsf{ApplyGate} and \textsf{Add}, in which case they are slower by a factor $\mathcal O(n^3)$.

Finally, \autoref{thm:limdd-superset-qmdd-plus-stabilizers} shows that the pseudo-cluster state $\ket \phi$ has a polynomial representation in \limdd. By definition of the pseudo-cluster state, post-selecting (constraining) the top qubit to 0 (or 1) yields the cluster state $\ket{G_n}$. Therefore, \qmdd for the pseudo-cluster state must have exponential size, as 
constraining can never increase the size of DD~\cite[Th~2.4.1]{wegener2000branching}.
Together with the universal simulation discussed above, this proves that the above also holds for 
for a simulator based on the combination \qmddc (\autoref{prop:simulation} \autoref{sim:qmddc}).

\subsubsection{\limdd is exponentially faster than MPS}

In \autoref{sec:simulation:qmdd}, we saw that \limdd can simulate the cluster state in polynomial time.
On the other hand, \autoref{thm:mps} shows that MPS for cluster states are exponentially sized.
It follows that in simulation also, there is an exponential separation between MPS and \limdd, proving \autoref{prop:simulation} \autoref{sim:mps}.

\subsubsection{\limdd is exponentially faster than \cliffordt}
\label{sec:simulation:cliffordt}

In this section, we consider a circuit family that \limdds can efficiently simulate, but which is difficult for the Clifford+$T$ simulator because the circuit contains many $T$ gates, assuming the Exponential Time Hypothesis (ETH, a standard complexity-theoretic assumption which is widely believed to be true).
This method decomposes a given quantum circuit into a circuit consisting only of Clifford gates and the $T=\diagonal{e^{i\pi/4}}$ gate, as explained in \autoref{sec:preliminaries}.

The circuit family, given my McClung \cite{mcclung2020constructionsOfWStates}, maps the input state $\ket{0}^{\otimes n}$ to the $n$-qubit $W$ state $\ket{W_n}$, which is the equal superposition over computational-basis states with Hamming weight $1$,
    \begin{equation*}
        \ket{W_n} = \frac{1}{\sqrt{n}} \left( \ket{100\dots00} + \ket{010\dots00} + \dots + \ket{000\dots01} \right)
    \end{equation*}
Arunachalam et al. showed that, assuming ETH, any circuit which deterministically produces the $\ket{W_n}$ state in this way requires $\Omega(n)$ $T$ gates \cite{arunachalam2022parameterized}.
Consequently, the Clifford + $T$ simulator cannot efficiently simulate the circuit family, even when one allows for preprocessing with a compilation algorithm aiming to reduce the $T$-count of the circuit (such as the ones developed in \cite{kissinger2019reducingTCount,thapliyal2019quantumOptimizingTCount}).

\autoref{thm:limdds-are-fast-for-wn} now shows that the exponential separation between simulation with \limdd and \cliffordt, i.e., that $\qcsim{\cliffordt} = \Omega(2^n \cdot \qcsim \limdd)$
(\autoref{prop:simulation} \autoref{sim:cliffordt}). 
\autoref{sec:efficient-w-state-preparation} gives its proof.

\begin{restatable}{theorem}{thmw}
\label{thm:limdds-are-fast-for-wn}
    There exists a circuit family $C_n$ such that $C_n \ket{0}^{\otimes n} = \ket{W_n}$, that Pauli-\limdds can efficiently simulate.
    Here simulation means that it constructs representations of all intermediate states, in a way which allows one to, e.g., efficiently simulate any single-qubit computational-basis measurement or compute any computational basis amplitude on any intermediate state and the output state.
\end{restatable}

We note that we could have obtained a similar result using the simpler scenario where one applies a $T$ gate to each qubit of the $(\ket{0} + \ket{1})^{\otimes n}$ input state.
However, our goal is to show that \limdds can natively simulate scenarios which are relevant to quantum applications, such as the stabilizer states from the previous section.
The $W$ state is a relevant example, as several quantum communication protocols use the $W$ state~\cite{wang2007quantum, liu2011efficientWState, lipinska2018anonymous}.
In contrast, the circuit with only $T$ gates yields a product state, hence it is not relevant unless we consider it as part of a larger circuit which includes multi-qubit operations.

Lastly, it would be interesting to analytically compare \limdd with general stabilizer rank based simulation (without assuming ETH).
However, this would require finding a family of states with provably superpolynomial stabilizer rank, which is a major open problem.
Instead, we implemented a heuristic algorithm by Bravyi et al. \cite{bravyi2019simulation} to empirically find upper bounds on the stabilizer rank and applied it to a superset of the $W$ states, so-called Dicke states, which can be represented as polynomial-size \limdd.
The $\mathcal O(n^2)$-size \limdd can be obtained via a construction by Bryant \cite{bryant86}, since the amplitude function of a Dicke state is a symmetric function.
The results hint at a possible separation but are inconclusive due to the small number of qubits which the algorithm can feasibly investigate in practice.
See \autoref{sec:stabrank_search} for details.

\section{Canonicity: Reduced \limdds with efficient \makeedge algorithm}
\label{sec:canonicity}
Unique representation, or canonicity, is a crucial property for the efficiency and effectiveness of decision diagrams.
In the first place, it allows for circuit analysis and simplification~\cite{bryant1995verification,miller2006qmdd}, by
facilitating efficient manipulation operations
through dynamic programming efficiently, as discussed in \autoref{sec:quantum-simulation}.
In the second place, a reduced diagram is smaller than an unreduced diagram because it merges nodes with the same semantics. For instance, Pauli-\limdds allow all states in the same $\simeq_{\text{Pauli}}$ equivalence class to be merged.
Here, we define a reduced \pauli-\limdd, which is canonical.

In general, many different \limdds can represent a given quantum state,
as illustrated in \autoref{fig:necessity-of-reduction-rules-example}.
However, by imposing a small number of constraints on the diagram, listed in \autoref{def:reduced-limdd} and visualized in \autoref{fig:necessity-of-reduction-rules}, we ensure that every quantum state is represented by a unique `\emph{reduced}'
\pauli-\limdd.
We present a \makeedge algorithm (\autoref{alg:make-edge} in \autoref{sec:makeedge}) that computes
a canonical node assuming its children are already canonical.
The algorithms for quantum circuit simulation in \autoref{sec:quantum-simulation} ensure that all intermediate \limdds are reduced by creating nodes exclusively through this subroutine.

\begin{figure}[t!]
\footnotesize
\centering
\begin{tabular}{|c|c|c|c|}
\hline
&&&\\[-1ex]
\hspace{-.1cm}
    \begin{tikzpicture}[->,>=stealth',shorten >=1pt,auto,node distance=1cm,
        thick, state/.style={circle,draw,minimum size=0.5cm},font=\scriptsize, scale=0.3,
    inner sep=0pt,]
        
    \node[state](1){};
    \node[state](1a)[below= .7cm of 1,xshift=-.5cm]{$v$};
    \node[state](1b)[below= .7cm of 1,xshift= .5cm]{$v'$};
    \node[leaf] (s)[below = .7cm of 1a, xshift= .5cm]{$1$};
        
    \draw[<-,e1] (1) --++(90:2.4cm) node[lbl] {$\frac 1{\sqrt 2}\id \otimes \id$};
    
    \path[]
    (1)  edge[e0= 20]  node[pos=.6,left] {} (1a)
    (1)  edge[e1= 20]  node[pos=.45,] {} (1b)
    (1a) edge[e0= 40]  node[pos=.6,left] {} (s)
    (1a) edge[e1=-10]  node[lbl,pos=.45,] {$0$} (s)
    (1b) edge[e0=-10]  node[lbl,pos=.25,left] {$0$} (s)
    (1b) edge[e1= 40]  node[pos=.6,] {} (s);
    
    \end{tikzpicture}\hspace{-.1cm}~
&\hspace{-.1cm}
    \begin{tikzpicture}[->,>=stealth',shorten >=1pt,auto,node distance=1cm,
        thick, state/.style={circle,draw,minimum size=0.5cm},font=\scriptsize, scale=0.3,
    inner sep=0pt,]
        
    \node[state](1){};
    \node[state](1a)[below= .7cm of 1,xshift=-.5cm]{$v$};
    \node[state](1b)[below= .7cm of 1,xshift= .5cm]{$v'$};
    \node[leaf] (s)[below = .7cm of 1a, xshift= .5cm]{$1$};
        
    \draw[<-,e1] (1) --++(90:2.4cm) node[lbl] {$\frac 1{\sqrt 2} X \otimes \id$};
    
    \path[]
    (1)  edge[e0=-20]  node[pos=.6,left] {} (1b)
    (1)  edge[e1=-20]  node[pos=.45,] {} (1a)
    (1a) edge[e0= 40]  node[pos=.6,left] {} (s)
    (1a) edge[e1=-10]  node[lbl,pos=.45,] {$0$} (s)
    (1b) edge[e0=-10]  node[lbl,pos=.25,left] {$0$} (s)
    (1b) edge[e1= 40]  node[pos=.6,] {} (s);
    
    \end{tikzpicture}\hspace{-.1cm}~
&\hspace{-.1cm} 
    \begin{tikzpicture}[->,>=stealth',shorten >=1pt,auto,node distance=1cm,
        thick, state/.style={circle,draw,minimum size=0.5cm},font=\scriptsize, scale=0.3,
    inner sep=0pt,]
        
    \node[state](1){};
    \node[state](1a)[below= .7cm of 1]{$v$};
    \node[leaf] (s)[below = .7cm of 1a]{$1$};
        
    \draw[<-,e1] (1) --++(90:2.4cm) node[lbl] {$\frac 1{\sqrt 2} \id \otimes \id$};
    
    \path[]
    (1)  edge[e0= 20]  node[pos=.6,left] {} (1a)
    (1)  edge[e1= 20]  node[lbl,pos=.45,] {$X$} (1a)
    (1a) edge[e0= 20]  node[pos=.6,left] {} (s)
    (1a) edge[e1= 20]  node[lbl,pos=.45,] {$0$} (s)
	;    
    \end{tikzpicture}\hspace{-.1cm}~
&\hspace{-.1cm} 
    \begin{tikzpicture}[->,>=stealth',shorten >=1pt,auto,node distance=1cm,
        thick, state/.style={circle,draw,minimum size=0.5cm},font=\scriptsize, scale=0.3,
    inner sep=0pt,]
        
    \node[state](1){};
    \node[state](1a)[below= .7cm of 1]{$v$};
    \node[leaf] (s)[below = .7cm of 1a]{$1$};
        
    \draw[<-,e1] (1) --++(90:2.4cm) node[lbl] {$\frac 1{\sqrt 2} Z \otimes \id$};
    
    \path[]
    (1)  edge[e0= 20]  node[pos=.6,left] {} (1a)
    (1)  edge[e1= 20]  node[lbl,pos=.45,] {$-X$} (1a)
    (1a) edge[e0= 20]  node[pos=.6,left] {} (s)
    (1a) edge[e1= 20]  node[lbl,pos=.45,] {$0$} (s)
	;    
    \end{tikzpicture}\hspace{-.1cm}~
\\[1ex]
\hline
\end{tabular}
\caption{
Four different \pauli-\limdds representing the Bell state $\frac1{\sqrt 2} (\ket{00} + \ket{11})$.
From left to right: as \id-\limdd, swapping high and low nodes $v,v'$ by placing an $X$ on the root LIM, merging $v'$ into $v$ by observing that $\ket v = X \ket{v'}$ and selecting a different high LIM $-X$ together with changing the root LIM. This section shows that selecting a unique high LIM is the most challenging, as in general many LIMs can be chosen.
}
\label{fig:necessity-of-reduction-rules-example}
\end{figure}

\subsection{\limdd canonical form}
\label{sec:def-reduced-limdd}

The main insight used to obtain canonical decision diagrams is that a canonical form
can be computed locally for each node, assuming its children are already canonical.
In other words, if the diagram is constructed bottom up, starting from the leaf,
it can immediately be made canonical.
(This is why decision diagram manipulation algorithms always construct the diagram in the backtrack of the recursion using a typical `MakeNode' procedure for constructing canonical nodes~\cite{fujita1997multi}, like in \autoref{sec:quantum-simulation}.)
For instance, a \qmdd node \lnode{\alpha}{v}{\beta}{w} with $\alpha,\beta \in \mathbb C\setminus \set 0$
can be \concept{reduced} into a canonical node by dividing out a common factor $\alpha$ and placing it on the root edge. Assuming that $v,w$ are canonical, the resulting node \lnode1v{\nicefrac\beta\alpha}w can be stored as a tuple \tuple{1,v,\nicefrac\beta\alpha, w} in a hash table.
Moreover, \emph{any other node that is equal to this node up to a scalar is reduced to the same tuple with this strategy}~\cite{miller2006qmdd} and thus merged in the hash table.

For \limdd, we use a similar approach of dividing out `common LIM factors.' However, we need to do additional work to obtain a unique high edge label ($\nicefrac \beta\alpha$ in the example above), as the \paulilim{} group is more complicated than the group of complex numbers (scalars).

\autoref{def:reduced-limdd} gives reduction rules for \limdds and  \autoref{fig:necessity-of-reduction-rules} illustrates them.
The merge (1) and low factoring (4) rules fulfill the same purpose as in the \qmdd case discussed above.
In a \pauli-\limdd, we may always swap high and low edges of a node $v$
by multiplying the root edge LIM with
$X\otimes \id$, as illustrated in \autoref{fig:necessity-of-reduction-rules-example}.
The low precedence rule (3) makes this choice deterministic, but only in case 
$\low v \neq \high v$.
Next, the zero edges (2) rule handles the case when
$\alpha$ \'or $\beta$ are zero in the above,
as in principle a edge $e$ with label 0 could point to any node on the next level $k$,
as this always yields a 0 vector of length $2^k$ (see semantics below \autoref{def:limdd}).
The rule forces $\low v = \high v$ in case either edge has a zero label.
We explain the interaction among the zero edges (2), low precedence (3) and low factoring (4) rules below.
Finally, the high determinism rule (5) defines a deterministic function to choose
LIMs on high edges,
solving the most challenging problem of uniquely selecting a LIM on the high edge.
We give an $\oh (n^3)$ algorithm for this function in \autoref{sec:makeedge}.

\begin{definition}[Reduced \limdd]
	\label{def:reduced-limdd}
	A \pauli-\limdd is \emph{reduced} when it satisfies the following constraints.
	It is \emph{semi-reduced} if it satisfies all constraints except possibly high determinism.
	\begin{enumerate}
		\item \textbf{Merge: } No two nodes are identical: We say
            two nodes $v,w$ are identical~if $\low v= \low w,$\\$\high v = \high w,
                \lbl(\low v)=\lbl(\low w)$, $\lbl(\high v)=\lbl(\high w)$.
		\item \textbf{(Zero) edge: } For any edge $(v,w) \in \High \cup \Low $,
if $\lbl(v,w)=0$, then both edges outgoing from $v$ point to the same node, i.e., $\high v = \low v = w$. 
		\label{rule:zero-edges}
		\item \textbf{Low precedence: } Each node $v$ has $\low v \beforeq \high v$, where  $\beforeq$ is a total order on nodes.
		\item \textbf{Low factoring: } The label on every low edge to a node $v$ is the identity $\id^{\otimes\index(v)}$.
		\item \textbf{High determinism: } The label on the high edge of any node $v$ is $\highlim =\highlabel(v)$, where $\highlabel$ is a function that takes as input a semi-reduced $n$-{\Pauli}-\limdd node~$v$, and outputs an $(n-1)$-$\Pauli$-LIM $\highlim$
        satisfying
       $\ket v \simeq_{\Pauli} \ket{0}\ket{\low v}+\ket{1}\otimes \highlim\ket{\high v}$.
        Moreover, for any other semi-reduced node $w$ with $\ket{v}\simeq_{\Pauli}\ket{w}$,
         it satisfies $\highlabel(w) = \highlim$.
		In other words, the function $\highlabel$ is constant within an isomorphism class.		
		\label{rule:high-determinism}
	\end{enumerate}
\end{definition}

\begin{figure}[t!]
\footnotesize
\begin{tabular}{|c|c|c|c|c|}
\hline
&&&&\\[-1ex]
\hspace{-.2cm}
    \begin{tikzpicture}[->,>=stealth',shorten >=1pt,auto,node distance=1cm,
        thick, state/.style={circle,draw,minimum size=0.5cm},font=\scriptsize, scale=0.3,
    inner sep=0pt,]
        
    \node[state](1a)[xshift=-.5cm]{$u$};
    \node[state](1b)[xshift= .5cm]{$v$};
    \node[state] (s)[below = .7cm of 1a, xshift= .5cm]{$w$};
        
    \draw[<-,e1] (1a) --++(90:2.4cm) node[lbl] {$C$};
    \draw[<-,e1] (1b) --++(90:2.4cm) node[lbl] {$D$};

    \path[]
    (1a) edge[e0= 40]  node[lbl,pos=.6,left] {$A$} (s)
    (1a) edge[e1=-20]  node[lbl,pos=.45,] {$B$} (s)
    (1b) edge[e0=-20]  node[lbl,pos=.25,left] {$A$} (s)
    (1b) edge[e1= 40]  node[lbl,pos=.6,] {$B$} (s);
    
    \node[right = .1cm of 1b,yshift=-.0cm,rotate=-0] {$\rightsquigarrow$}; 
    
    \node[state] (n)[right = .6cm of 1b]{$v$};
    \node[state] (u)[below = .7cm of n]{$w$};
    
    \draw[<-,e1] (n) --++(110:2.6cm) node[lbl] {$C$};
    \draw[<-,e1] (n) --++(70:2.6cm) node[lbl] {$D$};

    \path[]
    (n) edge[e0=20]  node[left,pos=.4,lbl] {$A$} (u)
    (n) edge[e1=20]  node[right,pos=.4,lbl] {$B$} (u);
    \end{tikzpicture}\hspace{-.1cm}~
&\hspace{-.15cm}
    \begin{tikzpicture}[->,>=stealth',shorten >=1pt,auto,node distance=1cm,
        thick, state/.style={circle,draw,minimum size=14pt},font=\scriptsize]
        
    \node[state](r) {$v$};
    \node[state](1a)[below = .7cm of r, xshift=-.29cm]{$u$};
    \node[state](1b)[below = .7cm of r, xshift= .29cm]{$w$};
    
    \node[above = .6cm of r] (x) {};

	\node[right = .1cm of r,yshift=-.0cm,rotate=-0] {$\rightsquigarrow$}; 
    
    \path[]
    (x)  edge         node[lbl,left,pos=.1] {$C$} (r)
    (r)  edge[e0]     node[left, pos=.3,lbl] {$0$} (1a)
    (r)  edge[e1]     node[right,pos=.3,lbl] {$B$} (1b)
    ;

    \node[state,right = .57cm  of r, inner sep = 0pt](r) {$v'$};
    \node[state](1a)[below = .7cm of r]{$w$};
    
    \node[above = .6cm of r] (x) {};
    
    \path[]
    (x)  edge         node[lbl,pos=.1] {$C$} (r)
    (r)  edge[e0=15]     node[left, pos=.2,lbl] {$0$} (1a)
    (r)  edge[e1=15]     node[right,pos=.2,lbl] {$B$} (1a)
    ;
    \end{tikzpicture}\hspace{-.1cm}~
&\hspace{-.3cm} 
~ \begin{tikzpicture}[->,>=stealth',shorten >=1pt,auto,node distance=1cm,
        thick, state/.style={circle,draw,minimum size=14pt},font=\scriptsize]
        
    \node[state](r) {$v$};
    \node[state](1a)[below = .7cm of r, xshift=-.29cm]{$w$};
    \node[state](1b)[below = .7cm of r, xshift= .29cm]{$u$};
    
    \node[above = .6cm of r] (x) {};

	\node[right = .2cm of r,yshift=-.0cm,rotate=-0] {$\rightsquigarrow$}; 
    
    \path[]
    (x)  edge         node[lbl,left,pos=.1] {$C$} (r)
    (r)  edge[e1]     node[left, pos=.3,lbl] {$B$} (1a)
    (r)  edge[e0]     node[right,pos=.3,lbl] {$A$} (1b)
    ;

    \node[state,right = .65cm  of r, inner sep = 0pt](r) {$v'$};
    \node[state](1a)[below = .7cm of r, xshift=-.29cm]{$u$};
    \node[state](1b)[below = .7cm of r, xshift= .29cm]{$w$};
    
    \node[above = .6cm of r] (x) {};
    
    \path[]
    (x)  edge         node[lbl,pos=.18,xshift=-.7cm] {$(X \otimes \id^{k}) \cdot C$} (r)
    (r)  edge[e0]     node[left, pos=.3,lbl] {$A$} (1a)
    (r)  edge[e1]     node[right,pos=.3,lbl] {$B$} (1b)
    ;
    \end{tikzpicture}\hspace{-.1cm}~
&\hspace{-.1cm}
    \tikz[->,>=stealth',shorten >=1pt,auto,node distance=1.5cm,
        thick, state/.style={circle,draw,minimum size=14pt},font=\scriptsize]{
    \node[state] (1) {$v$};
    \node (x) [above = .6cm of 1] {};
    \node[state] (1a) [below = .7cm of 1, xshift=-.29cm] {$u$};
    \node[state] (1b) [below = .7cm of 1, xshift= .29cm] {$w$};
    
    \path[]
    (x) edge     node[lbl,left,pos=.2] {$C$} (1)
    (1) edge[e0] node[lbl,above left] {$A$} (1a)
    (1) edge[e1] node[lbl,above right] {$B$} (1b)
    ;
    
	\node[right = .1cm of 1,yshift=-.0cm,rotate=-0] {$\rightsquigarrow$}; 
    
    \node[state,right = .7cm of 1]  (1){$v'$};
    \node (x) [above  = .6cm of 1] {};
    \node[state] (1a) [below = .7cm of 1, xshift=-.29cm] {$u$};
    \node [state](1b) [below = .7cm of 1, xshift= .29cm] {$w$};
    
    \path[]
    (x) edge       node[draw,lbl,xshift=-.8cm,pos=.3] {$(\id \otimes A)\cdot C$} (1)
    (1) edge[e0]   node[above left,pos=.7] {} (1a)  %
    (1) edge[e1]   node[draw,above right, pos=.6, lbl,xshift=-.19cm] {$ A^{-1}  B$} (1b)
    ;
    }\hspace{-.1cm}~
&\hspace{-.3cm}
    \begin{tikzpicture}[->,>=stealth',shorten >=1pt,auto,node distance=1cm,
        thick, state/.style={circle,draw,minimum size=14pt},font=\scriptsize]
        
    \node[state](r) {$v$};
    \node[state](1a)[below = .7cm of r, xshift= .3cm]{$u$};
    \node[state](1b)[below = .7cm of r, xshift=-.3cm]{$w$};
    
	\node[right = .2cm of r,yshift=-.0cm,rotate=-0] {$\rightsquigarrow$}; 

    \node[above = .6cm of r] (x) {};
    
    \path[]
    (x)  edge         node[lbl,left,pos=.2] {$C$} (r)
    (r)  edge[e1]     node[lbl, pos=.5] {$A$} (1a)
    (r)  edge[e0]     node[right,pos=.5] {} (1b);

    \node[state,right = .7 of r](r) {$v'$};
    \node[state](1a)[below = .7cm of r, xshift= .3cm]{$u$};
    \node[state](1b)[below = .7cm of r, xshift=-.3cm]{$w$};
    
    \node[above = .6cm of r] (x) {};
    
    \path[]
    (x)  edge         node[lbl,xshift=-.8cm,pos=.2] {$C\cdot (A H^{-1})$} (r)
    (r)  edge[e1]     node[lbl, pos=.5] {$H$} (1a)
    (r)  edge[e0]     node[right,pos=.5] {} (1b);
    \end{tikzpicture}\hspace{-.2cm}~
\\[1ex]
\hline
(1) Merge $u,v$ into $v$ & (2) Zero edges & (3) Low precedence & (4) Low factoring  & \hspace{-.1cm}(5) High determinism\hspace{-.1cm} \\
	&  & with $u \beforeq w$   &  &  $H = \highlabel(v)$\\
\hline
\end{tabular}
\caption{
Illustrations of the reduction rules from \autoref{def:reduced-limdd} applied at level $k+1$ (i.e., $k+1=\index(v) = \index(v')$).
Note that, in general, the top edges are not necessarily root edges, but could be high and low edges for nodes on level $k+2$.
So, in general, there can be multiple such incoming edges (dashed and solid).
}
\label{fig:necessity-of-reduction-rules}
\end{figure}

We make several observations about reduced \limdds.
First, let us apply this definition to a state
$\ket{0}\otimes A \ket{\phi} + \ket{1}\otimes B\ket{\psi}$ with $\ket \phi \not\simeq_\pauli \ket \psi$, where $A,B\in\paulilim$.
Assume we already have canonical \limdds for $\phi$ and $\psi$
(note that necessarily $\phi\neq \psi$).
We will transform this node so that it satisfies all the reduction rules above.
There is a choice between representing this state as either $\lnode{A}{\phi}{B}{\psi}$ or $\lnode{B}{\psi}{A}{\phi}$,
as these are related by the isomorphism $X\otimes \id$.
The low precedence rule resolves this choice here.
Assuming $\phi \before \psi$, low factoring can now be realized by dividing out the LIM $A$,
yielding a node $\lnode{\id}{\phi}{A^{-1}B}{\psi}$ (with root edge $\id \otimes A$ as in \autoref{fig:necessity-of-reduction-rules} (4)).
Otherwise, if $\psi \before \phi$, we obtain node $\lnode{\id}{\psi}{B^{-1}A}{\phi}$ with incoming edge
$X \otimes B$. 
Finally, since there might be other LIMs $\highlim$ not equal to $B^{-1}A$ that yield the same state, the high determinism rule is finally needed to obtain a canonical node
$\lnode{\id}{\psi}{\highlim}{\phi}$ as shown in \autoref{fig:reduced1}. This last step turns a semi-reduced node into a (fully) reduced node. \autoref{sec:makeedge} discusses it in detail.

Now, let us apply the definition to a state $\ket{1}\otimes A\ket{\phi}$.
First, notice that the zero edges rule forces $\low v= \high v = \phi$ in this case.
There is a choice between representing this state as either $\lnode{A}{\phi}{0}{\phi}$ or $\lnode{0}{\phi}{A}{\phi}$, which denote the states $\ket{0}\otimes A\ket{\phi}$ and $\ket{1}\otimes A\ket{\phi}$,
as these are related by the isomorphism $X\otimes \mathbb I$.
The low factoring rule requires that the low edge label is $\mathbb I$, yielding a node of the form $\lnode{A}{\phi}{0}{\phi}$ with root label $X\otimes A$:
In other words, this rule enforces swapping high and low edges, placing a $X$ on the root label,
\emph{and} dividing out the LIM $A$.
Consequently, the high edge must be labeled with $0$, and therefore, semi-reduction, in this case,
coincides with (full) reduction (no high determinism is required).
Notice also that there is no reduced \limdd for the $0$-vector, 
because low factoring requires low edges with label $\id$.
This is not a problem, since the $0$-vector is not a quantum state.

The rules in \autoref{def:reduced-limdd} are defined only for \pauli-\limdds, to which our results pertain (except for the brief mention of \gen X and \gen Z-\limdds in \autoref{sec:exponential-separations}). We briefly discuss alternative groups here.
If $G$ is a group without the element $X\not\in G$, the reduced \glimdd based on the same rules is not universal (does not represent all quantum states), because the low precedence rule cannot always be satisfied, since it requires that $v_0\beforeq v_1$ for every node.
Hence, in this case, reduced \glimdd cannot represent a state $\ket 0\ket{v_0}+\ket{1}\ket{v_1}$ when $v_1\before v_0$. However, it is not difficult to formulate rules to support these groups $G$; for instance, when $G=\{\id\}$, we recover the \qmdd and may use its reduction rules~\cite{zulehner2018advanced}.

\emph{Nodes and edges in a reduced \limdd need not represent normalized quantum states,} just like in (unreduced) \limdds as explained in \autoref{sec:isomorphism-qmdd}.
Consider, e.g., node $\ell_2$ in \autoref{fig:qmdd-isoqmdd-exposition}, 
which represents state $[1,1,i,i]^\top$.
Because the normalization constant was divided out (see factor $\nicefrac 14$ on the root edge),
this state is not normalized. In fact, the root node does not need to be normalized,
as even reduced \limdds can represent any vector (except for the zero vector).

Lastly, the literature on other decision diagrams~\cite{akers1978binary,bryant86,feinstein2011skipped} often considers a ``redundant test'' or ``deletion'' rule to remove nodes with the same high and low child. This would introduce the skipping of qubit levels, which our syntactic definition disallows, as already discussed in~Footnote~\ref{fn:skipping}.
However, if needed \autoref{def:limdd} could be adapted and a deletion rule could be added to \autoref{def:reduced-limdd}.

\label{sec:proof-of-canonicity}

We now give a proof of \autoref{thm:node-canonicity-strong}, which states that reduced LIMDDs are canonical.
\begin{theorem}[Node canonicity]
	\label{thm:node-canonicity-strong}
    For each $n$-qubit quantum state $\ket{\phi}$, there exists a unique reduced Pauli-\limdd $L$ with root node $v_L$
    such that $\ket{v_L} \simeq \ket{\phi}$.
\end{theorem}

\begin{proof}
	
    We use induction on the number of qubits $n$ to show universality (the existence of an isomorphic \limdd node) and uniqueness (canonicity).

    \textbf{Base case.}
    If $n=0$, then $\ket{\phi}$ is a complex number $\lambda$.
    A reduced Pauli-\limdd for this state is the leaf node representing the scalar $1$.
    To show it is unique, consider that nodes $v$ other than the leaf have an $\index(v) > 0$,
    by the edges rule, and hence represent multi-qubit states.
    Since the leaf node itself is defined to be unique, the merge rule is not needed and canonicity follows.\\    
    Finally, $\ket{\phi}$ is represented by root edge
    $\leafedge \lambda$.
    
    \textbf{Inductive case.}
    Suppose $n>0$.
    We first show existence, and then show uniqueness.
    
    \paragraph{Part 1: existence.}
    We use the unique expansion of $\ket{\phi}$ as $\ket{\phi} = \ket{0} \otimes \ket{\phi_0} + \ket{1}\otimes \ket{\phi_1}$ where $\ket{\phi_0}$ and $\ket{\phi_1}$ are either $(n-1)$-qubit state vectors, or the all-zero vector.
    We distinguish three cases based on whether $\ket{\phi_0}, \ket{\phi_1} = 0$.

   \textbf{Case $\boldsymbol{\ket{\phi_0}, \ket{\phi_1} = 0}$:}
    This case is ruled out because $\ket{\phi} \neq 0$.

    \textbf{Case $\boldsymbol{\ket{\phi_0}=0}$ or $\boldsymbol{\ket{\phi_1} = 0}$:}
        In case $\ket{\phi_0}\neq 0$,
       by the induction hypothesis, there exists a Pauli-\limdd with root node $w$ satisfying
        $\ket{w} \simeq \ket{\phi_0}$. By definition of $\simeq$,
        there exists an $n$-qubit Pauli isomorphism $A$ such that 
        $\ket{\phi_0} = A \ket{w}$.
       We construct the following reduced Pauli-\limdd for $\ket{\phi}$: 
        $\lnode[v] I{w}0{w}$, adding a root edge~~
        \ledge[e_r=]{\id\otimes A}{v} as illustrated in \autoref{fig:reduced1} (left). 
        In case $\ket{\phi_1}\neq 0$, we do the same for root node                                                                                                                                                                                                                                                                                                                                                                                                                                                                                                                                                                                                                                                                                                                                                                                                                                                                                                                                                                                                                                                                                                                     
        In case $\ket{\phi_1}\neq 0$, we do the same for root
            $\ket{w} \simeq \ket{\phi_1} = A \ket w$, but switch the high and the low edge
            by instead a root edge~~\ledge[e_r=]{X \otimes A}{v} (similar to \autoref{fig:necessity-of-reduction-rules} (3)).
        In both cases, it is easy to check that the root node $v$ is reduced as it can be represented by a tuple $(\id, w, 0, w)$, where $w$ is canonical because of the induction hypothesis.
        Also in both cases, we also have $\ket \phi = \ket{e_r}$ because 
        either $\ket \phi = \id \otimes A \ket v$ or 
        $\ket \phi = X \otimes A \ket v$.

\begin{figure}
\centering
\tikz[->,>=stealth',shorten >=1pt,auto,node distance=1.5cm,font=\footnotesize,
        thick, state/.style={circle,draw,inner sep=0pt,minimum size=14pt}]{
    \node[state] (1) {$v'$};
    \node[state] (1a) [below = 1cm of 1, xshift=1.7cm] {$w$};
    
    \node[above = .5cm of 1] (x1) {};%

    \path[]
    (x1) edge      node[pos=.4,above right,pos=.7] {} (1)
    (1) edge[e0,bend left=-20] node[pos=.3,lbl,left] {$A$} (1a)
    (1) edge[e1,bend right=-20] node[pos=.3,lbl,above right] {$0$} (1a)
    ;

    \node[state, right = 2.5cm of 1] (2) {$v$};
    \node[above = .5cm of 2] (x2) {};%
    \path[]
    (x2) edge    node[left,pos=0,lbl] {$\id \otimes A$} (2)
    (2) edge[e0,bend left=-20] node[pos=.3,left] {} (1a)  %
    (2) edge[e1,bend right=-20] node[pos=.3,lbl,right] {$0$} (1a)
    (1) --  node[yshift=.1cm] {$\rightsquigarrow$} (2)
    ;
    }~~~~~~~~~~~~~~~
\tikz[->,>=stealth',shorten >=1pt,auto,node distance=1.5cm,
        thick, state/.style={circle,draw,inner sep=0pt,minimum size=14pt}]{
    \node[state] (1) {$v''$};
    \node[above = .5cm of 1] (x1) {};%
    \node[state] (1a) [below = 1cm of 1, xshift=2.3cm] {$v_L$};
    \node[state] (1b) [below = 1cm of 1, xshift=3.8cm] {$v_R$};
    \path[]
    (1) edge[e0] node[pos=.37,lbl,left] {$A$} (1a)
    (1) edge[e1] node[pos=.2,lbl,above right] {$B$} (1b)
    (1a) --  node[yshift=-.2cm] {$\beforeq$} (1b)
    ;

    \node[state, right = 2.5cm of 1] (2) {$v'$};

    \node[above = .5cm of 2] (x2) {};%
    \path[]
    (x1) edge     node[above left,pos=.4] {} (1) %
    (x2) edge     node[left,lbl,pos=.0] {$\id \otimes A$} (2)
    (2) edge[e0] node[pos=.2,left] {} (1a) %
    (2) edge[e1] node[pos=.16,lbl,right] {$ A^{-1}B$} (1b)
    (1) --  node[yshift=.1cm] {$\rightsquigarrow$} (2)
    ;
    
    \node[state, right = 2.5cm of 2] (3) {$v$};
    \node[above = .5cm of 3] (x3) {};%
    \path[]
    (x3) edge    node[lbl,above left,pos=.4] {$(\id \otimes A)\rootlim$} (3)
    (3) edge[e0] node[pos=.2,above left] {} (1a) %
    (3) edge[e1] node[pos=.3,lbl,below right] {$\highlim$} (1b)
    (2) --  node[yshift=.1cm] {$\rightsquigarrow$} (3)
    ;
    }
	\caption{Reduced node construction in case $\ket{\phi_1} = 0$ (left), and
	        $\ket{\phi_0}, \ket{\phi_1} \neq 0$ and $v_L \beforeq v_R$ (right).
	        Not shown: for cases $\ket{\phi_0} = 0$ and $v_R \beforeq v_L$, we take instead root edge $ X \otimes A$ and swap low/high edges.
            }
	\label{fig:reduced1}
\end{figure}

    \textbf{Case $\boldsymbol{\ket{\phi_0}, \ket{\phi_1} \neq 0}$:}
    By applying the induction hypothesis twice, there exist \pauli-\limdds $L$ and $R$ with root nodes
    $\ket{v_{L}} \simeq \ket{\phi_0}$ and $\ket{v_{R}} \simeq \ket{\phi_1}$.
    {The induction hypothesis implies only a `local' reduction of \limdds $L$ and $R$, but not automatically a reduction of their union. For instance, $L$ might contain a node $v$ and $R$ a node $w$ such that $v \simeq w$. While the other reduction rules ensure that $v$ and $w$ will be structurally the same, the induction hypothesis only applies the merge rule $L$ and $M$ in isolation, leaving two copies of identical nodes $v,w$.
We can solve this by applying merge on the union of nodes in $L$ and $M$, to merge any equivalent nodes, as they are already structurally equivalent by the induction hypothesis.
This guarantees that (also) $v_L,v_R$ are identical nodes.}

    By definition of $\simeq$, there exist $n$-qubit Pauli isomorphisms $A$ and $B$ such that $\ket{\phi_0} = A \ket{v_{L}}$ and $\ket{\phi_1} = B \ket{v_{R}}$.
    In case $v_{L} \beforeq v_{R}$,
       we construct the following reduced Pauli-\limdd for $\ket{\phi}$: the root node is $\lnode[v] {\mathbb I}{v_{L}}E{v_{R}}$, where
    $E$ is the LIM computed by $\highlabel(\lnode{\mathbb I}{v_L}{A^{-1}B}{v_R})$ . 
    Otherwise, if $v_{R}\beforeq v_{L}$, then we construct the following reduced Pauli-\limdd for $\ket{\phi}$: the root node is $\lnode[v] I{v_{R}}F{v_{L}}$, where $F=\highlabel(\lnode{\mathbb I}{v_L}{B^{-1}A}{v_R})$.
    It is straightforward to check that, in both cases, this Pauli-\limdd is reduced.
    Moreover, $\ket v$ isomorphic to $\ket \phi$ 
    as illustrated in \autoref{fig:reduced1} (right).

    \def\Ptop{P_{\textnormal{top}}}
    \def\Prest{P_{\textnormal{rest}}}
    \paragraph{Part 2: uniqueness.}
    To show uniqueness, let $L$ and $M$ be reduced \limdds with root nodes $v_L, v_M$ such that $\ket{v_L} \simeq \ket{\phi} \simeq \ket{v_M}$, as follows,
    	\begin{align}
    	    \lnode[v_L]{A_L}{v_L^0}{B_L}{v_L^1} & & \lnode[v_M]{A_M}{v_M^0}{B_M}{v_M^1}
    	\end{align}
   	The fact that these nodes are isomorphic means that there is a Pauli isomorphism $P$ such that $P\ket{v_L}=\ket{v_M}$.
   	We write $P=\lambda \Ptop\otimes \Prest\ne 0$ where $\Ptop$ is a single-qubit Pauli matrix and $\Prest$ an $(n-1)$-qubit Pauli LIM.
    Expanding the semantics of $v_L$ and $v_M$, we obtain,
        \begin{align}
        \label{eq:canonicity-equation-1}
       \lambda \Ptop \otimes \Prest (\ket{0} \otimes A_L \ket{v_L^0} + \ket{1} \otimes B_L \ket{v_L^1})
        =
        \ket{0} \otimes A_M \ket{v_M^0} + \ket{1} \otimes B_M \ket{v_M^1}
        .
        \end{align}    
    We distinguish two cases from here on: where $\Ptop \in \{\id, Z\}$ or $\Ptop \in \{X,  Y\}$.
    
    \textbf{Case $\boldsymbol{\Ptop = I,Z}$.}
    If $\Ptop = \diag z$ for $z \in \{1, -1\}$,  then
    \autoref{eq:canonicity-equation-1} gives:
        \begin{align}
        \label{eq:canonicity-equation-2}
            \lambda \Prest A_L \ket{v_L^0} = A_M \ket{v_M^0}
            \qquad\textnormal{and}\qquad
            z\lambda \Prest B_L \ket{v_L^1} = B_M \ket{v_M^1}
        \end{align}
            By low factoring, we have $A_L = A_M = \id$, so we obtain
$ \lambda  \Prest \ket{v_L^0} = \ket{v_M^0}$.      
        Hence $\ket{v_L^0}$ is isomorphic with $\ket{v_M^0}$, so by the induction hypothesis, we have $v_L^0 = v_M^0$. 
         We now show that also $v_L = v_M$ by considering two cases.
\begin{description}
        \item[$B_L \neq 0$ and $B_M \neq 0$:] then $z\lambda \Prest B_L\ket{v_L^1}=B_M\ket{v_M^1}$, so the nodes $v_L^1$ and $v_M^1$ represent isomorphic states, so by the induction hypothesis we have $v_L^1=v_M^1$.
        We already noticed by the low factoring rule that $v_L$ and $v_M$ have $\mathbb I$ as low edge label.
        By the high edge rule, their high edge labels are $\highlabel(v_L)$ and $\highlabel(v_M)$, and since the nodes $v_L$ and $v_M$ are semi-reduced and $\ket{v_L} \simeq \ket{v_M}$, we have $\highlabel(v_M) = \highlabel(v_L)$ by definition of $\highlabel$.

        \item[$B_L = 0$ or $B_M= 0$:] In case $B_L = 0$,
        we see from \autoref{eq:canonicity-equation-2} that $0=B_M\ket{v_M^1}$.
        Since the state vector $\ket{v_M^1}\ne 0$ by the observation that a reduced node does not represent the zero vector, it follows that $B_M=0$.
        Otherwise, if $B_M=0$, then \autoref{eq:canonicity-equation-2} yields $z\lambda\Prest B_L\ket{v_L^1}=0$.
        We have $z \lambda\ne 0$, $\Prest \ne 0$ by definition, and we observed $\ket{v_L^1}\ne 0$ above.
        Therefore $B_L=0$. In both cases, $B_L=B_M$.
\end{description}
        We conclude that in both cases $v_L$ and $v_M$ have the same children and the same edge labels, so they are identical by the merge rule.

        \textbf{Case $\boldsymbol{\Ptop = X, Y}$.}
    If $\Ptop = \begin{smallmat}0 & z^* \\ z & 0\end{smallmat}$ for $z \in \{1, i\}$, then
    \autoref{eq:canonicity-equation-1} gives:
        \[
            \lambda z\Prest A_L \ket{v_L^0} = B_M \ket{v_M^1}
            \qquad\textnormal{and}\qquad
            \lambda z^*\Prest B_L \ket{v_L^1} = A_M \ket{v_M^0}
            .
        \]
        By low factoring, $A_L = A_M = \id$, so we obtain 
        $z\lambda  \Prest \ket{v_L^0} = B_M \ket{v_M^1}$
        and
        $\lambda z^*\Prest B_L \ket{v_L^1} = \ket{v_M^0}$.
To show that $v_L = v_M$, we consider two cases.       
\begin{description}
        \item[$B_L \neq 0$ and $B_M \neq 0$:] we find $\ket{v_L^0} \simeq \ket{v_M^1}$ and $\ket{v_L^1} \simeq \ket{v_M^0}$, so by the induction hypothesis, $v_L^0= v_M^1$ and $v_L^1= v_M^0$.
        By low precedence, it must be that $v_L^1 = v_M^1 = v_L^0 = v_M^0$.
        Now use high determinism to infer that $B_L = B_M$ as in the ${\Ptop = I,Z}$ case.

        \item[$B_L = 0$ or $B_M = 0$:]
       This case leads to a contradiction and thus cannot occur.
        $B_L$ cannot be zero, because then $\ket{v_M^0}$ is the all-zero vector, which we excluded.
        The other case: if $B_M = 0$, then it must be that $\lambda z\Prest A_L\ket{v_L^0}$ is zero.
        Since $\lambda z \Prest\ne 0$ and $A_L=\id$, it follows that $\ket{v_L^0}$ is the all-zero vector, which is again excluded.
\end{description}
 We conclude that $v_L$ and $v_M$ have the same children and the same edge labels  for all choices of $\Ptop$, so they are identical by the merge rule.
\end{proof}

\subsection{The \makeedge subroutine: Maintaining canonicity during simulation}
\label{sec:makeedge}

To construct new nodes and edges, our algorithms use the \makeedge subroutine (\autoref{alg:make-edge}), as discussed in \autoref{sec:def-reduced-limdd}.
\makeedge produces a reduced parent node (with root edge) given two reduced children, so that the \limdd representation becomes canonical.
Here we give the algorithm for \makeedge and show that it runs in time $O(n^3)$ (assuming the input nodes are reduced).

\label{sec:makeedge-basic-reduction}

The \makeedge subroutine distinguishes two cases, depending on whether both children are non-zero vectors, which both largely follow the discussion below \autoref{def:reduced-limdd}.
It works as follows:
\begin{itemize}
    \item 
First it ensures low precedence, switching $e_0$ and $e_1$ if necessary at
\autoref{l:swap}. This is also done if $e_0$'s label $A$ is $0$ to allow for low factoring (avoiding divide by zero).
    \item 
Low factoring, i.e., dividing
out the LIM $A$, placing it on the root node,
        is visualized in \autoref{fig:reduced1} for the cases $e_1 = 0 / e_1\neq 0$,
and done in the algorithm at \autoref{l:low1},\ref{l:low2} / \ref{l:low3},\ref{l:low4}.
\item 
The zero edges rule is enforced in the $B=0$ branch by taking $v_1 := v_0$.
\item 
The canonical high label $\highlim$ is computed by \textsc{GetLabels},
    discussed below, for the semi-reduced node $\lnode[w] {\id}{v_0}{\hat A}{v_1}$ with $v_0 \neq v_1$.
    With the resulting high label, it now satisfies the high determinism rule of \autoref{def:reduced-limdd} with $\highlabel(w)= \highlim$.
\item
Finally, we merge nodes by creating an entry $(v_0,\highlim,v_1)$ in 
a table called the \concept{unique table}~\cite{brace1991efficient} at \autoref{algline:find-v-in-unique}. 
\end{itemize}

All steps except for \textsc{GetLabels} have complexity $O(1)$ or $O(n)$ (for checking low precedence, we use the nodes' order in the unique table).
The algorithm \textsc{GetLabels}, which we sketch below in \autoref{sec:choose-canonical-isomorphism-pauli} and fully detail in \autoref{app:canonical-high-label}, has runtime $O(n^3)$ if both input nodes are reduced, yielding an overall complexity $O(n^3)$.

\begin{algorithm}
	\begin{algorithmic}[1]
        \Procedure{MakeEdge}{\Edge $\ledge[e_0]{A}{v_0}$, $\ledge[e_1]{B}{v_1}$, \textbf{with}  $v_0, v_1$ reduced, $A \neq 0$ \textbf{or} $B \neq 0$}
            \If{$v_0\not\beforeq v_1$ \textbf{or} $A=0$} \Comment{Enforce \textbf{low precedence} and enable \textbf{factoring}}
                \State \Return\label{l:swap}
                        $(X \otimes \id^{\otimes n})\cdot \text{MakeEdge}(e_1, e_0)$
			\EndIf
            \If{$B = 0$}\algstrut[1]
            \State $v_1 := v_0$     \Comment{Enforce  \textbf{zero edges}}
            \State $v := \lnode{\id^{\otimes n}}{v_0}{0}{v_0}$\label{l:low1}
            \Comment{Enforce \textbf{low factoring}}
            \State $\rootlim := \id \otimes A$   \Comment{$\rootlim \ket v = \ket 0 \otimes 
                                        A \ket{v_0} + \ket 1 \otimes  B \ket{v_1}$}\label{l:low2}
            \Else
            \State $\hat A := A^{-1}B$ \Comment{Enforce \textbf{low factoring}}\label{l:low3}
            \State $\highlim, \rootlim:=\textsc{GetLabels}(\hat A,v_0,v_1)$
            \label{algline:makeedge-get-labels}
            \Comment{Enforce \textbf{high determinism}}
            \State $v := \lnode{\id^{\otimes n}}{v_0}{\highlim}{v_1}$ 
             \Comment{$\rootlim \ket v = \ket 0 \otimes 
                                         \ket{v_0} + \ket 1 \otimes A^{-1} B \ket{v_1}$}
                \label{l:low4}
            \State $\rootlim := (\id \otimes A) \rootlim $ 
             \Comment{$(\id \otimes A)\rootlim \ket v = \ket 0 \otimes  A \ket{v_0} + \ket 1 \otimes  B \ket{v_1}$}
            \EndIf
			\Stateh $v_{\text{r}}:=$ Find or create unique table entry $\unique[v] = (v_0, \highlim, v_1)$       
			\label{algline:find-v-in-unique}
            \Comment{Enforce \textbf{merge}}
            \State \Return $\ledge{\rootlim}{v_{\text{r}}}$
		\EndProcedure
	\end{algorithmic}
	\caption{
        Algorithm \makeedge takes two root edges to (already reduced) nodes $v_0,v_1$, the children of a new node, and returns a reduced node with root edge.
	It assumes that %
	$\index(v_0) = \index(v_1) = n$.
	We indicate which lines of code are responsible for which reduction rule in \autoref{def:reduced-limdd}.
}
	\label{alg:make-edge}
\end{algorithm}

\subsubsection{Choosing a canonical high-edge label}
\label{sec:choose-canonical-isomorphism-pauli}

In order to choose the canonical high edge label of node $v$, the \makeedge algorithm calls \textsc{GetLabels} (\autoref{algline:makeedge-get-labels} of \autoref{alg:make-edge}).
The function \textsc{GetLabels} returns a uniquely chosen LIM \highlim among all possible high-edge labels which yield LIMDDs representing states that are Pauli-isomorphic to $\ket{v}$.
We sketch the algorithm for \textsc{GetLabels} here and provide the algorithm in full detail in \autoref{app:canonical-high-label} (\autoref{alg:find-canonical-edges}).
First, we characterize the eligible high-edge labels.
That is, given a semi-reduced node $\lnode[v]{\id}{v_0}{\hat{A}}{v_1}$, we characterize all $C$ such that the node $\lnode{\id}{v_0}{C}{v_1}$ is isomorphic to $\lnode[v]{\id}{v_0}{\hat{A}}{v_1}$.
Our characterization shows that, modulo some complex factor, the eligible labels $C$ are of the form
\begin{equation}
    \label{eq:double-coset-eligible-high-labels}
    C \propto g_0 \cdot \hat{A} \cdot g_1, \quad \textnormal{for $g_0 \in \Stab(\ket{v_0}), g_1 \in \Stab(\ket{v_1})$}
\end{equation}
where $\Stab(\ket{v_0})$ and $\Stab(\ket{v_1})$ are the stabilizer subgroups of $\ket{v_0}$ and $\ket{v_1}$, i.e., the already reduced children of our input node $v$.
Note that the set of eligible high-edge labels might be exponentially large in the number of qubits.
Fortunately, eq.~\eqref{eq:double-coset-eligible-high-labels} shows that this set has a polynomial-size description by storing only the generators of the stabilizer subgroups.

Our algorithm chooses the lexicographically smallest eligible label, i.e., the smallest $C$ of the form $C\propto g_0\hat{A} g_1$ (the definition of `lexicographically smallest' is given in \autoref{app:prelims-linear-algebra}).
To this end, we use two subroutines: (1) an algorithm which finds (a generating set of) the stabilizer group $\Stab(\ket{v})$ of a \limdd node $v$; and (2) an algorithm that uses these stabilizer subgroups of the children nodes to choose a unique representative of the eligible-high-label set from eq.~\eqref{eq:double-coset-eligible-high-labels}.

For (1), we use an algorithm which recurses on the children nodes.
First, we note that, if the Pauli LIM $A$ stabilizes both children, then $\mathbb I\otimes A$ stabilizes the parent node.
Therefore, we compute (a generating set for) the intersection of the children's stabilizer groups.
Second, our method finds out whether the parent node has stabilizers of the form $P_n\otimes A$ for $P_n\in \{X,Y,Z\}$.
This requires us to decide whether certain cosets of the children's stabilizer groups are empty.
These groups are relatively simple, since, modulo phase, they are isomorphic to a binary vector space, and cosets are hyperplanes.
We can therefore rely in large part on existing algorithms for linear algebra in vector spaces.
The difficult part lies in dealing with the non-abelian aspects of the Pauli group.
We provide the full algorithm, which is efficient, also in~\autoref{app:canonical-high-label}.

Our algorithm for (2) applies a variant of Gauss-Jordan elimination to the generating sets of $\Stab(\ket{v_0})$ and $\Stab(\ket{v_1})$ to choose $g_0$ and $g_1$ in eq.~\eqref{eq:double-coset-eligible-high-labels} which, when multiplied with $\hat{A}$ as in eq.~\eqref{eq:double-coset-eligible-high-labels}, yield the smallest possible high label $C$.
(We recall that Gauss-Jordan elimination, a standard linear-algebra technique, is applicable here because the stabilizer groups are group isomorphic to binary vector spaces, see also \autoref{app:prelims-linear-algebra}).
We explain the full algorithm in \autoref{app:canonical-high-label}.

\subsubsection{Checking whether two \limdds are Pauli-equivalent}
\label{sec:pauli-equivalence-checking}

To check whether two states represented as \limdds are Pauli-equivalent, it suffices to check whether they have the same root node.
Namely, due to canonicity, and in particular the Merge rule (in \autoref{def:reduced-limdd}), there is a unique \limdd representing a quantum state up to phase and local Pauli operators.

\section{Related work}
\label{sec:related-work}
We mention related work on classical simulation formalisms and decision diagrams other than~\qmdd.

The Affine Algebraic Decision Diagram, introduced by Tafertshofer and Pedam \cite{tafertshofer1997factored}, and by Sanner and McAllister \cite{sanner2005AffineADDs}, is akin to a \qmdd except that its edges are labeled with a pair of real numbers $(a,b)$, so that an edge $\ledge{(a,b)}{v}$ represents the state vector $a\ket{v}+b\ket{+}^{\otimes n}$ (i.e., here $b$ is added to each element of the vector $a\ket{v}$).
To the best of our knowledge, this diagram has not been applied to quantum computing.

Context-Free-Language Ordered Binary Decision Diagrams (CFLOBDDs)~\cite{sistla2023cflobdds,quasimodo} extend BDDs with insights from visibly pushdown automata~\cite{vpda}.
An extension of CFLOBDD to the complex domain~\cite{sistla2023weighted} shows good performance for various simulation of quantum computing benchmarks.
Sentential Decision Diagrams \cite{darwiche2011sdd} generalize BDDs by replacing their total variable order with a variable tree (vtree).
Although Kisa et al.~\cite{prsdd} introduced an SDD which represents probability distributions, SDDs have not yet been used to simulate quantum computing, to the best of our knowledge.
The Variable-Shift SDD (VS-SDD)~\cite{nakamura_et_al:LIPIcs:2020:12096} improves on the SDD by merging isomorphic vtree nodes.
We remark that CFLOBDDs are similar to VS-SDD with a balanced vtree.

G\"unther and Drechsler introduced a BDD variant \cite{gunther1998bddLinearTransformation} which, in \limdd terminology, has a label on the root node only.
To be precise, this diagram's root edge is labeled with an invertible matrix $A\in \mathbb F_2^{n\times n}$.
If the root node represents the function $r\colon \bool^{n}\to \bool$, then the diagram represents the function $f(\vec x)=r(A\cdot \vec x)$.
(This concepts extends trivially to the domain of pseudo-Boolean functions, by replacing the BDD with an \add.)
In contrast, \limdds allow a label on every edge in the diagram, not only the root edge.
We show that this is essential to capture stabilizer states.

A multilinear arithmetic formula is a formula over $+,\times$ which computes a polynomial in which no variable appears raised to a higher power than $1$.
Aaronson showed that some stabilizer states require superpolynomial-size multilinear arithmetic formulas~\cite{briegel2000persistent,aaronson2004multilinear}.

\section{Discussion \label{sec:discussion}}

We have introduced \limdd, a novel decision diagram-based method to simulate quantum circuits, which enables polynomial-size representation of a strict superset of stabilizer states and the states represented by polynomially large \qmdds.
To prove this strict inclusion, we have shown the first lower bounds on the size of \qmdds: they require exponential size for certain families of stabilizer states.
Our results show that these states are thus hard for \qmdds.
We also give the first analytical comparison between simulation based on decision diagrams, and matrix product states, and the \cliffordt simulator.

\limdds achieve a more succinct representation than \qmdds by representing states up to 
local invertible maps which uses single-qubit (i.e., local) operations from a group $G$.
We have investigated the choices $G=\pauli$, $G=\braket{Z}$ and $G=\braket{X}$,
and found that any choice suffices for an exponential advantage over \qmdds;
notably, the choice $G=\pauli$ allows us to succinctly represent any stabilizer state.
Furthermore, we showed how to simulate arbitrary quantum circuits, encoded as Pauli-\limdds.
The resulting algorithms for simulating quantum circuits are exponentially faster than for \qmdds in the best case, and never more than a polynomial factor slower.
In the case of Clifford circuits, the simulation by \limdds is in polynomial time (in contrast to \qmdds).

We have shown that Pauli-\limdds can efficiently simulate a circuit family outputting the $W$ states, in contrast to the \cliffordt simulator which requires exponential time to do so (assuming the widely believed ETH), even when allowing for preprocessing of the circuit with a $T$-count optimizer.

Since we know from experience that implementing a decision diagram framework is a major endeavor, we leave an implementation of the Pauli-\limdd, in order to observe its runtimes in practice on relevant quantum circuits, to future work.
We emphasize that from the perspective of algorithm design, we have laid all the groundwork for such an implementation, including the key ingredient for the efficiency of many operations for existing decision diagrams: the existence of a unique canonical representative of the represented function, combined with a tractable MakeEdge algorithm to find it.

Regarding extensions of the \limdd data structure, an obvious next step is to investigate other choices of $G$.
Of interest are both the representational capabilities of such diagrams
(do they represent interesting states?), and the algorithmic capabilities
(can we still find efficient algorithms which make use of these diagrams?).
In this vein, an important question is what the relationship is between \glimdds
(for various choices of $G$) and existing formalisms for the classical simulation of quantum circuits, such as those based on match gates~\cite{terhal2001classical,jozsa2008matchgates,hebenstreit2020computational} and tensor networks~\cite{orus2014practical,hong2020tensor}.
It would also be interesting to compare \limdds to graphical calculi such as the ZX calculus \cite{coecke2011interactingZXAlgebra}, 
following similar work for \qmdds~\cite{vilmart2021quantum}.

Lastly, we note that the current definition of \limdd imposes a strict total order
over the qubits along every path from root to leaf.
It is known that the chosen order can greatly influence the size of the DD~\cite{rudell1993dynamic,wegener2000branching},
making it interesting to investigate variants of \limdds with a flexible ordering, for example taking inspiration from the Sentential Decision Diagram \cite{darwiche2011sdd,kisa2014probabilistic}.

\section{Acknowledgements}
\label{sec:acknowledgements}

We thank the anonymous reviewers at Quantum for their valuable suggestions, which helped to greatly improve the presentation of the manuscript.
We thank Dan Browne for help with establishing stabilizer ranks.
We thank Marie Anastacio, Jonas Helsen, Yash Patel and Matthijs Rijlaarsdam for their feedback on early versions of the manuscript.
We thank Patrick Emonts and Adri\`an P\'erez-Salinas for discussions on MPS, and Kenneth Goodenough for useful discussions in general.
The second author acknowledges the QIA project (funded by European Union's Horizon 2020, Grant Agreement No. 820445).
The third and fourth author are funded by the Netherlands Organization for Scientific Research (NWO/OCW), as part of the Quantum Software Consortium program (Project No. 024.003.037/3368).
The last author is funded by the research program VENI with project number 639.021.649 of the Netherlands Organization for Scientific Research (NWO).

\bibliographystyle{quantum}
\bibliography{lit}

\onecolumn
\appendix

\section{Linear-algebra algorithms for Pauli operators}
\label{app:prelims-linear-algebra}

In~\autoref{sec:preliminaries}, we defined the stabilizer group for an $n$-qubit state $\ket{\phi}$ as the group of Pauli operators $A\in \Pauli_n$ which stabilize $\ket{\phi}$, i.e. $A\ket{\phi} = \ket{\phi}$.
Here, we explain existing efficient algorithms for solving various tasks regarding stabilizer groups (whose elements commute with each other).
We also outline how the algorithms can be extended and altered to work for general {\paulilim}s, which do not necessarily commute.
For sake of clarity, in the explanation below we first ignore the scalar $\lambda$ of a {\paulilim} or \pauli element $\lambda P$. At the end, we explain how the scalars can be taken into account when we use these algorithms as subroutine in \limdd operations.

Any $n$-qubit Pauli string can (modulo factor $\in \{\pm 1, \pm i\}$) be written as $(X^{x_n} Z^{z_n}) \otimes \dots \otimes (X^{x_1} Z^{z_1})$ for bits $x_j, z_j, 1 \leq j \leq n$.
We can therefore write an $n$-qubit Pauli string $P$ as a length-$2n$ binary vector as follows~\cite{aaronson2008improved},
\[
    (\underbrace{x_n, x_{n-1}, \dots x_1}_{\textnormal{X block}} | \underbrace{z_n, z_{n-1}, \dots, z_1}_{\textnormal{Z block}})
    ,
\] 
where we added the horizontal bar ($|$) only to guide the eye.
We will refer to such vectors as \emph{check vectors}.
For example, we have $X \sim (1, 0)$ and $Z \otimes Y \sim (0, 1 | 1, 1)$ .
This equivalence induces an ordering on Pauli strings following the lexicographic ordering on bit strings.
For example, $X<Y$ because $(1|0) < (1|1)$ and $Z\otimes \id < Z \otimes X$ because $(0 0 | 1 0) < (0 1 | 1 0)$.

A set of $k$ Pauli strings thus can be written as $2n\times k$ binary matrix, often called \emph{check matrix}, as the following example shows.
\[
\begin{pmatrix}
	X &\otimes& X &\otimes& X\\
	\id &\otimes& Z &\otimes& Y
\end{pmatrix}
\sim
\begin{pmatrix}
	1& 1& 1& |& 0& 0& 0\\
	0& 0& 1& |& 0& 1& 1
\end{pmatrix}
.
\]
Furthermore, if $P, Q$ are Pauli strings corresponding to binary vectors $(\vec{x}^P, \vec{z}^P)$ and $(\vec{x}^Q, \vec{z}^Q)$, then 
\[
P \cdot Q \propto
\bigotimes_{j=1}^n
\left(X^{x^P_j} Z^{z^P_j}\right) \left(X^{x^Q_j} Z^{z^Q_j}\right) 
=
\bigotimes_{j=1}^n
\left(X^{x^P_j \oplus x^Q_j } Z^{z^P_j \oplus z^Q_j} \right)
\]
and therefore the group of $n$-qubit Pauli strings with multiplication (disregarding factors) is group isomorphic to the vector space $\{0, 1\}^{2n}$ (i.e., $\mathbb F_2^{2n}$) with bitwise addition $\oplus$ (i.e., exclusive or; `xor').
Consequently, many efficient algorithms for linear-algebra problems carry over to sets of Pauli strings.
In particular, if $G = \{g_1, \dots, g_k\}$ are length$-2n$ binary vectors (/ $n$-qubit Pauli strings) with $k\leq n$, then we can efficiently perform the following operations.
\begin{description}
    \item[\emph{RREF:}] bring $G$ into a reduced-row echelon form (RREF) using Gauss-Jordan elimination (both are standard linear algebra notions) where each row in the check matrix has strictly more leading zeroes than the row above.
        The RREF is achievable by $O(k^2)$ row additions (/~multiplications modulo factor) and thus $O(k^2 \cdot n)$ time (see \cite{berg2020circuit} for a similar algorithm).
        In the RREF, the first $1$ after the leading zeroes in a row is called a `pivot'.
    \item[\emph{Construct Minimal-size Generator Set}] convert $G$ to a (potentially smaller) set $G'$ by performing the RREF procedure and discarding resulting all-zero rows.
    It holds that $\braket{G}=\braket{G'}$, i.e., these sets generate the same group modulo phase.
    \item[\emph{Membership:}] determining whether a given a vector (/~Pauli string) $h$ has a decomposition in elements of $G$.
    	This can be done by obtaining minimal-size generating sets $H_1,H_2$ for the sets $G$ and $G\cup \{h\}$, respectively.
    	Then the generating sets have the same number of elements (i.e., rows) if and only if $h\in \braket{G}$; otherwise, if $h\not\in \braket{G}$, it holds that $|H_2|=|H_1|+1$.
    \item[\emph{Intersection:}] determine all Pauli strings which, modulo a factor, are contained in both $G_A$ and $G_B$, where $G_A, G_B$ are generator sets for $n$-qubit stabilizer subgroups.
    More specifically, we obtain the generator set of this group, i.e., we obtain a set $G_C$ such that $\braket{G_C}=\braket{G_A}\cap \braket{G_B}$.
        This can be achieved using the Zassenhaus algorithm \cite{LUKS1997335} for computing the intersection of two subspaces of a vector space, in time $O(n^3)$.
    \item[\emph{Division remainder:}] given a vector $h$  (/~Pauli string $h$), determine \mbox{$ h^{\textnormal{rem}} := \min_{g\in \langle G\rangle} \{ g\oplus h\}$} (minimum in the lexicographic ordering).
        We do so in the check matrix picture by bringing $G$ into RREF, and then making the check vector of $h$ contain as many zeroes as possible by adding rows from $G$:
	\begin{algorithmic}[1]
        \For{column index $j=1$ to $2n$}
        \If{$h_j = 1$ and $G$ has a row $g_i$ with its pivot at position $j$}
		$h := h \oplus g_i$
        \EndIf
        \EndFor
	\end{algorithmic}
        The resulting  $h$ is $h^{\textnormal{rem}}$.
This algorithm's runtime is dominated by the RREF step; $O(n^3)$.
\end{description}

To include the scalar into the representation, we remark that Pauli LIMs that appear as labels on diagrams may have $\lambda\in\mathbb C$, i.e., any complex number is allowed.  %
Therefore, to store LIMs, we use a minor extension to the check vector form introduced above, in order to also include the phase.
Specifically, the phase is stored using two real numbers, by writing $\lambda=r\cdot e^{i\theta}$ with $r \in \mathbb{R}_{> 0}$ and $\theta\in [0, 2\pi)$.
Consequently, the check vector has $2n+2$ entries, where the last entries store $r$ and $\theta$, e.g.:
\[
    \begin{pmatrix}
        3X &\otimes& X &\otimes& X\\
        -\frac{1}{2} i\id &\otimes& Z &\otimes& Y
    \end{pmatrix}
    \sim
    \begin{pmatrix}
        1& 1& 1& |& 0& 0& 0& |&3 & 0\\
        0& 0& 1& |& 0& 1& 1& |&\frac{1}{2} & \frac{3\pi}{2}
    \end{pmatrix}
\]
where we used $3 = 3\cdot e^{i\cdot 0}$ and $-\frac{1}{2}i = \frac{1}{2} \cdot e^{3\pi i/2}$.
This extended check vector also easily allows a total ordering, namely, we simply use the ordering on real numbers for $r$ and $\theta$.
For example, $(1, 1, | 0, 0 | 2, \frac{1}{2} ) < (1, 1 | 1, 0 | 3, 0)$.
Let us stress that the factor encoding $(r, \theta)$ is less significant than the Pauli string encoding $(x_n, \dots, x_1 | z_n, \dots, z_1)$.
As a consequence, we can greedily determine the minimum of two Pauli operators, by reading their check vectors from left to right.

Finally, we emphasize that the algorithms above rely on bitwise xor-ing, which is a commutative operation.
Since conventional (i.e., factor-respecting) multiplication of Pauli operators is not commutative, the algorithms above are not straightforwardly applicable to arbitrary $\paulilim_n$ input.
(When the input consist of pairwise commuting Pauli operators, such as stabilizer subgroups \cite{aaronson2008improved}, the algorithms can be made to work by adjusting row addition to keep track of the scalar.)
Fortunately, since Pauli strings either commute or anti-commute, row addition may only yield factors up to the $\pm$ sign, not the resulting Pauli strings.
This feature, combined with the stipulated order assigning least significance to the factor,
enables us to invoke the algorithms above as subroutine. %
We do so in~\autoref{sec:choose-canonical-isomorphism-pauli} and \autoref{sec:pauli-isomorphism-detection}.

\section{Proof that cluster states and coset states need exponentially-large QMDDs}
\label{sec:graph-state-lower-bound}

In this appendix, we show that \qmdds which represent both clusters states, and coset states, are exponentially large in the worst case
(respectively, \autoref{thm:graph-state-qmdd-lower-bound} and \autoref{thm:random-xor-state-hard-for-qmdd}).
On the other hand, in \autoref{sec:graph-states-limdds}, we will show that these states can be represented using only $\mathcal O(n)$ nodes by $\braket{X}$-\limdds, showing that they are exponentially more succinct than \qmdds.
We first fix notation and definitions, after which we prove 
the theorem using two lemmas.

Let $G$ be an undirected graph with vertices $V_G=\{v_1, \dots, v_n\}$ and edge set $E_G \subseteq V_G\times V_G$.
For a subset of vertices $S\subseteq V_G$, the $S$-induced subgraph of $G$ has vertices $S$ and edge set $(S \times S) \cap E$.
Given $G$, its graph state $\ket{G}$ is expressed as
\begin{align}
	\label{eq:graph-state-definition}
    \ket{G} = \sum_{\vec x\in \set{0,1}^n} (-1)^{f_G(\vec x) } \ket{\vec x}
\end{align}
where $f_G(\vec{x})$ is the number of edges in the $S$-induced subgraph of $G$.

For a function $f: \{0, 1\}^n \rightarrow \mathbb{C}$ and bit string $\vec a=a_1\cdots a_k\in\{0, 1\}^k$, we denote by $f_{\vec{a}}$ the subfunction of $f$ restricted to $\vec a$:
\begin{align}
f_{\vec a}(x_{k+1},\ldots, x_n) := f(a_1,\ldots, a_k,x_{k+1},\ldots, x_n)
\end{align}
We also say that $f_{\vec a}$ is a subfunction of $f$ of \emph{order} $|\vec a|=k$.

We will also need the notions of boundary and strong matching.

\begin{definition}[Boundary]
	For a set $S\subseteq V_G$ of vertices in $G$, the \emph{boundary} of $S$ is the set of vertices in $S$ adjacent to a vertex outside of $S$.
\end{definition}

\begin{definition}[Strong Matching]\label{def:smatch}
	Let $G=(V,E)$ be an undirected graph. A \emph{strong matching} is a
	subset of edges $M \subseteq E$ that do not share any vertices (i.e., it is a matching)
	and no two edges of $M$ are incident to the same edge of $G$, i.e.,
	an edge in $E \setminus M$. Alternatively, a strong matching is a matching $M$ s.t. $G[V(M)] = M$.
	We say that $M$ is an $(S,T)$-strong matching for two sets of vertices $S,T\subset V$ if $M\subseteq S\times T$.
	For a strong matching $M$ and a vertex $v\in V(M)$, we let $M(v)$ denote the unique vertex to which $v$ is matched by $M$.
\end{definition}

Using these definitions and notation, we prove \autoref{thm:graph-state-qmdd-lower-bound}.

\thmgraphqmddlower*
\begin{proof}
Let $G=\text{lattice}(n,n)$ be the undirected graph of the $n\times n$ lattice,
with vertex set $V=\{v_1,\ldots, v_{n^2}\}$.
Let $\sigma=v_1v_2\cdots v_{n^2}$ be a variable order, and let $S=\{v_1,v_2,\ldots, v_{\frac{1}{2}n^2}\}\subset V$ be the first $\half n^2$ vertices in this order.

The proof proceeds broadly as follows.
First, in \autoref{thm:strong-matchings-yield-subfunctions}, we show that any $(S,\overline S)$-strong matching $M$ effects $2^{|M|}$ different subfunctions of $f_G$.
Second, \autoref{thm:large-strong-matching} shows that the lattice contains a large $(S,\overline S)$-strong matching for any choice of $S$.
    Put together, this will prove the lower bound on the number of QMDD nodes as in \autoref{thm:graph-state-qmdd-lower-bound} by the fact that a QMDD for the cluster state $G$ has a node per unique subfunction of the function $f_G$.
	\autoref{fig:strong-matching-in-grid} illustrates this setup for the $5\times 5$ lattice.

\begin{lemma}
	\label{thm:strong-matchings-yield-subfunctions}
	Let $M$ be a non-empty $(S,\overline S)$-strong matching for the vertex set $S$ chosen above.
	If $\sigma=v_1v_2\cdots v_{n^2}$ is a variable order where all vertices in $S$ appear before all vertices in $\overline S$, then $f_G(x_1,\ldots, x_{n^2})$ has $2^{|M|}$ different subfunctions of order $|S|$.
\end{lemma}
\begin{proof}
	Let $S_M:=S\cap V(M)$ and $\overline S_M:=\overline S\cap M$ be the sets of vertices that are involved in the strong matching.
Write $\chi(x_1, \dots, x_n)$ for the indicator function for vertices: $\chi(x_1, \dots, x_n) := \set{v_i \mid x_i=1, i \in [n]} $.
	Choose two different subsets $A,B\subseteq S_M$ and let $\vec{a}=\chi^{-1}(A)$ and $\vec{b}=\chi^{-1}(B)$ be the corresponding length-$|S|$ bit strings.
	These two strings induce the two subfunctions $f_{G,\vec{a}}$ and $f_{G,\vec{b}}$.
	We will show that these subfunctions differ in at least one point.
	
	First, if $f_{G,\vec{a}}(0,\ldots, 0)\ne f_{G,\vec{b}}(0,\ldots, 0)$, then we are done.
	Otherwise, take a vertex $s\in A\oplus B$
    and say w.l.o.g. that $s\in A\setminus B$.
	Let $t=M(s)$ be its partner in the strong matching.
	Then we have, $|E[A\cup \{t\}]| = |E[A]|+1$ but $|E[B\cup \{t\}]|=|E[B]|$.
	Therefore we have 
	\begin{align}
	f_{G,\vec a}(0,\ldots, 0, x_t=0, 0,\ldots, 0) ~~\ne~~ &f_{G,\vec a}(0,\ldots,0,x_t=1,0,\ldots, 0) \\
	f_{G,\vec b}(0,\ldots, 0, x_t=0, 0,\ldots, 0) ~~=~~ & f_{G,\vec b}(0,\ldots, 0,x_t=1,0,\ldots, 0)
	\end{align}
	We see that each subset of $S_M$ corresponds to a different subfunction of $f_G$. Since there are $2^{|M|}$ subsets of $M$, $f_G$ has at least that many subfunctions.
\end{proof}

We now show that the $n\times n$ lattice contains a large enough strong matching.

\begin{lemma}
	\label{thm:large-strong-matching}
	Let $S=\{v_1,\ldots, v_{\frac{1}{2} n^2}\}$ be a set of $\half n^2$ vertices of the $n\times n$ lattice, as above.
	Then the graph contains a $(S,\overline S)$-strong matching of size at least $\floor{\frac{1}{12}n}$.
\end{lemma}
\begin{proof}
	Consider the boundary $B_S$ of $S$.
	This set contains at least $n/3$ vertices, by Theorem 11 in \cite{lipton1979generalized}.
	Each vertex of the boundary of $S$ has degree at most $4$. 
    It follows that there is a set of $\left\lfloor \frac{1}{4}|B_S|\right\rfloor$ vertices which share no neighbors.
	In particular, there is a set of $\left\lfloor \frac{1}{4}|B_S| \right\rfloor\geq \floor{\frac{1}{12}n}$ vertices in $B_S$ which share no neighbors in $\overline S$.
\end{proof}
Put together, every choice of half the vertices in the lattice yields a set with a boundary of at least $n/3$ nodes, which yields a strong matching of at least $\floor{\frac{1}{12}n}$ edges, which shows that $f_G$ has at least $2^{\floor{\frac{1}{12}n}}$ subfunctions of order $\frac{1}{2}n^2$.
\end{proof}

\begin{figure}[h!]
\centering
\includegraphics[width=0.45\textwidth]{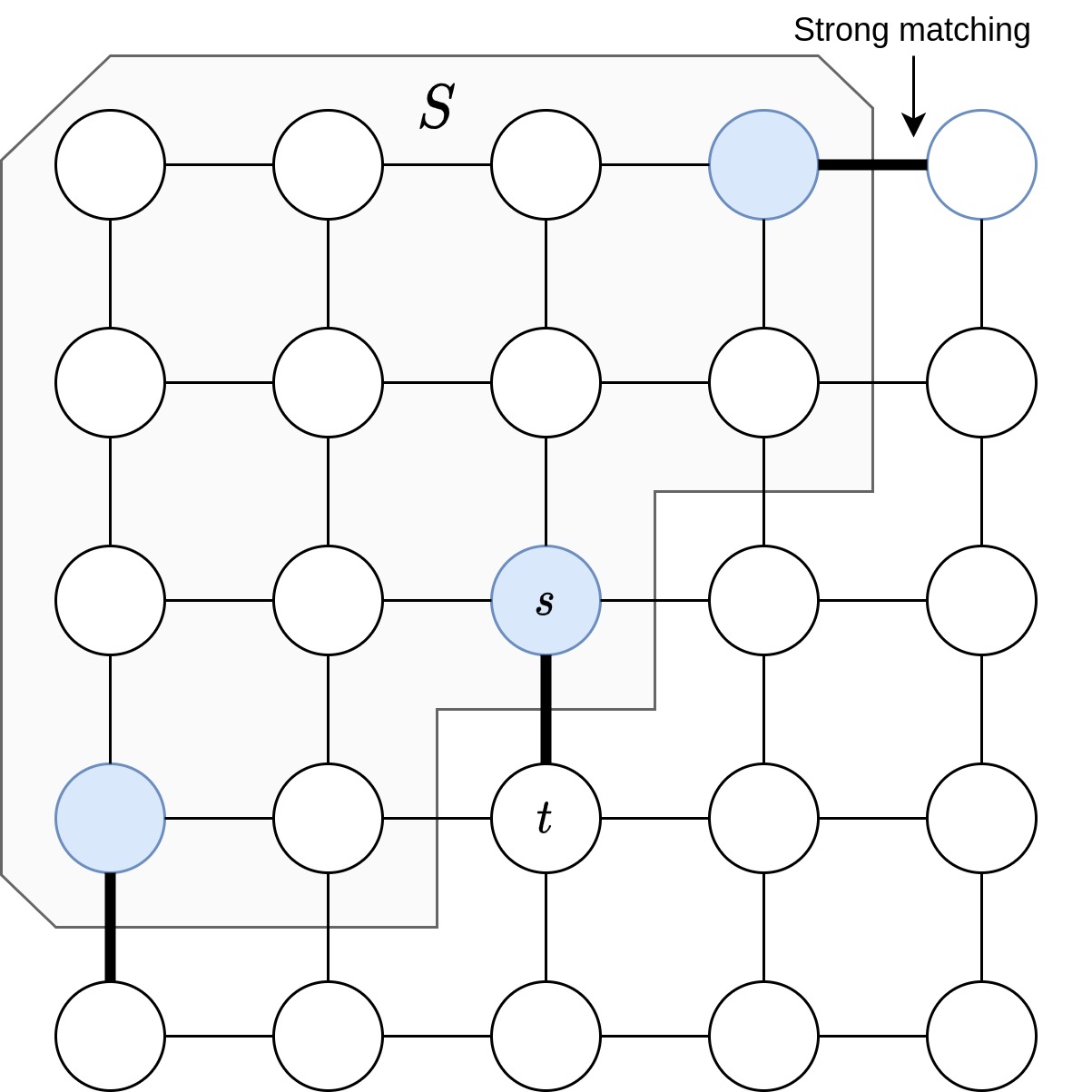}
\caption{
    The $5\times 5$ lattice, partitioned in a vertex set $S$ and its complement $\overline{S}$.
    A strong matching between $S$ and $\overline{S}$ is indicated by black edges.
}
\label{fig:strong-matching-in-grid}
\end{figure}

\paragraph{Proof that coset states need exponentially large QMDDs.}

We now show that \qmdds which represent coset states are exponentially large in the worst case.
We will use the following result by \v{D}uri\v{s} et al. on binary decision diagrams (\bdds), which are \qmdds with codomain $\{0, 1\}$.
This result concerns vector spaces, but of course, every vector space of $\{0,1\}^n$ is, in particular, a coset.

\begin{theorem}[\v{D}uri\v{s} et al.\cite{vdurivs2004multi}]
	\label{thm:random-vector-space-hard-for-bdd}
	The characteristic function $f_V: \{0, 1\}^n \rightarrow \{0, 1\}$ of a randomly chosen vector space $V$ in $\{0, 1\}^n$, defined as $f_V(x)=1$ if $x\in V$ and  $0$ otherwise, needs a \bdd of size $2^{\Omega(n)}/(2n)$ with high probability.
\end{theorem}

Our result follows by noting that if $f$ has codomain $\{0,1\}$ as above, then the \qmdd of the state $\ket{f}=\sum_xf(x)\ket{x}$ has the same structure as the \bdd of $f$.
Consequently, in particular the \bdd and \qmdd have the same number of nodes.

\begin{corollary}
	\label{thm:random-xor-state-hard-for-qmdd}
	For a random vector space $V\subseteq \{0,1\}^n$, the coset state $\ket{V}$ requires \qmdds of size $2^{\Omega(n)}/(2n)$ with high probability.
\end{corollary}
\begin{proof}
	We will show that the \qmdd has the same number of nodes as a \bdd.
	A \bdd encodes a function $f\colon \{0,1\}^n\to\{0,1\}$.
	In this case, the \bdd encodes $f_V$, the characteristic function of $V$.
	A \bdd is a graph which contains one node for each subfunction of $f$.
	(In the literature, such a \bdd is sometimes called a Full BDD, so that the term \bdd is reserved for a variant where the nodes are in one-to-one correspondence with the subfunctions $f$ which satisfy $f_0\ne f_1$).
	
	Similarly, a \qmdd representing a state $\ket{\phi}=\sum_xf(x)\ket{x}$ can be said to represent the function $f\colon\{0,1\}^n\to\mathbb C$, and contains one node for each subfunction of $f$ modulo scalars.
	We will show that, two distinct subfunctions of $f_V$ are never equal up to a scalar.
	To this end, let $f_{V,a},f_{V,b}$ be distinct subfunctions of $f_V$ induced by partial assignments $a,b\in\{0,1\}^k$.
	We will show that there is no $\lambda\in\mathbb C^{\ast}$ such that $f_{V,a}=\lambda f_{V,b}$.
	Since the two subfunctions are not pointwise equal, say that the two subfunctions differ in the point $x\in\{0,1\}^{n-k}$, i.e., $f_{V,a}(x)\ne f_{V,b}(x)$.
	Say without loss of generality that $f_{V,a}(x)=0$ and $f_{V,b}(x)=1$.
	Then, since $\lambda\ne 0$, we have $\lambda=\lambda f_{S,b}(x)\ne f_{V,a}(x)=0$, so $f_{V,a}\ne \lambda f_{B,b}$.
	
	Because distinct subfunctions of $f_V$ are not equal up to a scalar, the \qmdd of $\ket{V}$ contains a node for every unique subfunction of $f_V$.
	We conclude that, since by \autoref{thm:random-vector-space-hard-for-bdd} with high probability the \bdd representing $f_V$ has exponentially many nodes, so does the \qmdd representing $\ket{V}$.
\end{proof}

\section{How to write graph states, coset states and stabilizer states as Tower-\limdds}
\label{sec:graph-states-limdds}
\label{sec:proof-stabilizer-states-tower-limdds}

In this appendix, we prove that the families of $\langle Z\rangle$-, 
$\langle X\rangle$-, and \gen\Pauli-Tower-\limdds correspond to graph states, coset states, and stabilizer states, respectively, in \autoref{thm:graph-states-z-limdd}, \autoref{lemma:vector-space-x-limdd} and \autoref{thm:pauli-tower-limdds-are-stabilizer-states} below.
\autoref{def:reduced-limdd} for reduced \Pauli-\limdds requires modification for  $G=\braket{Z}$-\limdds because of the absence of $X$ as discussed below the definition. Note that the proofs do not rely on the specialized definition of reduced \limdds, but only on \autoref{def:limdd} which allows  parameterization of the LIM $G$. They only rely on the Tower \limdd in \autoref{def:tower}.

Before we give the proof, we remark that graph states present an interesting special case because the \limdd's edge labels contain meaningful information.
Namely, the labels on the high edges of a graph state's LIMDD are precisely the edges in the original graph.
Specifically, suppose a graph $G$ gives rise to a graph state $\ket{\phi_G}$ represented by a LIMDD.
Let $P=P_{k-1}\otimes\cdots\otimes P_1$ be the label on the high edge out of the LIMDD node at level $k$.
Then $G$ contains an edge $(v_k,v_j)$ if and only if $P_j=Z$ (with the roles of $k$ and $j$ reversed if $k<j$).
These edge labels come about in a straightforward manner during the construction of the graph state.
Namely, the graph state $\ket{\phi_G}$ is produced by starting from the state $\ket{+}^{\otimes n}$, and applying controlled-$Z$ gates to qubit pairs $(u,v)$ for every edge $(u,v)$ in the graph.
Applying such a controlled-$Z$ gate to qubit pair $(u,v)$ has the effect of setting $P_v$ to $Z$ in the high edge outgoing from the vertex at level $u$.
In general, however, the labels on the high edges cannot be easily inferred from the stabilizer state.

A $G$-Tower-\limdd representing an $n$-qubit state is a \limdd which has $n$ nodes, not counting the leaf. It has $G$-LIMs on its high edges. \autoref{def:tower} gives an exact definition.

\begin{theorem}[Graph states are $\langle Z\rangle$-Tower-\limdds]
    \label{thm:graph-states-z-limdd}
    Let $n\geq 1$.
    Denote by $\mathcal{G}_n$ the set of $n$-qubit graph states and write $\mathcal{Z}_n$ for the set of $n$-qubit quantum states which are represented by \gen Z-Tower-\limdds a defined in \autoref{def:tower}, i.e, a tower with low-edge-labels $\id$ and high-edge labels $\lambda \bigotimes_j P_j$ with $P_j \in \{\id, Z\}$ and $\lambda =1$, except for the root edge where $\lambda \in \mathbb C\setminus \set 0$.
    Then $\mathcal{G}_n = \mathcal{Z}_n$.
\end{theorem}
\begin{proof}
    We establish $\mathcal{G}_n \subseteq \mathcal{Z}_n$ by providing a procedure to convert any graph state in $\mathcal{G}_n$ to a \gen Z-Tower-\limdd in $\mathcal{Z}_n$.
See \autoref{fig:graph-state-as-line-iso-qmdd} for an example of a $4$-qubit graph state.
	We describe the procedure by induction on the number $n$ of qubits in the graph state.
	
	\textbf{Base case: $n=1$.}
    We note that there is only one single-qubit graph state by definition (see \autoref{eq:graph-state-definition}), which is \mbox{$\ket{+} := (\ket{0} + \ket{1}) / \sqrt{2}$} and can be represented as \limdd by a single node (in addition to the leaf node): see \autoref{fig:graph-state-as-line-iso-qmdd}(a).
    
    \textbf{Induction case.}
	We consider an $(n+1)$-qubit graph state $\ket{G}$ corresponding to the graph~$G$.
	We isolate the $(n+1)$-th qubit by decomposing the full state definition from  \autoref{eq:graph-state-definition}:
	\begin{equation}
        \ket{G} = 
        \frac{1}{\sqrt{2}}\left(\ket{0} \otimes \ket{G_{1..n}} + \ket{1} \otimes 
            \underset{\textnormal{Isomorphism B}}{ \underbrace{
\left[
            \bigotimes_{(n+1,j)\in E} Z_j
            \right] 
            }}
        \ket{G_{1..n}}\right)
		\label{eq:graph-state-induction}
	\end{equation}
where $E$ is the edge set of $G$ and $G_{1..n}$ is the induced subgraph of $G$ on vertices $1$ to $n$.
Thus, $\ket{G_{1..n}}$ is an $n$-qubit graph state on qubits 1 to $n$.
    Since $\ket{G_{1..n}}$ is a graph state on $n$ qubits, by the induction hypothesis, we have a procedure to convert it to a \gen Z-Tower-\limdd $\in \mathcal{Z}_{n}$.
Now we construct a \gen Z-Tower-\limdd for $\ket{G}$ as follows.
The root node has two outgoing edges, both going to the node representing $\ket{G_{1.. n}}$.
    The node's low edge has label $\mathbb I$, and the node's high edge has label $B$, as follows,
\begin{align}
	B = \bigotimes_{(n+1,j)\in E}Z_j
\end{align}
Thus the root node represents the state $\ket{0}\ket{G_{1..n}}+\ket{1}B\ket{G_{1..n}}$, satisfying \autoref{eq:graph-state-induction}.
\begin{figure*}
\begin{tabular}{l|l|l|ll}
    \begin{tikzpicture}[->,>=stealth',shorten >=1pt,auto,node distance=1cm,
        thick, state/.style={circle,draw,minimum size=0.5cm},font=\scriptsize, scale=0.3,
    inner sep=0pt,]       
 
    \node[state](1){};
    \node[leaf] (s)[below = .5cm of 1]{$1$};
        
    \draw[<-,e1] (1) --++(0:2.4cm) node[] {};
    
    \path[]
    (1)  edge[e0= 20]  node[pos=.6,left] {} (s)
    (1) edge[e1= 20]  node[pos=.45,] {} (s)
	;

\node[state,draw=white,above = .5cm of 1] (a13) {};
\node[state,draw=white,above = .5cm of a13] (a12) {};
\node[state,draw=white,above = .5cm of a12] (a11) {};

\node[circle, fill=black, above = .8cm of a11,minimum size=.1cm,xshift=-.2cm] (c1) {};
\node[below right=-.1cm of c1] {$A$};

\node[left =.5cm of a11,yshift=1cm] {a)};
    \end{tikzpicture}
&
	\begin{tikzpicture}[->,>=stealth',shorten >=1pt,auto,node distance=1cm,
        thick, state/.style={circle,draw,minimum size=0.5cm},font=\scriptsize, scale=0.3,
    inner sep=0pt,]
        
    \node[state](1){};
    \node[state](1a)[below= .5cm of 1]{};
    \node[leaf] (s)[below = .5cm of 1a]{$1$};
        
    \draw[<-,e1] (1) --++(0:2.4cm) node[] {};
    
    \path[]
    (1)  edge[e0= 20]  node[pos=.6,left] {} (1a)
    (1)  edge[e1= 20]  node[lbl,pos=.45,] {$Z$} (1a)
    (1a) edge[e0= 20]  node[pos=.6,left] {} (s)
    (1a) edge[e1= 20]  node[pos=.45,] {} (s)
	;
\node[state,draw=white,above = .5cm of 1] (a12) {};
\node[state,draw=white,above = .5cm of a12] (a11) {};

\node[circle, fill=black, above = .8cm of a11,minimum size=.1cm,xshift=-.2cm] (c1) {};
\node[below right=-.1cm of c1] {$A$};
\node[circle, fill=black, right = .4cm of c1,minimum size=.1cm] (c2) {};
\node[below right=-.1cm of c2] {$B$};
\draw[-] (c1) -- (c2);

\node[left =.5cm of a11,yshift=1cm] {b)};
    \end{tikzpicture}
&
	\begin{tikzpicture}[->,>=stealth',shorten >=1pt,auto,node distance=1cm,
        thick, state/.style={circle,draw,minimum size=0.5cm},font=\scriptsize, scale=0.3,
    		inner sep=0pt,]
    \node[state,right= 1cm of a1] (a1) {};
    \node[state, below =.5cm of a1] (a3) {};
    \node[state, below =.5cm of a3] (a4) {};
    \node[draw,rectangle,minimum size=0.4cm, below= .5cm of a4] (w4) {1};

    \draw[<-] (a1) --++(0:2cm) node[pos=1.4] {};
    \draw[e0=25] (a1) edge  node[] {} (a3);
    \draw[e1=25] (a1) edge  node[lbl,right] {$Z\otimes \id$} (a3);
    \draw[e0=25] (a1) edge  node[] {} (a3);
    \draw[e1=25] (a3) edge  node[lbl,right] {$Z$} (a4);
    \draw[e0=25] (a3) edge  node[] {} (a4);
    \draw[e1=25] (a4) edge  node[right] {} (w4);
    \draw[e0=25] (a4) edge  node[] {} (w4);

    \node[state,draw=white,above = .5cm of a1] (a11) {};

\node[circle, fill=black, above = .8cm of a11,minimum size=.1cm,xshift=-.2cm] (c1) {};\node[below right=-.1cm of c1] {$A$};
\node[circle, fill=black, right = .4cm of c1,minimum size=.1cm] (c2) {};
\node[below right=-.1cm of c2] {$B$};
\node[circle, fill=black, below = .4cm of c2,minimum size=.1cm] (c4) {};
\node[below right=-.1cm of c4] {$C$};
\draw[-] (c1) -- (c2);
\draw[-] (c4) -- (c2);

\node[left =.5cm of a11,yshift=1cm] {c)};
\end{tikzpicture}
&
	\begin{tikzpicture}[->,>=stealth',shorten >=1pt,auto,node distance=1cm,
        thick, state/.style={circle,draw,minimum size=0.5cm},font=\scriptsize, scale=0.3,
    		inner sep=0pt,]
    \node[state,right= 1cm of a1] (a1) {};
    \node[state, below =.5cm of a1] (a3) {};
    \node[state, below =.5cm of a3] (a4) {};
    \node[state, below =.5cm of a4] (a5) {};
    \node[draw,rectangle,minimum size=0.4cm, below= .5cm of a5] (w4) {1};

    \draw[<-] (a1) --++(0:2cm) node[pos=1.4] {};
    \draw[e0=25] (a1) edge  node[] {} (a3);
    \draw[e1=25] (a1) edge  node[lbl,right] {$Z\otimes \id \otimes Z$} (a3);
    \draw[e0=25] (a1) edge  node[] {} (a3);
    \draw[e1=25] (a3) edge  node[lbl,right] {$Z\otimes \id$} (a4);
    \draw[e0=25] (a3) edge  node[] {} (a4);
    \draw[e1=25] (a4) edge  node[lbl,right] {$Z$} (a5);
    \draw[e0=25] (a4) edge  node[] {} (a5);
    \draw[e0=25] (a5) edge  node[] {} (w4);
    \draw[e1=25] (a5) edge  node[lbl,right] {0} (w4);
    
\node[circle, fill=black, above = .8cm of a1,minimum size=.1cm,xshift=1.5cm] (c1) {};
\node[below right=-.1cm of c1] {$A$};
\node[circle, fill=black, right = .4cm of c1,minimum size=.1cm] (c2) {};
\node[below right=-.1cm of c2] {$B$};
\node[circle, fill=black, below = .4cm of c1,minimum size=.1cm] (c3) {};
\node[below right=-.1cm of c3] {$C$};
\node[circle, fill=black, below = .4cm of c2,minimum size=.1cm] (c4) {};
\node[below right=-.1cm of c4] {$D$};
\draw[-] (c1) -- (c2);
\draw[-] (c1) -- (c3);
\draw[-] (c4) -- (c3);
\draw[-] (c4) -- (c2);

\node[left =.5cm of a1,yshift=1cm] {d)};
    
	\end{tikzpicture}
&\hspace{-1.2cm}
\begin{tikzpicture}[->,>=stealth',shorten >=1pt,auto,node distance=.7cm,
        thick, state/.style={circle,draw,minimum size=0.5cm},font=\scriptsize, scale=0.3,
    		inner sep=0pt,]
\node[state,] (n1) {};

\node[state](n2)[below = of n1, xshift=-1.3cm]{};
\node[state](n3)[below = of n1, xshift= 1.3cm, inner sep = 0pt]{};

\node[state](n21)[below = of n2, xshift=-.7cm, inner sep = 0pt]{};
\node[state](n22)[below = of n2, xshift= .7cm, inner sep = 0pt]{};
\node[state](n31)[below = of n3, xshift=-.7cm, inner sep = 0pt]{};
\node[state](n32)[below = of n3, xshift= .7cm, inner sep = 0pt]{};

\node[state](n41)[below = of n21, xshift= .7cm, inner sep = 0pt]{};
\node[state](n42)[below = of n31, xshift= .7cm, inner sep = 0pt]{};

\node[draw, leaf,below = of n41, xshift=1.3cm] (e) {$1$};

\draw[<-] (n1) --++(0:2cm) node[pos=.6] {};

\path[]
(n1) edge[e0] node[left,pos=.5] {} (n2)
(n1) edge[e1] node[]			{} (n3)
(n2) edge[e0] node[left,pos=.5] {} (n21)
(n2) edge[e1] node[,] 		{} (n22)
(n3) edge[e0] node[left,pos=.5] {} (n31)
(n3) edge[e1] node[lbl] {$-1$} (n32)

(n21) edge[e0=20] node[left,pos=.5] {}  (n41)
(n21) edge[e1=20] node[		  ] {}  (n41)
(n22) edge[e1] node[left,pos=0.7,below,lbl] {$-1$}  (n42)
(n22) edge[e0] node[right,yshift=.cm] {}  (n41)
(n31) edge[e0] node[left,pos=.5] {} (n41)
(n31) edge[e1] node[pos=.7	 ] {} (n42)
(n32) edge[e1=0] node[lbl,left,pos=.05] {$-1$} (n41)
(n32) edge[e0=20] node[lbl       ] {} (n42)

(n41) edge[e0=20] node[pos=.5] {} (e)
(n41) edge[e1=20] node[pos=.5] {} (e)
(n42) edge[e0=20] node[pos=.5] {} (e)
(n42) edge[e1=20] node[lbl] {$-1$} (e);
\end{tikzpicture}
\\
\end{tabular}
	\caption{
		Construction of the \gen Z-Tower \limdd for the 4-qubit cluster state, by iterating over the vertices in the graph, as described in the proof of \autoref{thm:graph-states-z-limdd}.
    (a) First, we consider the single-qubit graph state, which corresponds to a the subgraph containing only vertex $A$.
    (b) Then, we add vertex $B$, which is connected to $A$ by an edge.
	The resulting \limdd is constructed from the \limdd from (a) by adding a new root node.
	In the figure, the isomorphism is $Z_B\otimes \id[A]$, since vertex $C$ is connected to vertex $B$ (yielding the $Z$ operator) but not to $A$ (yielding the identity operator $\id$).
    (c) This process is repeated for a third vertex $C$ until we reach the \limdd of the full $4$-qubit cluster state (d).
	For comparison, (d) also depicts a regular \qmdd for the same graph state, which has width $4$ instead of $1$ for the \limdd.
		\label{fig:graph-state-as-line-iso-qmdd}
		}
\end{figure*}

    To prove $\mathcal{Z}_n \subseteq \mathcal{G}_n$, we show how to construct the graph corresponding to a given $\braket{Z}$-Tower \limdd. Briefly, we simply run the algorithm outlined above in reverse, constructing the graph one node at a time.
    Here we assume without loss of generality that the low edge of every node is labeled $\mathbb I$.

    \textbf{Base case.} The \limdd node above the Leaf node, representing the state $\ket +$, always represents the singleton graph, containing one node.
    
    \textbf{Induction case.} Suppose that the \limdd node $k+1$ levels above the Leaf has a low edge labeled $\mathbb I$, and a high edge labeled $P_{k}\otimes \cdots\otimes P_1$, with $P_j=Z^{a_j}$ for $j=1\ldots k$.
    Here by $Z^{a_j}$ we mean $Z^0=\mathbb I$ and $Z^1=Z$.
    Then we add a node labeled $k+1$ to the graph, and connect it to those nodes $j$ with $a_j=1$, for $j=1\ldots k$.
    The state represented by this node is of the form given in \autoref{eq:graph-state-induction}, so it represents a graph state.
    
	A simple counting argument based on the above construction shows that $\sizeof{\mathcal{Z}_n} = \sizeof{\mathcal{G}_n} = 2^{n \choose 2}$, so the conversion is indeed a bijection.
Namely, there are $2^{\binom n2}$ graphs, since there are $\binom n2$ edges to choose, and there are $2^{\binom n2}$ \gen Z-Tower-\limdds, because the total number of single-qubit operators of the LIMs on the high edges is $\binom n2$, each of which can be independently chosen to be either $\mathbb I$ or $Z$.
\end{proof}

We now prove that coset states are represented by $\braket{X}$-Tower-\limdds.

\begin{theorem}[coset states are $\langle X\rangle$-Tower-\limdds]
    \label{lemma:vector-space-x-limdd}
    Let $n\geq 1$.
    Denote by $\mathcal{V}_n$ the set of $n$-qubit coset states and write $\mathcal{X}_n$ for the set of $n$-qubit quantum states which are represented by \gen X-Tower-\limdds as per \autoref{def:tower}, i.e., a tower with low edge labels $\id$ and high edge labels $\lambda \bigotimes_j P_j$ with $P_j \in \{\mathbb I, X\}$ and $\lambda \in \{0, 1\}$, except for the root edge where $\lambda \in \mathbb C\setminus \set 0$.
    Then $\mathcal{V}_n = \mathcal{X}_n$.
\end{theorem}
\begin{proof}
    We first prove $\mathcal{V}_n \subseteq \mathcal{X}_n$ by providing a procedure for constructing a Tower-\limdd for a coset state.
    We prove the statement for the case when $C$ is a group rather than a coset; the result will then follow by noting that, by placing the label $X^{a_n}\otimes\cdots\otimes X^{a_1}$ on the root edge, we obtain the coset state $\ket{C+a}$.
    The procedure is recursive on the number of qubits.

	\textbf{Base case: $n=1$.} In this case, there are two coset states: $\ket{0}$ and $(\ket{0} + \ket{1}) / \sqrt{2}$, which are represented by a single node which has a low and high edge pointing to the leaf node with low/high edge labels 1/0 and 1/1, respectively.
	
	\textbf{Induction case. }Now consider an $(n+1)$-qubit coset state $\ket{S}$ for a group $S\subseteq \{0,1\}^{n+1}$ for some $n\geq 1$ and assume we have a procedure to convert any $n$-qubit coset state into a Tower-\limdd~in~$ \mathcal{X}_n$. 
We consider two cases, depending on whether the first bit of each element of $S$ is zero:
    \begin{enumerate}[(a)]
        \item{
                The first bit of each element of $S$ is $0$.
                Thus, we can write $S = \set{0x \mid  x \in S_0 }$ for some set $S_0\subseteq\{0, 1\}^n$.
            Then $0a,0b\in S \implies 0a\oplus 0b\in S$ implies $a,b\in S_0 \implies a\oplus b\in S_0$ and thus $S_0$ is an length-$n$ bit string vector space.
            Thus by assumption, we have a procedure to convert it to a Tower-\limdd in $\mathcal{X}_n$.
            Convert it into a Tower-\limdd in $\mathcal{X}_{n+1}$ for $\ket{S}$ by adding a fresh node on top with low edge label $\id^{\otimes n}$ and high edge label $0$, both pointing to the the root $S$.
            }
         \item{
                 There is some length-$n$ bit string $u$ such that $1u\in S$.
                 Write $S$ as the union of the sets $\set{0x \mid x \in S_0}$ and $\set{1x \mid x\in S_1}$ for sets $S_0, S_1 \subseteq \{0,1\}^n$.
                 Since $S$ is closed under element-wise XOR, we have $1u \oplus 1x = 0(u\oplus x) \in S$ for each $x\in S_1$ and therefore $u\oplus x \in S_0$ for each $x \in S_1$.
                 This implies that $S_1 = \set{u\oplus x\mid x\in S_0}$ and thus $S$ is the union of $\set{0x \mid  x \in S_0}$ and $\set{1u \oplus 0x\mid x\in S_0}$.
                 By similar reasoning as in case (a), we can show that $S_0$ is a vector space on length-$n$ bit strings.
                 
                 We build a Tower-\limdd for $\ket S$ as follows.
                 By the induction hypothesis, there is a Tower-\limdd with root node $v$ which represents $\ket{v}=\ket{S_0}$.
                 We construct a new node whose two outgoing edges both go to this node $v$.
                 Its low edge has label $\id^{\otimes n}$ and its high edge has label $P=P_n\otimes\cdots\otimes P_1$ where $P_j=X$ if $u_j=1$ and $P_j=\mathbb I$ if $u_j=0$.
             }
    \end{enumerate}

	We now show $\mathcal{V}_n\subseteq \mathcal X_n$, also by induction.
	
	\textbf{Base case: $n=1$.} There are only two Tower-\limdds on $1$ qubit satisfying the description above, namely
	\begin{enumerate}[(1)]
		\item A node whose two edges point to the leaf. Its low edge has label $1$, and its high edge has label $0$.
		This node represents the coset state $\ket 0$, corresponding to the vector space $V=\{0\}\subseteq\{0,1\}^1$.
		\item A node whose two edges point to the leaf. Its low edge has label $1$ and its high edge also has label $1$.
		This node represents the coset state $\ket 0+\ket 1$, corresponding to the vector space $V=\{0,1\}$.
	\end{enumerate}
	
	\textbf{Induction case. } Let $v$ be the root node of an $n+1$-qubit Tower $\braket{X}$-\limdd as described above.
	We distinguish two cases, depending on whether $v$'s high edge has label $0$ or not.
	\begin{enumerate}[(a)]
		\item The high edge has label $0$.
		Then $\ket{v}=\ket{0}\ket{v_0}$ for a node $v_0$, which represents a coset state $\ket{v_0}$ corresponding to a coset $V_0\subseteq\{0,1\}^{n}$, by the induction hypothesis.
		Then $v$ corresponds to the coset $\set{0x\mid x\in V_0}$.
    \item the high edge has label $P=P_n\otimes\cdots\otimes P_1$ with $P_j \in \{\id, X\}$.
		Then $\ket{v}=\ket 0\ket{v_0}+\ket{1}\otimes P\ket{v_0}$.
            By the observations above, this is a coset state, corresponding to the vector space $V=\{0x|x\in V_0\}\cup \{1(ux)|x\in V_0\}$ where $u\in\{0,1\}^n$ is a string whose bits are $u_j=1$ if $P_j=X$ and $u_j=0$ if $P_j=\id$, and $V_0$ is the vector space corresponding to the coset state $\ket{v_0}$.
	\end{enumerate}
	\vspace{-2em}
\end{proof}

Lastly, we prove the stabilizer-state case, showing that they are exactly equivalent to the \gen\pauli-Tower-\limdd,
	as defined in \autoref{def:tower}.
For this, we first need \autoref{lemma:clifford-tower} and \autoref{thm:clifford-gate-stabilizer-limdd-general}, which state that, if one applies a Clifford gate to a \gen\pauli-Tower-\limdd, the resulting state is another \gen\pauli-Tower-\limdd.
First, \autoref{lemma:clifford-tower} treats the special case of applying a gate to the top qubit; then \autoref{thm:clifford-gate-stabilizer-limdd-general} treats the general case of applying a gate to an arbitrary qubit.

\begin{lemma}
    \label{lemma:clifford-tower}
    Let $\ket{\phi}$ be an $n$-qubit stabilizer state which is represented by a \gen\pauli-Tower-\limdd as defined in \autoref{def:tower}.
    Let $U$ be either a Hadamard gate or $S$ gate on the top qubit ($n$-th qubit), or a downward CNOT with the top qubit as control.
    Then $U\ket{\phi}$ is still represented by a \gen\Pauli-Tower-\limdd.
\end{lemma}
\begin{proof}
    The proof is on the number $n$ of qubits.
    
    \textbf{Base case: $n=1$.} For $n=1$, there are six single-qubit stabilizer states $\ket{0}, \ket{1}$ and $(\ket{0} + \alpha \ket{1}) / \sqrt{2}$ for $\alpha \in \{\pm 1, \pm i\}$.
    There are precisely represented by Pauli-Tower-\limdds with high edge label factor $\in \{0, \pm 1, \pm i\}$ as follows:
    \begin{itemize}
        \item for $\ket{0}$: $\lnode{1}{1}{0}{1}$
        \item for $\ket{1}$: $A \cdot \lnode{1}{1}{0}{1}$ where $A \propto X$ or $A\propto Y$
            \item for $(\ket{0} + \alpha \ket{1}) / \sqrt{2})$: $\lnode{1}{1}{\alpha}{1}$
    \end{itemize}
    Since the $H$ and $S$ gate permute these six stabilizer states, $U\ket{\phi}$ is represented by a \gen\pauli-Tower-\limdd if $\ket{\phi}$ is.
 
    \textbf{Induction case.} For $n>1$, we first consider $U =S$ and $U = \text{CNOT}$.
    Let $R$ be the label of the root edge.
    If $U=S$, then the high edge of the top node is multiplied with $i$, while a downward CNOT (target qubit with index $k$) updates the high edge label $A \mapsto X_k A$.
    Next, the root edge label is updated to $URU^\dagger$, which is still a Pauli string, since $U$ is a Clifford gate.
    Since the high labels of the top qubit in the resulting diagram is still a Pauli string, and the high edge's weights are still $\in\{0,\pm 1,\pm i\}$, we conclude that both these gates yield a \gen\pauli-Tower-\limdd.
    Finally, for the Hadamard, we decompose $\ket{\phi} = \ket{0} \otimes \ket{\psi} + \alpha \ket{1} \otimes P\ket{\psi}$ for some $(n-1)$-qubit stabilizer state $\ket{\psi}$, $\alpha \in \{0, \pm 1, \pm i\}$ and $P$ is an $(n-1)$-qubit Pauli string.
    Now we note that $H\ket{\phi} \propto \ket{0} \otimes \ket{\psi_0} + \ket{1} \otimes \ket{\psi_1}$ where $\ket{\psi_x} := (\id + (-1)^x \alpha P) \ket{\psi}$ with $x\in \{0, 1\}$.
    Now we consider two cases, depending on whether $P$ commutes with all stabilizers of $\ket\psi$:
	\begin{enumerate}[(a)]
        \item There exist a stabilizer $g$ of $\ket{\psi}$ which anticommutes with $P$.
            We note two things.
            First,  $\langle \psi | P | \psi \rangle = \langle \psi | P g|\psi \rangle = \langle \psi | g\cdot (-P) |\psi \rangle = - \langle \psi |P|\psi\rangle$, hence $\langle \psi | P|\psi \rangle = 0$.
            It follows from Lemma 15 of \cite{garcia2012efficient} that $\ket{\psi_x}$ is a stabilizer state, so by the induction hypothesis it can be written as a \gen\pauli-Tower-\limdd.
            Let $v$ be the root node of this \limdd.
            Next, we note that $g\ket{\psi_0} = g (\id + \alpha P) \ket{\psi} = (\id - \alpha P) g\ket{\psi} = \ket{\psi_1}$.
            Hence, $\lnode{\id}{v}{g}{v}$ is the root node of a \gen\pauli-Tower-\limdd for $H\ket{\phi}$.
        \item All stabilizers of $\ket{\psi}$ commute with $P$. Then $(-1)^y P$ is a stabilizer of $\ket{\psi}$ for either $y = 0$ or $y=1$. Hence, $\ket{\psi_x} = (\id + (-1)^x \alpha P) \ket{\psi} = (1 + (-1)^{x + y} \alpha) \ket{\psi}$.
            Therefore, $\ket{\phi} = \ket{a} \otimes\ket{\psi}$ where $\ket{a} := (1 + (-1)^y \alpha) \ket{0} + (1 + (-1)^{y+1} \alpha \ket{1})$.
            It is not hard to see that $\ket{a}$ is a stabilizer state for all choices of $\alpha \in \{0, \pm 1,\pm i\}$.
            By the induction hypothesis, both $\ket{a}$ and $\ket{\psi}$ can be represented as \gen\pauli-Tower-\limdds.
            We construct a \gen\pauli-Tower-\limdd for $H\ket{\phi}$ by replacing the leaf of the \limdd of $\ket{a}$ by the root node of the \limdd of $\ket{\psi}$, and propagating the root edge label of $\ket{\psi}$ upwards.
            Specifically, if the root edge of $\ket a$ is $\ledge Av$ with $v=\lnode{1}{1}{\beta}{1}$, and if the root edge of $\ket\psi$ is $\ledge Bw$, then a \gen\pauli-Tower-\limdd for $H\ket\phi$ has root node $\lnode{w}{\mathbb I}{\beta\mathbb I}{w}$ and has root edge label $A\otimes B$.
    \end{enumerate}
    \vspace{-2em}
\end{proof}

\begin{lemma}
	\label{thm:clifford-gate-stabilizer-limdd-general}
	Let $\ket\phi$ be an $n$-qubit state state represented by a \gen\pauli-Tower-\limdd,
	as defined in \autoref{def:tower}.
	Let $U$ be either a Hadamard gate, an $S$ gate or a CNOT gate.
	Then $U\ket\phi$ is a state which is also represented by a \gen\pauli-Tower-\limdd.
\end{lemma}
\begin{proof}
	The proof is by induction on $n$.
	The case $n=1$ is covered by \autoref{lemma:clifford-tower}.
	Suppose that the induction hypothesis holds, and let $\ket\phi$ be an $n+1$-qubit state represented by a \gen\pauli-Tower-\limdd.
	First, we note that a CNOT gate $CX_c^t$ can be written as $CX_c^t = (H\otimes H) CX_t^c (H \otimes H)$, so without loss of generality we may assume that $c>t$.
	We treat two cases, depending on whether $U$ affects the top qubit or not.
	\begin{enumerate}[(a)]
		\item $U$ affects the top qubit.
		Then $U\ket\phi$ is represented by a \gen\pauli-Tower-\limdd, according to \autoref{lemma:clifford-tower}.
		\item $U$ does not affect the top qubit.
		Suppose $\ket{\phi} = \ket{0} \otimes \ket{\phi_0} + \ket{1} \otimes \alpha P\ket{\phi_0}$ (with $P$ a Pauli string and $\alpha\in\{0,\pm 1,\pm i\}$).
		Then $U\ket{\phi} = \ket{0} \otimes U\ket{\phi_0} + \ket{1} \otimes (\alpha UPU^{\dagger})U\ket{\phi_0}$.
		Since $U$ is either a Hadamard, $S$ gate or CNOT, and $\ket{\phi_0}$ is an $n$-qubit state, the induction hypothesis states that the state $U\ket{\phi_0}$ is represented by a \gen\pauli-Tower-\limdd.
		Let $\ledge Av$ be the root edge of this \gen\pauli-Tower-\limdd, representing $U\ket{\phi_0}$.
		Then $U\ket{\phi}$ is represented by the root edge $\ledge {\mathbb I\otimes A}w$, where $w$ is the node $\lnode{\id}{v}{\alpha A^{-1}UPU^{\dagger}A}{v}$.
		The label $\alpha A^{-1}UPU^\dagger A$ is a Pauli LIM, and may therefore be used as the label on the high edge of $w$.
	\end{enumerate}
	\vspace{-2em}
\end{proof}

Finally, we show that stabilizer states are precisely the \gen\pauli-Tower-\limdds.

\thmpaulitower*
\begin{proof}
    We first prove that each stabilizer state is represented by a \gen\Pauli-Tower-\limdd.
    We recall that each stabilizer state can be obtained as the output state of a Clifford circuit on input state $\ket{0}^{\otimes n}$.
    Each Clifford circuit can be decomposed into solely the gates $H, S$ and CNOT.
    The state $\ket{0}^{\otimes n}$ is represented by a \gen\pauli-Tower-\limdd.
    According to \autoref{thm:clifford-gate-stabilizer-limdd-general}, applying an $H$, $S$ or CNOT gate to a \gen\pauli-Tower-\limdd results a state represented by another \gen\pauli-Tower-\limdd.
    One can therefore apply the gates of a Clifford circuit to the initial state $\ket 0$, and obtain a \gen\pauli-Tower-\limdd for every intermediate state, including the output state.
    Therefore, every stabilizer state is represented by a \gen\pauli-Tower-\limdd.

    For the converse direction, the proof is by induction on $n$.
    We only need to note that a state represented by a \gen\pauli-Tower-\limdd can be written as $\ket{\phi} = \ket{0} \otimes \ket{\phi_0} + \ket{1} \otimes \alpha P \ket{\phi_0} = C(P) (\ket{0} + \alpha \ket{1}) \otimes \ket{\phi_0}$ where $C( P) := \dyad{0} \otimes \id + \dyad{1} \otimes P$ is the controlled-$(P)$ gate.
    Using the relations $Z = HXH$, $Y = SXS^{\dagger}$ and $S=Z^2$, we can decompose $C(P)$ as CNOT, $H$ and $S$, hence $C(P)$ is a Clifford gate.
    Since both $\ket{0} + \alpha \ket{1}$ and $\ket{\phi_0}$ can be written as \gen\Pauli-Tower-\limdds, they are stabilizer states by the induction hypothesis.
    Therefore, the state $\ket\psi = (\ket 0+\alpha\ket 1)\otimes \ket{\phi_0}$ is also a stabilizer state.
    Thus, the state $\ket\phi=C(P)\ket\psi$ is obtained by applying the Clifford gate $C(P)$ to the stabilizer state $\ket\phi$.
    Therefore, $\ket\phi$ is a stabilizer state.
\end{proof}

\section{Efficient algorithms for choosing a canonical high label \label{app:canonical-high-label}}

\begin{figure}[b]
~~~~~~~~~~~~~~~~~~~~~~~~~~~~~~~~~\tikz[->,>=stealth',shorten >=1pt,auto,node distance=1.5cm,
        thick, state/.style={circle,draw,inner sep=0pt,minimum size=14pt}]{

    \node[state] (2) {$w$};
    \node[state] (1a) [below = 1cm of 2, xshift=1.3cm] {$v_0$};
    \node[right = 0.25cm of 1a] () {$\beforeq$};
    \node[state] (1b) [below = 1cm of 2, xshift=2.8cm] {$v_1$};
    \node[above = 0.5cm of 2,xshift=1.95cm,fill=black] (x)  {};
    \path[]
    (x) edge[bend left=-20]    node[above right,pos=.8] {} (2)
    (2) edge[e0] node[pos=.5,left] {$\id^{\otimes n}$} (1a)
    (2) edge[e1] node[pos=.15,right, xshift = .5cm] {$\lambda P$} (1b)
    ;
    
    \node[state, right = 3.3cm of 2] (4) {$v^{\textnormal{r}}$};
    \path[]
    (x) edge[bend left=20]     node[pos=.9] {
        $\rootlim = \left(\lambda X\otimes P\right)^{x} \cdot
        \left(Z^s \otimes \left(\gmax_0\right)^{-1}\right)$
        } (4)
    (4) edge[e0] node[pos=.3,above left] {$\id^{\otimes n}$} (1a)
    (4) edge[e1] node[pos=.3,below right] {
        $\highlim = 
			(-1)^s\lambda^{(-1)^x}g_0Pg_1
        $
    } (1b)
    (2) --  node[yshift=.1cm] {$\rightsquigarrow$} (4)
    (1a)  edge[loop left, dashed] node[left]{$g_0\in\Aut(v_0)$} (1a)
    (1b)  edge[loop right, dashed] node[right]{$\Aut(v_1) \ni g_1$} (1b)
    ;
    }
    \vspace{1mm}
    \newline
    \centering
    Choose $s,x\in \{0,1\},g_0\in \Stab(v_0),g_1\in \Stab(v_1)$ s.t. $\highlim$ is minimal and $x=0$ if $v_0\ne v_1$.
	\caption{
        Illustration of finding a canonical high label for a semi-reduced node $w$,
        yielding a reduced node $v^{\text{r}}$.
The chosen high label is the minimal element from the set of eligible high labels
based on stabilizers $g_0,g_1$ of $v_0,v_1$ (drawn as self loops).
The minimal element holds a factor $\lambda^{(-1)^x}$ for some $x \in \{0, 1\}$.
There are two cases: if $v_0 \neq v_1$ or $x=0$, then the factor is $\lambda$ and the root edge should be adjusted with an $\id$ or $Z$ on the root qubit.
The other case, $x=1$, leads to an additional multiplication with an $X$ on the root qubit.
}
	\label{fig:reduced}
\end{figure}

Here, we present an efficient algorithm which, on input Pauli-LIMDD node $\lnode[w]{\id}{v_0}{\lambda P}{v_1}$, returns a canonical choice for the high label \highlim (algorithm \textsc{GetLabels}, in \autoref{alg:find-canonical-edges}).
By \emph{canonical}, we mean that it returns the same high label for any two nodes in the same isomorphism equivalence class, i.e., for any two nodes $v, w$ for which $\ket{v} \simeq_{\text{Pauli}} \ket{w}$.

We first characterize all eligible labels \highlim in terms of the stabilizer subgroups of the children nodes $v_0,v_1$, denoted as $\Stab(v_0)$ and $\Stab(v_1)$ (see \autoref{sec:preliminaries} for the definition of stabilizer subgroup).
Then, we provide the algorithm \textsc{GetLabels} which correctly finds the lexicographically minimal eligible label (and corresponding root label), and runs in time $O(n^3)$ where $n$ is the number of qubits.

\autoref{fig:reduced} illustrates this process.
In the figure, the left node $w$ summarizes the status of the {\makeedge} algorithm on \autoref{algline:makeedge-get-labels}, when this algorithm has enough information to construct the semi-reduced node $\lnode[w]{\id^{\otimes n}}{v_0}{\lambda P}{v_1}$, shown on the left.
The node $v^r$, on the right, is the canonical node, and is obtained by replacing $w$'s high edge's label by the canonical label $B_{\text{high}}$.
This label is chosen by minimizing the expression $B_{\text{high}}=(-1)^s\lambda^{(-1)^x}g_0Pg_1$, where the minimization is over $s,x\in \{0,1\}, g_0\in Stab(\ket{v_0}),g_1\in Stab(\ket{v_1})$, subject to the constraint that $x=0$ if $v_0\ne v_1$.
We have $\ket{w}\simeq_{\text{Pauli}}\ket{v^r}$ by construction as intended, namely, they are related via $\ket{w} = B_{\text{root}}\ket{v^r}$.
\autoref{thm:eligible-isomorphisms-pauli} shows that this way to choose the high label indeed captures all eligible high labels, i.e., a node $\lnode[v^r]{\id}{v_0}{B_{\text{high}}}{v_1}$ is isomorphic to $\ket w$ if and only if $B_{\text{high}}$ is of this form.

\begin{theorem}
    [Eligible high-edge labels]
	\label{thm:eligible-isomorphisms-pauli}
    Let $\lnode[w]{\id^{\otimes n}}{v_0}{\lambda P}{v_1}$ be a semi-reduced $n$-qubit node
    in a Pauli-\limdd, where $v_0, v_1$ are reduced, $P$ is a Pauli string and $\lambda \neq 0$.
    For all nodes $v = \lnode[v]{\id^{\otimes n}}{v_0}{\highlim}{v_1}$, it holds that $\ket{w} \simeq \ket{v}$  if and only if 
    \begin{equation}
        \label{eq:eligible-high-label}
    \highlim = (-1)^s \cdot \lambda^{(-1)^x} g_0 P g_1
    \end{equation}
        for some $g_0 \in \Aut(v_0), g_1 \in \Aut(v_1), s,x\in \{0, 1\}$ and
        $x=0$ if $v_0 \neq v_1$.
    An isomorphism mapping $\ket{w}$ to $\ket{v}$ is 
    \begin{equation}
        \label{eq:root-label-eligible-high-label}
        \rootlim = (X \otimes \lambda P)^{x} \cdot (Z^s \otimes (g_0)^{-1}).
    \end{equation}
\end{theorem}
\begin{proof}
    \def\brest{C_{\textnormal{rest}}}
    \def\qtop{C_{\textnormal{top}}}
    It is straightforward to verify that the isomorphism $\rootlim$ in eq.~\eqref{eq:root-label-eligible-high-label} indeed maps $\ket{w}$ to $\ket{v}$ (as $x = 1$ implies $v_0 = v_1$), which shows that $\ket{w} \simeq \ket{v}$.
    For the converse direction, suppose there exists an $n$-qubit Pauli LIM $C$ such that $C\ket{w} = \ket{v}$, i.e.,
    \begin{equation}
        \label{eq:eligible-proof}
        C
        \left(\ket{0}\otimes \ket{v_0} + \lambda \ket{1} \otimes P \ket{v_1}\right)
        =
        \ket{0}\otimes \ket{v_0} + \ket{1} \otimes \highlim \ket{v_1}
        .
    \end{equation}
    We show that if $\highlim$ satisfies eq.~\eqref{eq:eligible-proof}, then it has a decomposition as in eq.~\eqref{eq:eligible-high-label}.
    We write $C = \qtop \otimes \brest$ where $\qtop$ is a single-qubit Pauli operator and $\brest$ is an $(n-1)$-qubit Pauli LIM (or a complex number $\neq 0$ if $n=1$).
    We treat the two cases $\qtop \in \{\id, Z\}$ and $\qtop \in \{X, Y\}$ separately:

\begin{enumerate}[(a)]
\item \textbf{Case} $\qtop \in \{\id, Z\}$.
    Then $\qtop = \begin{smallmat} 1& 0\\0 & (-1)^y\end{smallmat}$ for $y \in \{0, 1\}$.
        In this case, \autoref{eq:eligible-proof} implies $\qtop\ket 0\brest\ket{v_0}=\ket 0\ket{v_0}$, so $\brest\ket{v_0}=\ket{v_0}$, in other words $\brest \in \Aut(\ket{v_0})$.
        Moreover, \autoref{eq:eligible-proof} implies $(-1)^y\lambda \brest P \ket{v_1} = \highlim \ket{v_1}$, or, equivalently, $(-1)^{-y} \lambda^{-1} P^{-1} \brest^{-1} \highlim \in \Aut(v_1)$.
        Hence, by choosing $s = y$ and $x = 0$, we compute
        \[
            (-1)^y \lambda^{(-1)^0} \underbrace{\brest}_{\in \Aut(v_0)} P \underbrace{(-1)^{-y} \lambda^{-1} P^{-1} \brest^{-1} \highlim}_{\in \Aut(v_1)}
            =
            \frac{(-1)^{y} \lambda^{(-1)^0}}{ (-1)^y \lambda} \highlim
            =
            \highlim
        \]

\item \textbf{Case} $\qtop \in \{X, Y\}$.
    Write $\qtop = \begin{smallmat}0& z^{-1}\\ z&0\end{smallmat}$ where $z \in \{1, i\}$. Now, eq.~\eqref{eq:eligible-proof} implies
		\begin{equation}
		    \label{eq:z}
            z \brest \ket{v_0} = \highlim \ket{v_1}
		\qquad\textnormal{and}\qquad
            z^{-1} \lambda \brest P \ket{v_1} = \ket{v_0}.
		\end{equation}
From ~\autoref{eq:z}, we first note that $\ket{v_0}$ and $\ket{v_1}$ are isomorphic, so by Corollary~\ref{thm:node-canonicity-strong}, and because the diagram has merged these two nodes, we have $v_0 = v_1$.
    Consequently, we find from \autoref{eq:z} that $z^{-1}\brest^{-1} \highlim \in \Aut(v_0)$ and $z^{-1}\lambda \brest P \in \Aut(v_1)$.
    Now choose $x=1$ and choose $s$ such that $(-1)^s \cdot z^{-2} \brest^{-1} \highlim \brest = \highlim$ (recall that Pauli LIMs either commute or anticommute, so $\highlim\brest = \pm \brest \highlim$).
    This yields:
    \[
        (-1)^s \lambda^{-1} \cdot \underbrace{z^{-1}\brest^{-1} \highlim}_{\in \Aut(v_0)} \cdot P \cdot \underbrace{z^{-1} \lambda P \brest}_{\in \Aut(v_1)}
        =
        \lambda^{-1} \cdot \lambda \cdot
        (-1)^s 
        z^{-2} \cdot \left(\brest^{-1} \highlim \brest \right)
        =
        \highlim
    \]
    where we used the fact that $P^2 = \id^{\otimes (n-1)}$ because $P$ is a Pauli string.
\end{enumerate}
\end{proof}

\begin{corollary}\label{cor:highlabel}
As a corollary of \autoref{thm:eligible-isomorphisms-pauli}, we find that taking, as in \autoref{fig:reduced},
\[
\highlabel(\lnode[v] {\id}{v_0}{\lambda P}{v_1}) = \displaystyle \min_{\hspace{-3mm} i,s,x \in \{0, 1\},g_i \in \Aut(v_i)}(\set{(-1)^s \cdot \lambda^{(-1)^x} \cdot g_0 \cdot P \cdot g_1 \mid 
         x \neq 1 \text{ if } v_0 \neq v_1 })
\]
yields a proper implementation of \highlabel as required by \autoref{def:reduced-limdd},
because it considers all possible \highlim such that
$\ket v \simeq_{\Pauli} \ket{0}\ket{v_0}+\ket{1}\otimes \highlim \ket{v_1}$.
\end{corollary}

A naive implementation for \textsc{GetLabels} would follow the possible decompositions of eligible LIMs (see ~\autoref{eq:eligible-high-label}) and attempt to make this LIM smaller by greedy multiplication, first  with stabilizers of
$g_0 \in \Aut(v_0)$, and then with stabilizers $g_1\in \Aut(v_1)$.
To see why this does not work, consider the following example:
the high edge label is $Z$ and the stabilizer subgroups are $\Aut(v_0) = \langle X\rangle$ and $\Aut(v_1) = \langle Y \rangle$.
Then the naive algorithm would terminate and return $Z$ because $X, Y> Z$, which is incorrect since the high-edge label $X \cdot Z \cdot Y = -i \id$ is smaller than $Z$.

\begin{algorithm}
	\caption{
		Algorithm for finding LIMs $\highlim$ and $\rootlim$ required by \makeedge.
		Its parameters represent a semi-reduced node $\lnode[v]{\id}{v_0}{\lambda P}{v_1}$
		and it returns LIMs $\highlim, \rootlim$ such that $\ket v = \rootlim \ket{w}$
		with $\lnode[w]{\id}{v_0}{\highlim}{v_1}$.
		The LIM $\highlim$ is chosen canonically as the lexicographically smallest from 
		the set characterized in \autoref{thm:eligible-isomorphisms-pauli}.
		It runs in $O(n^3)$-time (with $n$ the number of qubits),
		provided $\getautomorphisms$ has been computed for children $v_0, v_1$
		(an amortized cost).
	}
    \label{alg:find-canonical-edges}
	\begin{algorithmic}[1]
		\Procedure{GetLabels}{\textsc{PauliLim} $\lambda P$, \Node $v_0, v_1$ \textbf{with} $\lambda \neq 0$ and $v_0, v_1$ reduced}
		\Statex \textbf{Output}: canonical high label $\highlim$ and root label $\rootlim$
		\State $G_0, G_1 := \getautomorphisms(v_0), \getautomorphisms(v_1)$
		\State $(g_0, g_1) := \textsc{ArgLexMin}(G_0, G_1, \lambda P)$
		\label{line:getlabels-argmin}
		\If{$v_0=v_1$}
		\label{algline:getlabels-start-minimizing}
		\State $(x,s):=\displaystyle\argmin_{(x,s)\in\{0,1\}^2} \set{(-1)^s\lambda^{(-1)^x}g_0Pg_1}$
		\Else
		\State $x:=0$
		\State $s:=\displaystyle \argmin_{s\in \{0,1\}} \set{ (-1)^s\lambda g_0Pg_1 }$
		\label{line:minimized-lim} 
		\EndIf
		\Stateh $\highlim := (-1)^s \cdot \lambda^{(-1)^x} \cdot g_0 \cdot P \cdot g_1$
		\State $\rootlim := (X \otimes \lambda P)^{x} \cdot (Z^s \otimes (g_0)^{-1})$
		\State \Return $(\highlim, \rootlim)$
		\EndProcedure
	\end{algorithmic}
\end{algorithm}

To overcome this, we consider the group closure of \emph{both} $\Aut(v_0)$ \emph{and} $\Aut(v_1)$.
See \autoref{alg:find-canonical-edges} for the $O(n^3)$-algorithm for \textsc{GetLabels}, which proceeds in two steps.
In the first step (\autoref{line:getlabels-argmin}), we use the subroutine \textsc{ArgLexMin} for finding the minimal Pauli LIM $A$ such that $A = \lambda P \cdot g_0 \cdot g_1$ for $g_0\in \Aut(v_0), g_1\in \Aut(v_1)$.
We will explain and prove correctness of this subroutine below in \autoref{sec:lexmin}.
In the second step (\autoref{algline:getlabels-start-minimizing}-\ref{line:minimized-lim}), we follow \autoref{cor:highlabel} by also minimizing over $x$ and $s$.
Finally, the algorithm returns $\highlim$, the minimum of all eligible edge labels according to \autoref{cor:highlabel}, together with a root edge label $\rootlim$ which ensures the represented quantum state remains the same.

Below, we will explain $O(n^3)$-time algorithms for finding generating sets for the stabilizer subgroup of a reduced node and for \textsc{ArgLexMin}.
Since all other lines in \autoref{alg:find-canonical-edges} can be performed in linear time, its overall runtime is $O(n^3)$.

\subsection{Constructing the stabilizer subgroup of a \limdd node}
\label{sec:pauli-isomorphism-detection}

In this section, we give a recursive subroutine \getautomorphisms to construct the stabilizer subgroup $\Aut(\ket{v}) = \set{A \in \paulilim_n \mid A\ket{v} = \ket{v} }$ of an $n$-qubit \limdd node~$v$ (see \autoref{sec:preliminaries}).
This subroutine is used by the algorithm \textsc{GetLabels} to select a canonical label for the high edge and root edge.
If the stabilizer subgroup of $v$'s children have been computed already, \getautomorphisms's runtime is $O(n^3)$.
\getautomorphisms returns a generating set for the group $\Stab(\ket{v})$.
Since these stabilizer subgroups are generally exponentially large in the number of qubits $n$, but they have at most $n$ generators, storing only the generators instead of all elements may save an exponential amount of space.
Because any generator set $G$ of size $|G|>n$ can be brought back to at most $n$ generators in time $\oh(|G| \cdot n^2)$ (see \autoref{app:prelims-linear-algebra}), we will in the derivation below show how to obtain generator sets of size linear in $n$ and leave the size reduction implicit.
We will also use the notation $A \cdot G$ and $G \cdot A$ to denote the sets $\set{A \cdot g \mid g\in G }$ and $\set{g \cdot A \mid g \in G}$, respectively, when $A$ is a Pauli LIM.

We now sketch the derivation of the algorithm.
The base case of the algorithm is the Leaf node of the \limdd, representing the number $1$, which has stabilizer group $\{1\}$.
For the recursive case, we wish to compute the stabilizer group of a reduced $n$-qubit node $v=\lnode[v]{\mathbb I}{v_0}{\highlim}{v_1}$.
If $\highlim=0$, then it is straightforward to see that $\lambda P_n \otimes P'\ket{v} = \ket{v}$ implies $P_n \in \{\id, Z\}$, and further that $\Aut(\ket{v}) = \langle \set{P_n \otimes g \mid g\in G_0, P_n \in \set{\id, Z} } \rangle$, where $G_0$ is a stabilizer generator set for $v_0$.

Otherwise, if $\highlim \neq 0$, then we expand the stabilizer equation $\lambda P \ket{v} = \ket{v}$:
\[
\lambda P_n \otimes P' \left(\ket 0 \otimes\ket{v_0} + \ket 1 \otimes \highlim \ket{v_1} \right)  = \ket 0 \otimes\ket{v_0} +  \ket 1 \otimes \highlim \ket{v_1}, \text{which implies:}
\]
\begin{align}
  \lambda P' \ket{v_0} =  \ket{v_0}  ~&\text{and } z \lambda P' \highlim \ket{v_1} = \highlim \ket{v_1}
            & \textbf{if } P_n= \diag z\text{ with }z\in\set{1,-1} \label{eq:diag} \\
  y^* \lambda P' \highlim \ket{v_1} =  \ket{v_0}  &\text{and } \lambda P' \ket{v_0} = y^*\highlim \ket{v_1}
             & \textbf{if } P_n= \yy,\text{ with } y\in\set{1,i} 
             \label{eq:anti}
\end{align}
The stabilizers can therefore be computed according to \autoref{eq:diag} and \ref{eq:anti} as follows.
\begin{align}     
    \nonumber
    \Aut(\ket{v}) =
    \bigcup_{\hspace{-8mm}z = \in\set{1,-1} , y \in\set{1, i}\hspace{-8mm}}&
          \diag z \otimes ( \Aut(\ket{v_0}) \cap z\cdot \Aut(\highlim\ket{v_1}) )
          \\
          &
 \cup 
          \
    \ww \otimes  
    \big(
    \Iso( y^* \highlim\ket{v_1}, \ket{v_0})  \cap \Iso( \ket{v_0}, y^* \highlim\ket{v_1})  
    \big) 
    \label{eq:aut-a}
\end{align}
where $\Iso(v, w)$ denotes the set of Pauli isomorphisms $A$ which map $\ket{v}$ to $\ket{w}$ and we have denoted $\pi \cdot G := \set{\pi \cdot g \mid g \in G }$ for a set $G$ and a single operator $\pi$.
\autoref{lemma:isomorphism-set-characterization} shows that such an isomorphism set can be expressed in terms of the stabilizer group of $\ket{v}$.

\def\Pauli{\textnormal{\textsc{Pauli}}}
\begin{lemma}
    \label{lemma:isomorphism-set-characterization}
    Let $\ket{\phi}$ and $\ket{\psi}$ be quantum states on the same number of qubits.
    Let $\pi$ be a Pauli isomorphism mapping $\ket{\phi}$ to $\ket{\psi}$.
    Then the set of Pauli isomorphisms mapping $\ket{\phi}$ to $\ket{\psi}$ is
    $\Iso(\ket{v},\ket{w})=\pi \cdot \Aut(\ket{\phi})$.
    That is, the set of isomorphisms $\ket{\phi} \rightarrow \ket{\psi}$ is a coset of the stabilizer subgroup of $\ket{\phi}$.
\end{lemma}
\begin{proof}
    If $P\in \Aut(\ket{\phi})$, then $\pi \cdot P$ is an isomorphism since $\pi \cdot P \ket{\phi} = \pi \ket{\phi} = \ket{\psi}$.
    Conversely, if $\sigma$ is a Pauli isomorphism which maps $\ket{\phi}$ to $\ket{\psi}$, then $ \pi^{-1} \sigma \in \Aut(\ket{\phi})$ because $\pi^{-1} \sigma \ket{\phi} = \pi^{-1} \ket{\psi} = \ket{\phi}$.
    Therefore $\sigma=\pi(\pi^{-1}\sigma)\in \pi \cdot \Aut(\ket{\phi})$.
\end{proof}
With \autoref{lemma:isomorphism-set-characterization} we can rewrite eq.~\eqref{eq:aut-a} as
\begin{align}     
    \nonumber
    \Aut(\ket{v}) =& 
          \id \otimes \underbrace{( \Aut(\ket{v_0}) \cap \Aut(\highlim\ket{v_1}) )}_{\textnormal{stabilizer subgroup}} \\
          & \cup Z \otimes \underbrace{( \id \cdot \Aut(\ket{v_0}) \cap -\id \cdot \Aut(\highlim\ket{v_1}) )}_{\textnormal{isomorphism set}} 
          \nonumber
          \\
          &\cup 
          \bigcup_{y \in\set{1, i}}
          \
    \ww \otimes  \underbrace{
        \big(
        \pi
        \cdot
        \Aut(y^* \highlim \cdot \ket{v_1})
        \cap
        \pi^{-1}
        \cdot
        \Aut(\ket{v_0})}_{\textnormal{isomorphism set}}
        \big)
    \label{eq:aut-simplified}
\end{align}
where $\pi$ denotes a single isomorphism $y^* \highlim\ket{v_1}  \rightarrow \ket{v_0}$.

Given generating sets for $\Aut(v_0)$ and $\Aut(v_1)$, evaluating eq.~\eqref{eq:aut-simplified} requires us to:

\begin{itemize}
        \setlength\itemsep{1em}
    \item \textbf{Compute $\Aut(A\ket{w})$ from $\Aut(w)$ (as generating sets) for Pauli LIM $A$ and node $w$.} It is straightforward to check that $\set{A g A^{\dagger} \mid g \in G}$, with $\langle G \rangle = \Aut(w)$, is a generating set for $\Aut(A\ket{w})$.
    \item \textbf{Find a single isomorphism between two edges, pointing to reduced nodes.} In a reduced \limdd, edges represent isomorphic states if and only if they point to the same nodes. This results in a straightforward algorithm, see \autoref{alg:getsingleisomorphism}.
    \item \textbf{Find the intersection of two stabilizer subgroups, represented as generating sets $G_0$ and $G_1$ (\textsc{IntersectStabilizerGroups}, \autoref{alg:intersectstabilizergroups}).}
        First, it is straightforward to show that the intersection of two stabilizer subgroups is again a stabilizer subgroup (namely, it is abelian and does not contain $-\mathbb I$. It is never empty since $\id$ is a stabilizer of all states).%
        \footnote{To be clear, here we consider the stabilizers including their phase, i.e., we are not considering the groups modulo phase. Indeed, computing the intersection of two groups modulo phase is relatively easy, as shown in \autoref{app:prelims-linear-algebra}.}
        The algorithm proceeds in two steps: first, we compute the intersection of $G_0$ and $G_1$ modulo phase; second, we ``correct for'' the fact that the phases play a role.
        
        Very broadly speaking, we use the following algebraic properties of Pauli groups.
        First, when considering a Pauli string $\lambda P$ modulo phase, it is convenient to think of it as simply the Pauli string $P$ with phase equal to $+1$.
        This allows us to take an element $P\in \overline{G_0}$ from a Pauli stabilizer group $\overline{G_0}$ modulo phase, and, by abuse of language, multiply it by a phase $\lambda\in \mathbb C$ to obtain $\lambda\cdot P\in G_0$.
        Second, for any Pauli string $P\in \overline{G_0}\cap \overline{G_1}$, i.e., in the group modulo phase, there exists a unique $\lambda$ such that $\lambda\cdot P\in \gen{G_0}$; and a unique $\omega$ such that $\omega\cdot P\in \gen{G_1}$.
        Moreover, we have $\omega=\pm \lambda$ in this case due to anti-commutativity.
        Consequently, if $S=\{P_1,\ldots, P_\ell\}$ is a generating set for the group $\braket{S}=\gen{\overline{G_0}}\cap \gen{\overline{G_1}}$ modulo phase, then we can divide these generators $P_j$ into two sets, $S=S_0\cup S_1$, where each $P\in S_0$ satisfies $\lambda P\in G_0\cap G_1$ for some $\lambda\in \mathbb C$ and each $P\in S_1$ satisfies $\lambda P\in G_0$ and $-\lambda P\in G_1$.
        The algorithm finds these sets $S_0$ and $S_1$, including phase, in the loop in \autoref{algline:group-intersect-start-first-loop}-\ref{algline:group-intersect-end-first-loop}.
        Given such sets $S_0,S_1$, any element in $G_0\cap G_1$ (i.e., the set we are interested in) can be written as a product of elements of $S_0$ and an \emph{even number} of elements from $S_1$.
        The set $\{e_1\cdot e_2,\ldots, e_1\cdot e_m\}$, found by the algorithm in \autoref{algline:group-intersect-start-second-loop}-\ref{algline:group-intersect-end-second-loop}, generates precisely this set of elements generated from an even number of elements of $S_1$.

		All of the above steps can be performed in $\mathcal O(n^3)$ time, where $n$ is the number of qubits.
		In particular, a generating set for the intersection of $G_0$ and $G_1$ modulo phase is simply the intersection of two vector spaces over $\mathbb F_2$, which is constructed in time $\mathcal O(n^3)$ on \autoref{algline:group-intersect-construct-intersection-mod-phase} using the Zassenhaus algorithm.
		On line \autoref{algline:group-intersect-membership-check}, checking whether $\lambda P\in G_0$ for given $\lambda,P$ can be done in $\mathcal O(n^2)$ time; this happens at most $\mathcal O(|S|)=\mathcal O(n)$ times, so in total this operation takes up $\mathcal O(n^3)$ time.
		Lastly, the loop in \autoref{algline:group-intersect-start-second-loop}-\ref{algline:group-intersect-end-second-loop} runs in $\mathcal O(n^2)$ time, as there are at most $|S|-1=\mathcal O(n)$ multiplications of Pauli strings, each of which takes $\mathcal O(n)$ time.
		We remark that the Zassenhaus algorithm cannot be straightforwardly applied to find the intersection of the groups $G_0$ and $G_1$ directly, since the elements of $G_0$ may not commute with those~of~$G_1$.
		
    \item \findisomorphismsetintersection: \textbf{Find the intersection of two isomorphism sets, represented as single isomorphism ($\pi_0, \pi_1$) with a generator set of a stabilizer subgroup ($G_0, G_1$), see \autoref{lemma:isomorphism-set-characterization}.} 
        This is the \emph{coset intersection problem} for the $\paulilim_n$ group.
        Isomorphism sets are coset of stabilizer groups (see \autoref{lemma:isomorphism-set-characterization}) and it is not hard to see that that the intersection of two cosets, given as isomorphisms $\pi_{0/1}$ and generator sets $G_{0/1}$, is either empty, or a coset of $\langle G_0 \rangle \cap \langle G_1 \rangle$ (this intersection is computed using \autoref{alg:intersectstabilizergroups}).
        Therefore, we only need to determine an isomorphism $\pi \in \pi_0 \langle G_0\rangle \cap \pi_1 \langle G_1 \rangle$, or infer that no such isomorphism exists.

        We solve this problem in $O(n^3)$ time in two steps (see \autoref{alg:findisointersection} for the full algorithm).
First, we note that that $\pi_0 \langle G_0 \rangle \cap \pi_1 \langle G_1 \rangle = \pi_0 [\langle G_0 \rangle \cap (\pi_0^{-1} \pi_1) \langle G_1 \rangle]$, so we only need to find an element of the coset $S:= \langle G_0 \rangle \cap (\pi_0^{-1} \pi_1) \langle G_1 \rangle$.
        Now note that $S$ is nonempty if and only if there exists $g_0 \in \langle G_0 \rangle, g_1 \in \langle G_1 \rangle$ such that $g_0 = \pi_0^{-1} \pi_1 g_1$, or, equivalently, $\pi_0^{-1} \pi_1 \cdot g_1 \cdot g_0^{-1} = \id$.
        We show in \autoref{lemma:id-smallest-in-coset} that such $g_0, g_1$ exist if and only if $\id$ is the smallest element in the set $S\pi_0^{-1}\pi_1\braket{G_1}\cdot\braket{G_0}$.
        Hence, for finding out if $S$ is empty we may invoke the \textsc{LexMin} algorithm we have already used before in \textsc{GetLabels} and we will explain below in \autoref{sec:lexmin}.
        If it is not empty, then we obtain $g_0, g_1$ as above using \textsc{ArgLexMin}, and output $\pi_0 \cdot g_0$ as an element in the intersection.
        Since \textsc{Lexmin} and \textsc{ArgLexMin} take $O(n^3)$ time, so does \autoref{alg:findisointersection}.
\end{itemize}

\begin{lemma}
    \label{lemma:id-smallest-in-coset}
    The coset $S:= \langle G_0 \rangle \cap \pi_1^{-1}\pi_0 \cdot \langle G_1 \rangle$ is nonempty if and only if the lexicographically smallest element of the set $S=\pi_0^{-1}\pi_1\braket{G_1}\cdot\braket{G_0}=\set{\pi_0^{-1}\pi_1g_1g_0 \mid g_0\in G_0,g_1\in G_1 }$ is $1 \cdot \id$.
\end{lemma}
\begin{proof}
	(Direction $\Rightarrow$)
	Suppose that the set $\braket{G_0}\cap \pi_0^{-1}\pi_1\braket{G_1}$ has an element $a$.
	Then $a=g_0=\pi_0^{-1}\pi_1g_1$ for some $g_0\in \braket{G_0},g_1\in\braket{G_1}$.
	We see that $\mathbb I=\pi_0^{-1}\pi_1g_1g_0^{-1}\in \pi_0^{-1}\pi_1\braket{G_1}\cdot \braket{G_0}$, i.e., $\mathbb I\in S$.
	Note that $\mathbb I$ is, in particular, the lexicographically smallest element, since its check vector is the all-zero vector $(\vec 0|\vec 0|00)$.
	
	(Direction $\Leftarrow$)
	Suppose that $\mathbb I\in \pi_0^{-1}\pi_1\braket{G_1}\cdot\braket{G_0}$.
	Then $\mathbb I=\pi_0^{-1}\pi_1g_1g_0$, for some $g_0\in \braket{G_0},g_1\in\braket{G_1}$, so we get $g_0^{-1}=\pi_0^{-1}\pi_1g_1\in \braket{G_0}\cap \pi_0^{-1}\pi_1\braket{G_1}$, as promised.
\end{proof}

The four algorithms above allow us to evaluate each of the four individual terms in eq.~\eqref{eq:aut-simplified}.
To finish the evaluation of eq.~\eqref{eq:aut-simplified}, one would expect that it is also necessary that we find the union of isomorphism sets.
However, we note that if $\pi G$ is an isomorphism set, with $\pi$ an isomorphism and $G$ an stabilizer subgroup, then $P_n \otimes (\pi g) = (P_n \otimes \pi) (\id \otimes g)$ for all $g\in G$.
Therefore, we will evaluate eq.~\eqref{eq:aut-simplified}, i.e. find (a generating set) for all stabilizers of node $v$ in two steps.
First, we construct the generating set for the first term, i.e. $\id \otimes ( \Aut(\ket{v_0}) \cap \Aut(\highlim \ket{v_1}) )$, using the algorithms above.
Next, for each of the other three terms $P_n \otimes (\pi G)$, we add only \textit{a single} stabilizer of the form $P_n \otimes \pi$ for each $P_n \in \{X, Y, Z\}$.
We give the full algorithm in \autoref{alg:getautomorphisms} and prove its efficiency below.

\begin{lemma}[Efficiency of function \getautomorphisms]
    Let $v$ be an $n$-qubit node.
    Assume that generator sets for the stabilizer subgroups of the children $v_0, v_1$ are known, e.g., by an earlier call to \getautomorphisms, followed by caching the result (see \autoref{line:autocache-store} in \autoref{alg:getautomorphisms}).
   Then \autoref{alg:getautomorphisms} (function \getautomorphisms), applied to $v$, runs in time $O(n^3)$.
\end{lemma}
\begin{proof}
    If $n=1$ then \autoref{alg:getautomorphisms} only evaluates \autoref{line:stabalgo-first}--\ref{line:stabalgo-second}, which run in constant time.
    For $n>1$, the algorithm performs a constant number of calls to \getsingleisomorphism (which only multiplies two Pauli LIMs and therefore runs in time $O(n)$) and four calls to \findisomorphismsetintersection.
    Note that the function \findisomorphismsetintersection~from \autoref{alg:findisointersection} invokes $O(n^3)$-runtime external algorithms:
    \begin{itemize}
\item the Zassenhaus algorithm~\cite{LUKS1997335} to calculate a basis for the intersection  of two subspaces of a vector space,
\item the RREF algorithms mentioned in \autoref{app:prelims-linear-algebra}, and
\item  Algorithm 2 from \cite{garcia2012efficient} to synthesize a circuit that transforms any stabilizer state to a basis state.
	Specifically, this algorithm receives as input a stabilizer subgroup $G$ and outputs a Clifford circuit $U$ such that $UGU^\dagger=\{Z_1,\ldots, Z_{|G|}\}$.
	We remark that Garc\'ia et al. assume in their work that $G$ is the stabilizer group of a stabilizer state, i.e., $|G|=n$, but in fact the algorithm works also without that assumption, i.e., in the more general case when $G$ is any abelian group of Pauli operators not containing $-\mathbb I$.
	Our algorithms use this more general use case.
    \end{itemize}
    Therefore, \getautomorphisms has runtime is $O(n^3)$.
\end{proof}

\begin{algorithm}
    \caption{Algorithm for constructing the Pauli stabilizer subgroup of a Pauli-\limdd~node.
    The algorithm always returns a set in reduced row echelon form (see \autoref{app:prelims-linear-algebra}), which is accomplished in line \ref{algline:stabgenset-to-RREF}.
    In particular, the set always returns at most $n$ elements for $n$-qubit states.}
    \label{alg:getautomorphisms}
    \begin{algorithmic}[1]
        \Procedure{\getautomorphisms}{\Edge $\ledge[e_0]{\id^{\otimes n}}{v_0}, \ledge[e_1]{\highlim}{v_1}$ \textbf{with} $v_0, v_1$ reduced}
        \If{n=1}
        \label{line:stabalgo-first}
        \If{ there exists $P \in \pm 1 \cdot \{X, Y, Z\}$ \textbf{such that} $P \ket v = \ket v$} \Return $P$ \Else \mbox{ } \Return \none
        \label{line:stabalgo-second}
        \EndIf
        \Else \algstrut[1]      
        \If{$v \in \autocache[v]$}
        \Return $\autocache[v]$
        \EndIf
         \Stateh $G_0 := \getautomorphisms(v_0)$
        \If{$\highlim = 0$}
       	     \State \Return $\set{\mathbb I_2\otimes g \mid g\in G_0}\cup \{Z\otimes \mathbb I^{\otimes n-1}\}$
             \label{line:stab-fork}
        \Else
       	\State $G:= \emptyset$
        \State $G_1 := \set{\highlim \cdot g\cdot \highlim^\dagger \mid g \in \getautomorphisms(v_1) }$
        \State $(\pi, B):= \findisomorphismsetintersection(( \id^{\otimes n - 1}, G_0), ( \id^{\otimes n - 1}, G_1))$
        \State $G := G \cup \set{\mathbb I_2\otimes g  \mid g\in B}$
        \Comment {Add all stabilizers of the form $\id \otimes \dots$}
        \Stateh
        \State $\pi_0, \pi_1 := \id^{\otimes n - 1}, \getsingleisomorphism(e_1,-1 \cdot e_1)$ 
        \State $(\pi, B):= \findisomorphismsetintersection((\pi_0, G_0), (\pi_1, G_1))$
         \If{$\pi \neq \text{None}$ }   $G := G \cup \{Z\otimes \pi\}$ 
     	         \Comment Add stabilizer of form $Z \otimes \dots$
     	 \EndIf
        \Stateh
        \State $\pi_0, \pi_1 := \getsingleisomorphism(e_0,e_1), \getsingleisomorphism(e_1, e_0))$
        \State $(\pi, B):= \findisomorphismsetintersection((\pi_0, G_0), (\pi_1, G_1))$
         \If{$\pi \neq \text{None}$ }   $G := G \cup \{X\otimes \pi\}$ 
     	         \Comment Add stabilizer of form $X \otimes \dots$
     	 \EndIf 
        \Stateh
        \State $\pi_0, \pi_1 := \getsingleisomorphism(e_0, -i \cdot e_1), \getsingleisomorphism(-i\cdot e_1, e_0))$
        \State $(\pi, B):= \findisomorphismsetintersection((\pi_0, G_0), (\pi_1, G_1))$
         \If{$\pi \neq \text{None}$ }   $G := G \cup \{Y\otimes \pi\}$ 
     	         \Comment Add stabilizer of form $Y \otimes \dots$
     	 \EndIf
     	\State $G:=RREF(G)$ \Comment{Bring $G$ to reduced row echelon form, potentially pruning some rows}
        	\label{algline:stabgenset-to-RREF}
        \Stateh $\autocache[v] := G$
        \label{line:autocache-store}
		\State \Return $G$
        \EndIf
        \EndIf
        \EndProcedure
    \end{algorithmic}
\end{algorithm}

\begin{algorithm}
    \caption{An $O(n^3)$ algorithm for computing the intersection of two sets of isomorphisms, each given as single isomorphism with a stabilizer subgroup (see \autoref{lemma:isomorphism-set-characterization}).
    \label{alg:findisointersection}
    }
  	\Output Pauli LIM $\pi$, stabilizer subgroup generating set $G$ s.t. $\pi \langle G \rangle = \pi_0 \langle G_0 \rangle \cap \pi_1 \langle G_1 \rangle$
    \begin{algorithmic}[1]
        \Procedure{IntersectIsomorphismSets}{stabilizer subgroup generating sets $G_0, G_1$,
        
        ~~~~~~~~~~~~~~~~~~~~~~~~~~~~~~~~~~~~~~~~~~~~~~~~~~~ Pauli-LIMs $\pi_0, \pi_1$}
        \State $\pi := LexMin(G_0, G_1, \pi_1^{-1}\pi_0)$
        \If{$\pi = \id$}
        \State $(g_0, g_1) = ArgLexMin(G_0, G_1, \pi_1^{-1}\pi_0)$
        \State $\pi := \pi_0 \cdot g_0$
        \State $G := IntersectStabilizerGroups(G_0,G_1)$
        \State \Return $(\pi, G)$
        \Else
        \State \Return \none
        \EndIf
        \EndProcedure
    \end{algorithmic}
\end{algorithm}

\begin{algorithm}
    \caption{Algorithm for finding the intersection $\braket{G_0}\cap \braket{G_1}$ of the groups generated by two stabilizer subgroup generating sets $G_0$ and $G_1$.}
    \label{alg:intersectstabilizergroups}
    \Output a generating set for $\langle G_0 \rangle \cap \langle G_1 \rangle$
    \begin{algorithmic}[1]
        \Procedure{IntersectStabilizerGroups}{stabilizer subgroup generating sets $G_0, G_1$}
        \State\textit{Use the Zassenhaus algorithm to compute the intersection modulo phase}
        \State $S:=\textsc{IntersectGroupsModuloPhase}(G_0,G_1)$%
        \label{algline:group-intersect-construct-intersection-mod-phase}
        \State $J, S_0, S_1:=\emptyset$
		\For{$P\in S$}    \label{algline:group-intersect-start-first-loop}
        	\State \textit{By abuse of language, we treat $P$ as a Pauli string with phase $+1$}
        	\State Find $\lambda\in \{\pm 1,\pm i\}$ such that $\lambda\cdot P\in \braket{G_0}$
        	\label{algline:group-intersect-membership-check}
        	\If{$\lambda \cdot P\in \braket{G_1}$}
        		\State $S_0:=S_0\cup\{\lambda \cdot P\}$
        	\ElsIf{$-\lambda\cdot P\in \braket{G_1}$}
       			\State $S_1:=S_1\cup \{\lambda \cdot P\}$
        	\EndIf
        \EndFor
        \label{algline:group-intersect-end-first-loop}
        \State $J:=S_0$
		\If {$\exists e\in S_1$}
			\For{$e'\in S_1\setminus \set{e}$}
	        	\label{algline:group-intersect-start-second-loop}
	        	\State $q:=e'\cdot e$ %
	        	\State $J:=J\cup \{q\}$
	        \EndFor
        \EndIf
        \label{algline:group-intersect-end-second-loop}
        
        \State \Return $J$
        \EndProcedure
    \end{algorithmic}
\end{algorithm}

\begin{algorithm}
    \caption{Algorithm for constructing a single isomorphism between the quantum states represented by two Pauli-\limdd~edges, each pointing to a canonical node.
    }
    \label{alg:getsingleisomorphism}
    \begin{algorithmic}[1]
        \Procedure{\getsingleisomorphism}{\Edge $\ledge Av$,  $\ledge Bw$ \textbf{with} $v, w$ reduced, \mbox{$A\neq 0 \vee B \neq 0$}}
        \If{$v = w \textbf{ and } A,B \neq 0 $}
        \State \Return $B \cdot A^{-1}$
        \EndIf
        \State \Return \none
        \EndProcedure
    \end{algorithmic}
\end{algorithm}

\subsection{Efficiently finding a minimal LIM by multiplying with stabilizers}
\label{sec:lexmin}

Here, we give $O(n^3)$ subroutines solving the following problem: given generators sets $G_0, G_1$ of stabilizer subgroups on $n$ qubits, and an $n$-qubit Pauli LIM $A$, determine $\min_{(g_0, g_1) \in \langle G_0 , G_1 \rangle} A \cdot g_0 \cdot g_1$, and also find the $g_0, g_1$ which minimize the expression.
We give an algorithm for finding both the minimum (\textsc{LexMin}) and the arguments of the minimum (\textsc{ArgLexMin}) in \autoref{alg:lexmin}.
The intuition behind the algorithms are the following two steps: first, the lexicographically minimum Pauli LIM \emph{modulo scalar} can easily be determined using the scalar-ignoring DivisionRemainder algorithm from \autoref{app:prelims-linear-algebra}.
Since in the lexicographic ordering, the scalar is least significant (\autoref{app:prelims-linear-algebra}), the resulting Pauli LIM has the same Pauli string as the the minimal Pauli LIM \emph{including scalar}.
We show below in \autoref{thm:pauli-group-means-pm-1} that if the scalar-ignoring minimization results in a Pauli LIM $\lambda P$, then the only other eligible LIM, if it exists, is $-\lambda P$.
Hence, in the next step, we only need to determine whether such LIM $-\lambda P$ exists and whether $- \lambda < \lambda$; if so, then $-\lambda P$ is the real minimal Pauli LIM $\in \langle G_0 \cup G_1\rangle$.

\begin{lemma}
	\label{thm:pauli-group-means-pm-1}
    Let $v_0$ and $v_1$ be \limdd nodes, $R$ a Pauli string and $\nu, \nu' \in \mathbb{C}$.
    Define $G = \Aut(v_0) \cup \Aut(v_1)$.
If $\nu R, \nu' R \in \langle G\rangle$, then $\nu = \pm \nu'$.
\end{lemma}
\begin{proof}
	We prove $g\in \braket{G}$ and $\lambda g\in\braket{G}$ implies $\lambda=\pm 1$.
    Since Pauli LIMs commute or anticommute, we can decompose both $g$ and $\lambda g$ as $g = (-1)^x g_0 g_1$ and $\lambda g = (-1)^y h_0 h_1$ for some $x, y \in \{0, 1\}$ and $g_0, h_0\in \Aut(v_0)$ and $g_1, h_1 \in \Aut(v_1)$.
    Combining these yields $\lambda (-1)^xg_0g_1=(-1)^yh_0h_1$.
    We recall that, if $g\in \Aut(\ket\phi)$ is a stabilizer of any state, then $g^2=\mathbb I$.
    Therefore, squaring both sides of the equation, we get $\lambda^2(g_0g_1)^2=(h_0h_1)^2$, so $\lambda^2\mathbb I=\mathbb I$, so $\lambda=\pm 1$.
\end{proof}

The central procedure in \autoref{alg:lexmin} is \textsc{ArgLexMin}, which, given a LIM $A$ and sets $G_0,G_1$ which generate stabilizer groups, finds $g_0\in \braket{G_0},g_1\in\braket{G_1}$ such that $A\cdot g_0\cdot g_1$ reaches its lexicographic minimum over all choices of $g_0,g_1$.
It first performs the scalar-ignoring minimization (\autoref{line:division-remainder}) to find $g_0,g_1$ modulo scalar.
The algorithm \textsc{LexMin} simply invokes \textsc{ArgLexMin} to get the arguments $g_0, g_1$ which yield the minimum and uses these to compute the actual minimum.

The subroutine \textsc{FindOpposite} finds an element $g \in G_0$ such that $-g \in G_1$, or infers that no such $g$ exists.
It does so in a similar fashion as \textsc{IntersectStabilizerGroups} from \autoref{sec:pauli-isomorphism-detection}: by conjugation with a suitably chosen unitary $U$, it maps $G_1$ to $\{Z_1, Z_2, \dots, Z_{|G_1|}\}$.
Analogously to our explanation of \textsc{IntersectStabilizerGroups}, the group generated by $UG_1 U^{\dagger}$ contains precisely all Pauli LIMs which satisfy the following three properties:
(i) the scalar is $1$;
(ii) its Pauli string has an $\id$ or $Z$ at positions $1, 2, \dots, |G_1|$;
(iii) its Pauli string has an $\id$ at positions $|G_1|+1, \dots, n$.
Therefore, the target $g$ only exists if there is a LIM in $\langle U G_0 U^{\dagger}\rangle$ which (i') has scalar $-1$ and satisfies properties (ii) and (iii).
To find such a $g$, we put $UG_0 U^{\dagger}$ in RREF form and check all resulting generators for properties (i'), (ii) and (iii).
(By definition of RREF, it suffices to check only the generators for this property)
If a generator $h$ satisfies these properties, we return $U^{\dagger} h U$ and $\none$ otherwise.
The algorithm requires $O(n^3)$ time to find $U$, the conversion $G \mapsto UGU^{\dagger}$ can be done in time $O(n^3)$, and $O(n)$ time is required for checking each of the $O(n^2)$ generators.
Hence the runtime of the overall algorithm is $O(n^3)$.

\begin{algorithm}
    \caption{
        Algorithms \textsc{LexMin} and \textsc{ArgLexMin} for computing the minimal element from the set $A \cdot \langle G_0\rangle \cdot \langle G_1\rangle=\set{Ag_0g_1 \mid g_0\in G_0,g_1\in G_1}$, where $A$ is a Pauli LIM and $G_0, G_1$ are generating sets for stabilizer subgroups.
        The algorithms make use of a subroutine \textsc{FindOpposite} for finding an element $g \in \langle G_0\rangle$ such that $-g \in \langle G_1\rangle$.
        A canonical choice for the \textsc{Rootlabel} (see \autoref{sec:applygate}) of an edge $e$ pointing to a node $v$ is $\textsc{LexMin}(G, \{\id\}, \lbl(e))$ where $G$ is a stabilizer generator group of $\Aut(v)$.
        \label{alg:lexmin}
    }
    \Output $\min_{(g_0, g_1) \in  \langle G_0 \cup G_1 \rangle} A \cdot g_0 \cdot g_1$
    \begin{algorithmic}[1]
        \Procedure{LexMin}{stabilizer subgroup generating sets $G_0, G_1$ and Pauli LIM $A$}
        \State $(g_0, g_1) := \textsc{ArgLexMin}(G_0, G_1, A)$
        \State \Return $A \cdot g_0 \cdot g_1$
        \EndProcedure
        \Statex 
             \Statex \textbf{output:} $\argmin_{g_0 \in G_0, g_1\in G_1} A \cdot g_0 \cdot g_1$
        \Procedure{ArgLexMin}{stabilizer subgroup generating sets $G_0, G_1$ and Pauli LIM $A$}
        \State $(g_0, g_1) := \displaystyle \argmin_{(g_0, g_1) \in \langle G_0 \cup G_1 \rangle} \set{h \mid h \propto A \cdot g_0 \cdot g_1 }$
        \Statex
        \Comment Using the scalar-ignoring DivisionRemainder algorithm from \autoref{app:prelims-linear-algebra},
        \label{line:division-remainder}
        \State $g' := \textsc{FindOpposite}(G_0, G_1, g_0, g_1)$
        \If{$g'$ is $\none$}
        \State \Return $(g_0, g_1)$
        \Else
        \State $h_0, h_1 := g_0 \cdot g', (-g') \cdot g_1$   
                \Comment{$g_0 g_1 = - h_0 h_1$}
        \If{$A\cdot h_0 \cdot h_1 <_{\text{lex}} A \cdot g_0 \cdot g_1$} \Return $(h_0,h_1)$
        \label{line:choose-smaller}
        \Else \ \Return $(g_0,g_1)$
        \EndIf
        \EndIf
        \EndProcedure
        \Statex
        \Statex \textbf{output:} $g\in G_0$ such that $-g \in G_1$, or \none~if no such $g$ exists
        \Procedure{FindOpposite}{stabilizer subgroup generating sets $G_0, G_1$}
        \State Compute $U$ s.t. $U G_1 U^{\dagger} = \{Z_1, Z_2, \dots, Z_{|G_1|}\}$, using Algorithm 2 from \cite{garcia2012efficient}
        \Statex
                \Comment{$Z_j$ is the $Z$ gate applied to qubit with index $j$}
        \State $H_0 := UG_0 U^{\dagger}$
        \State $H_0^{RREF} := H_0$ in RREF form
        \For{ $h \in H_0^{RREF}$}
        \If{$h$ satisfies all three of the following: (i) $h$ has scalar $-1$; the Pauli string of $h$ (ii) contains only $\id$ or $Z$ at positions $1, 2, \dots, |G_1|$, and (iii) only $\id$ at positions $|G_1|+1, \dots, n$}
        \State \Return $U^{\dagger} h U$
        \EndIf
        \EndFor
        \State \Return \none
        \EndProcedure
    \end{algorithmic}
\end{algorithm}

\section{Measuring an arbitrary qubit}
\label{sec:advanced-algorithms}
\label{subsec:measure-arbitrary-qubit}

\autoref{alg:measure-arbitrary-qubit} allows one to measure a given qubit.
Specifically, given a quantum state $\ket e$ represented by a \limdd edge $e$, a qubit index $k$ and an outcome $b\in \{0,1\}$, it computes the probability of observing $\ket{b}$ when measuring the $k$-th significant qubit of $\ket e$.
The algorithm proceeds by traversing the \limdd with root edge $e$ at \autoref{l:traverse1}.
Like \autoref{alg:measurement-top-qubit}, which measured the top qubit, this algorithm finds the probability of a given outcome by computing the squared norm of the state when the $k$-th qubit is projected onto $\ket 0$, or $\ket 1$.
The case that is added, relative to \autoref{alg:measurement-top-qubit}, is the case when $n>k$, in which case it calls the procedure \textsc{SquaredNormProjected}.
On input $e,y,k$, the procedure \textsc{SquaredNormProjected} outputs the squared norm of $\Pi_k^y\ket{e}$, where $\Pi_k^y=\id[n-k] \otimes \ket{y}\bra{y}\otimes \id[k-1]$ is the projector which projects the $k$-th qubit onto $\ket{y}$.

\begin{algorithm}[t!]
	\caption{Compute the probability of observing $\ket{y}$ when measuring the $k$-th qubit of the state $\ket{e}$. Here $e$ is given as \limdd on $n$ qubits, $y$ is given as a bit, and $k$ is an integer index.
	For example, to measure the top-most qubit, one calls $\textsc{Measure}(e,0,n)$.
	The procedure $\textsc{SquaredNorm}(e,y,k)$ computes the scalar $\bra{e}(\mathbb I\otimes \ket{y}\bra{y}\otimes \mathbb I)\ket{e}$, i.e., computes the squared norm of the state $\ket{e}$ after the $k$-th qubit is projected to $\ket{y}$. For readability, we omit calls to the cache, which implement dynamic programming.}
	\label{alg:measure-arbitrary-qubit}
	\begin{algorithmic}[1]
		\Procedure{MeasurementProbability}{\Edge
		          $\ledge[e] {\lambda P_n\otimes P^\prime}v$, $y\in \{0,1\}$, 
		          $k\in [1 \dots \index(v)]$}
		\If{$n=k$}
			\State $p_0:=\textsc{SquaredNorm}(\follow 0e)$
			\State $p_1:=\textsc{SquaredNorm}(\follow 1e)$
			\State \Return $p_j/(p_0+p_1)$ \textbf{where} $j=0$ if $P_n\in\{\mathbb I, Z\}$ and $j=1$ if $P_n \in\set{X,Y}$ 
		\Else
			\State $p_0:=\textsc{SquaredNormProjected}(\follow 0e, y, k)$ \label{l:traverse1}
			\State $p_1:=\textsc{SquaredNormProjected}(\follow 1e, y, k)$
			\State \Return $(p_0+p_1)/\textsc{SquaredNorm}(e)$
		\EndIf
		\EndProcedure
		\Procedure{SquaredNorm}{$\Edge \ledge{\lambda P}{v}$}
			\If{$n=0$}
			\Return $|\lambda|^2$
			\EndIf
			\State $s:=\textsc{Add}(\textsc{SquaredNorm}(\follow 0{\ledge \id v}),\textsc{SquaredNorm}(\follow 1{\ledge \id v}))$
			\State \Return $|\lambda|^2s$
		\EndProcedure
		\Procedure{SquaredNormProjected}{\Edge $\ledge[e] {\lambda P_n\otimes P'}v$, $y\in\{0,1\}$, $k\in [1\dots \index(v)]$}
		\State $b:=(P_n \in\set{X,Y})$
						         \Comment{i.e., $b = 1$ iff $P_n$ is Anti-diagonal}
			\If{$n=0$}
				\State \Return $|\lambda|^2$
			\ElsIf{$n=k$}
				\State \Return $\textsc{SquaredNorm}(\follow{b\oplus y}{e})$
			\Else
				\State $\alpha_0:=\textsc{SquaredNormProjected}(\follow 0{\ledge {\mathbb I}v}, b\oplus y, k)$
				\State $\alpha_1:=\textsc{SquaredNormProjected}(\follow 1{\ledge {\mathbb I}v}, b\oplus y, k)$
				\State \Return $|\lambda|^2\cdot (\alpha_0+\alpha_1)$
			\EndIf
		\EndProcedure
	\end{algorithmic}
\end{algorithm}

After measurement of a qubit $k$, a quantum state is typically projected to $\ket 0$ or $\ket 1$
 ($b=0$ or $b=1$) on that qubit, depending on the outcome.
\autoref{alg:measure-arbitrary-qubit-update2} realizes this.
It does so by traversing the \limdd until a node $v$ with $\index(v) = k$ is reached.
It then returns an edge to a new node by calling $\makeedge(\follow 0e, 0)$ to project onto $\ket 0$ or $\makeedge(0, \follow 1e)$ to project onto $\ket 1$, on \autoref{algline:project-project},
recreating a node on level $k$ in the backtrack on \autoref{l:project-bt}.
The projection operator $\Pi_k^b$ commutes with
any LIM $P$ when $P_k$ is a diagonal operator (i.e., $P_k\in \{\id[2],Z\}$).
Otherwise, if $P_k$ is an antidiagonal operator (i.e, $P_k\in \{X,Y\}$), have $\Pi_k^b\cdot P=P\Pi_k^{(1-b)}$.
The algorithm applies this correction on \autoref{l:project-diag2}.
The resulting state should still be normalized as shown in \autoref{sec:measurement}.

\begin{algorithm}
    \begin{algorithmic}[1]
        \Procedure{\project}{\Edge $\ledge {\lambda P_n\otimes \cdots \otimes P_1}v$, $k\in [1 \dots \index(v)]$, $b\in\set{0,1}$}
        \State $b' :=  x \oplus b $ \textbf{where} $x=0$ if $P_k\in\set{\mathbb I,Z}$ and $x=1$ if $P_k\in\set{X,Y}$
        \label{l:project-diag2}  \Comment{flip $b$ if $P_k$ is anti-diagonal}         
        \State \textbf{if} $(v,k, b') \in \cache$ \textbf{then} \Return $\cache[v,k, b']$
		\State $n := \index(v)$
         \If{$n = k$}
                \State $e := \makeedge((1- b') \cdot \Low_v,~~ b' \cdot \High_v)$ 
                                                    \Comment{Project $\ket v$ to $\ket{b'}\bra{b'}\otimes \id[2]^{\otimes n-1}$}   
                     \label{algline:project-project}
            \Else \Comment{$n \neq k$:}
                \State $e :=\textsc{MakeEdge}(\project(\Low_v,k, b'),
                                            \project(\High_v,k,b'))$ \label{l:project-bt}
             \EndIf
                 \Stateh $\cache[v,k, b'] := e$
                 \State \Return $e$
                \EndProcedure
        \end{algorithmic}
        \caption{Project the state given by \limdd \ledge Av to state $\ket b$ for qubit $k$, i.e., produce a \limdd representing the state
           $(\id[n-k] \otimes \ket{b}\bra{b}\otimes \id[k-1])\cdot A\ket{v}$, with $A=\lambda P_n\otimes \cdots \otimes P_1$.}
        \label{alg:measure-arbitrary-qubit-update2}
\end{algorithm}

\section{\limdds prepare the W state efficiently}
\label{sec:efficient-w-state-preparation}

In this section, we prove \autoref{thm:limdds-are-fast-for-wn}. To this end, we show that \limdds can efficiently simulate a circuit family given by McClung \cite{mcclung2020constructionsOfWStates}, which prepares the $\ket W$ state when initialized to the $\ket 0$ state.
We thereby show a separation between \limdd and the Clifford+$T$ simulator, as explained in \autoref{sec:simulation:cliffordt}.
Figure \autoref{fig:circuit-prepare-W} shows the circuit for the case of $8$ qubits.

\begin{figure}[b!]
	\begin{center}
		\includegraphics[width=0.8\linewidth]{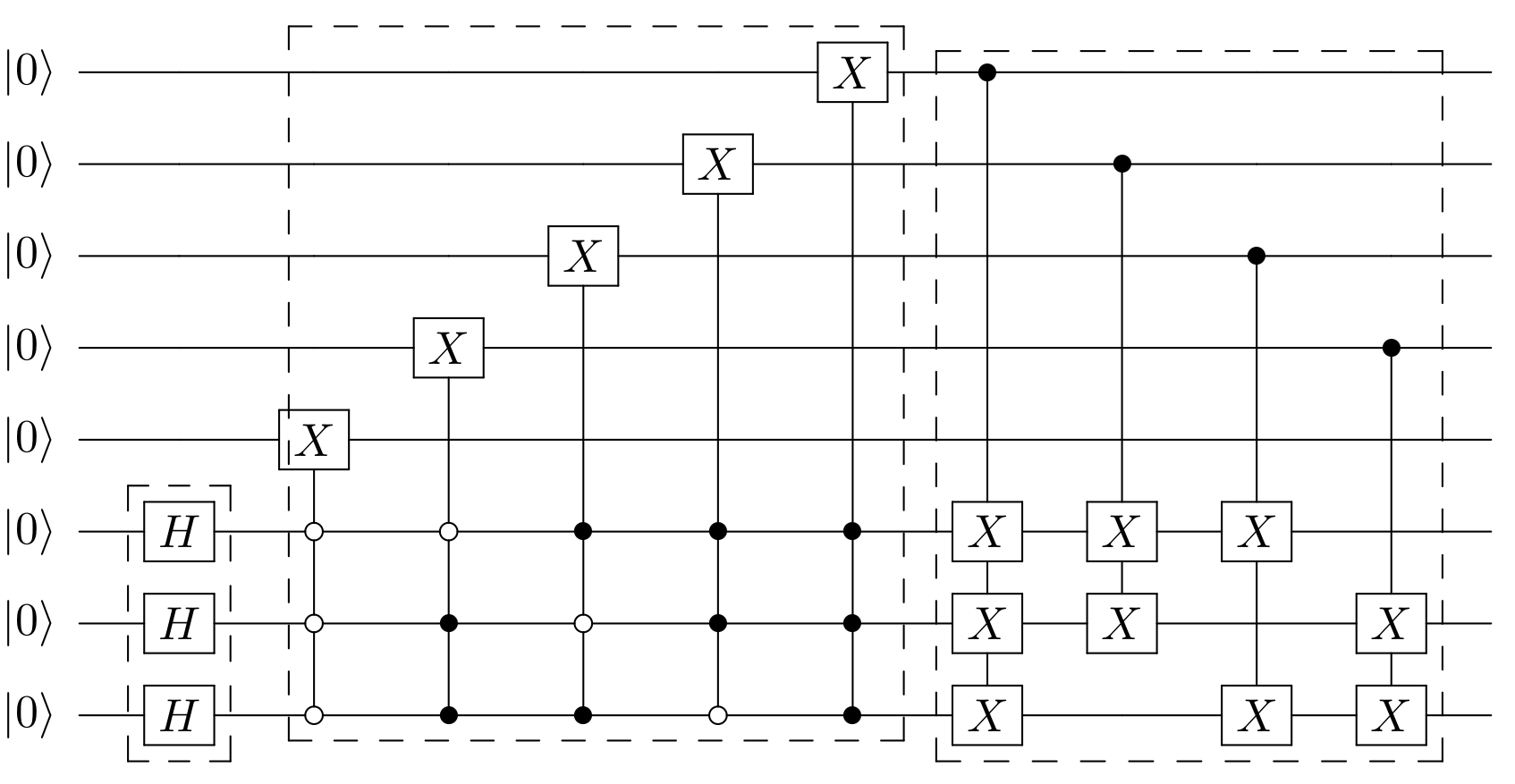}
	\end{center}
	\caption{Reproduced from McClung \cite{mcclung2020constructionsOfWStates}.
	A circuit on eight qubits $(n=8)$ which takes as input the ${\ket 0}^{\otimes 8}$ state and outputs the $\ket{W_8}$ state.
	In the general case, it contains $\log n$ Hadamard gates, and its Controlled-$X$ gates act on one target qubit and at most $\log n$ control qubits.}
	\label{fig:circuit-prepare-W}
\end{figure}

\thmw*
\begin{proof}
The proof outline is as follows.
First, we establish that the \limdd of each intermediate state (\autoref{thm:intermediate-states-have-small-limdd}), as well as of each gate (\autoref{thm:gates-have-small-limdd}), has polynomial size.
Second, we establish that the algorithms presented in \autoref{sec:quantum-simulation} can apply each gate to the intermediate state in polynomial time (\autoref{thm:applygate-polynomial-time}).
To this end, we observe that the circuit only produces relatively simple intermediate states.
Specifically, each intermediate and output state is of the form $\ket{\psi_t}=\frac{1}{\sqrt n}\sum_{k=1}^{n}\ket{x_k}$ where the $x_k\in \{0,1\}^n$ are computational basis vectors (\autoref{thm:intermediate-state-has-small-comprank}).
For example, the output state has $\ket{x_k}=\ket{0}^{k-1}\ket 1\ket{0}^{\otimes n-k}$.
The main technical tool we will use to reason about the size of the \limdds of these intermediate states, are the \emph{subfunction rank} and \emph{computational basis rank} of a state.
Both these measures are upper bounds of the size of a \limdd (in \autoref{thm:limdd-width-bounded-by-comprank}), and also allow us to upper bound the time taken by the \textsc{ApplyGate} and \textsc{Add} algorithms (in \autoref{thm:applygate-polynomial-calls} for \textsc{ApplyGate} and \autoref{thm:add-calls-bounded-by-subfunction-rank} \textsc{Add}).

The theorem follows from \autoref{thm:applygate-polynomial-time} and \autoref{cor:w}.
\end{proof}

\autoref{fig:circuit-prepare-W} shows the circuit for the case of $n=8$ qubits.
For convenience and without loss of generality, we only treat the case when the number of qubits is a power of $2$, since the circuit is simplest in that case.
In general, the circuit works as follows.
The qubits are divided into two registers; register $A$, with $\log n$ qubits, and register $B$, with the remaining $n-\log n$ qubits.
First, the circuit applies a Hadamard gate to each qubit in register $A$, to bring the state to the superposition $\ket +^{\otimes \log n}\ket 0^{n-\log n}$.
Then it applies $n-\log n$ Controlled-$X$ gates, where, in each gate, each qubit of register $A$ acts as the control qubits and one qubit in register $B$ is the target qubit.
Lastly, it applies $n-\log n$ Controlled-$X$ gates, where, in each gate, one qubit in register $B$ is the control qubit and one or more qubits in register $A$ are the target qubits.
Each of the three groups of gates is highlighted in a dashed rectangle in \autoref{fig:circuit-prepare-W}.
On input $\ket 0^{\otimes n}$, the circuit's final state is $\ket{W_n}$.
We emphasize that the Controlled-$X$ gates are permutation gates (i.e., their matrices are permutation matrices).
Therefore, these gates do not influence the number of non-zero computational basis state amplitudes of the intermediate states.
We refer to the $t$-th gate of this circuit as $U_t$, and the $t$-th intermediate state as $\ket{\psi_t}$, so that $\ket{\psi_{t+1}}=U_t\ket{\psi_{t}}$ and $\ket{\psi_0}=\ket 0$ is the initial state.

\newcommand{\comprank}[1]{\ensuremath{\chi_{\textnormal{comp}}(#1)}}
\newcommand{\subrankspecific}[2]{\ensuremath{\chi_{\textnormal{sub}}(#1,#2)}}
\newcommand{\subrank}[1]{\ensuremath{\chi_{\textnormal{sub}}(#1)}}

We refer to the number of computational basis states with nonzero amplitude as a state's \emph{computational basis rank}, denoted $\comprank{\ket\psi}$.
\begin{definition}{(Computational basis rank)}
	\label{def:computational-basis-rank}
	Let $\ket{\psi}=\sum_{x\in \{0,1\}^n}\alpha(x)\ket{x}$ be a quantum state defined by the amplitude function $\alpha\colon \{0,1\}^n\to\mathbb C$.
	Then the \emph{computational basis rank} of $\ket{\psi}$ is $\comprank{\ket\psi}=|\set{x\mid\alpha(x)\ne 0}|$, the number of nonzero computational basis amplitudes.
\end{definition}
Recall that, for a given function $\alpha\colon \{0,1\}^{n}\to\mathbb C$, a string $a\in \{0,1\}^\ell$ induces a \emph{subfunction} $\alpha_y\colon \{0,1\}^{n-\ell}\to\mathbb C$, defined as $\alpha_y(x)=\alpha(y,x)$.
We refer to the number of subfunctions of a state's amplitude function as its \emph{subfunction rank}.
The following definition makes this more precise.
\begin{definition}{(Subfunction rank)}
	Let $\ket{\psi}=\sum_{x\in \{0,1\}^n}\alpha^\psi(x)\ket{x}$ be a quantum state defined by the amplitude function $\alpha^\psi\colon \{0,1\}^n\to\mathbb C$, as above.
	Let $\subrankspecific{\ket\psi}{\ell}$ be the number of unique non-zero subfunctions induced by strings of length $\ell$, as follows,
	\begin{align}
	\subrankspecific{\ket\psi}{\ell}=|\set{\alpha^\psi_y\colon \{0,1\}^{n-\ell}\to\mathbb C\mid\alpha_y\ne 0,y\in\{0,1\}^\ell }|
	\end{align} 
	We define the \emph{subfunction rank} of $\ket\psi$ as $\subrank{\ket\psi}=\max_{\ell=0,\ldots n}\subrankspecific{\ket\psi}{\ell}$.
	We extend these definitions in the natural way for an $n$-qubit matrix $U=\sum_{r,c\in\{0,1\}^n}\alpha^U(r,c)\ket{r}\bra{c}$ defined by the function $\alpha^U\colon \{0,1\}^{2n}\to\mathbb C$.
\end{definition}

It is easy to check that $\subrank{\ket\psi} \leq \comprank{\ket\psi}$ holds for any state.

For the next lemma, we use the notion of a \emph{prefix} of a \limdd node.
This lemma will serve as a tool which allows us to show that a \limdd is small when its computational basis rank is low.
We apply this tool to the intermediate states of the circuit in \autoref{thm:intermediate-states-have-small-limdd}.
\begin{definition}[Prefix of a \limdd node]
	For a given string $x\in \{0,1\}^\ell$, consider the path traversed by the $\follow x{\ledge Rr}$ subroutine, which starts at the diagram's root edge and ends at a node $v$ on level $\ell$.
	We will say that $x$ is a \emph{prefix} of the node $v$.
	We let $\text{Labels}(x)$ be the product of the LIMs on the edges of this path (i.e., including the root edge).
	The set of prefixes of a node $v$ is denoted $\text{pre}(v)$.
\end{definition}

\begin{lemma}
	\label{thm:limdd-width-bounded-by-comprank}
	If a \limdd represents the state $\ket\phi$, then its width at any given level (i.e., the number of nodes at that level) is at most $\chi_{\textnormal{comp}}(\ket\phi)$.
\end{lemma}
\begin{proof}
	For notational convenience, let us number the levels so that the root node is on level $0$, its children are on level $1$, and so on, with the Leaf on level $n$ (contrary to \autoref{fig:qmdd-isoqmdd-exposition}).
	Let $r$ be the root node of the \limdd, and $R$ the root edge's label.
	By construction of a \limdd, the state represented by the \limdd can be expressed as follows, for any level $\ell\geq 0$,
	\begin{align}
	\label{eq:root-node-as-follow-superposition}
	R\ket r = \sum_{x\in \{0,1\}^\ell}\ket{x}\otimes \follow x{\ledge Rr}
	\end{align}
	Since $\ledge Rr$ is the root of our diagram, if $x$ is a prefix of $v$, then
	\begin{align}
	\label{eq:state-as-prefix-superposition}
	\follow x{\ledge Rr} = \text{Labels}(x)\cdot \ket{v}
	\end{align}
	A string $x\in\{0,1\}^\ell$ can be a prefix of only one node; consequently, the prefix sets of two nodes on the same level are disjoint, i.e., $\text{pre}(v_p)\cap \text{pre}(v_q)=\emptyset$ for $p\ne q$.
	Moreover, each string $x$ is a prefix of \emph{some} node on level $\ell$ (namely, simply the node at which the $\follow x{\ledge Rr}$ subroutine arrives).
	Say that the $\ell$-th level contains $m$ nodes, $v_1,\ldots, v_m$.
	Therefore, the sets $\text{pre}(v_1),\ldots, \text{pre}(v_m)$ partition the set $\{0,1\}^\ell$.
	Therefore, by putting \autoref{eq:state-as-prefix-superposition} and \autoref{eq:root-node-as-follow-superposition} together, we can express the root node's state in terms of the nodes $v_1,\ldots, v_m$ on level $\ell$:
	\begin{align}
	R\ket r = & \sum_{k=1}^m \sum_{x\in \text{pre}(v_k)}\ket x\otimes \follow x{\ledge Rr} \\
	= & \sum_{k=1}^m \sum_{x\in \text{pre}(v_k)}\ket x\otimes \text{Labels}(x)\cdot \ket{v_k}
	\end{align}
	We now show that each term $\sum_{x\in \text{pre}(v_k)}\ket{x}\otimes \text{Labels}(x)\cdot\ket {v_k}$ contributes a non-zero vector.
	It then follows that the state has computational basis rank at least $m$, since these terms are vectors with pairwise disjoint support, since the sets $\text{pre}(v_k)$ are pairwise disjoint.
	Specifically, we show that each node has at least one prefix $x$ such that $\text{Labels}(x)\cdot\ket{v}$ is not the all-zero vector.
	In principle, this can fail in one of three ways: either $v$ has no prefixes, or all prefixes $x\in \text{pre}(v_k)$ have $\text{Labels}(x)=0$ because the path contains an edge labeled with the $0$ LIM, or the node $v$ represents the all-zero vector (i.e., $\ket v=\vec{0}$).
	First, we note that each node has at least one prefix, since each node is reachable from the root, as a \limdd is a connected graph.
	Second, due to the zero edges rule (see \autoref{def:reduced-limdd}), for any node, at least one of its prefixes has only non-zero LIMs on the edges.
	Namely, each node $v$ has at least one incoming edge labeled with a non-zero LIM, since, if it has an incoming edge from node $w$ labeled with $0$, then this must be the high edge of $w$ and by the zero edges rule the low edge of $w$ must also point to $v$ and moreover must be labeled with $\mathbb I$ by the low factoring rule.
	Together, via a simple inductive argument, there must be at least one non-zero path from $v$ to the root.
	Lastly, no node represents the all-zero vector, due to the low factoring rule (in \autoref{def:reduced-limdd}).
	Namely, if $v$ is a node, then by the low factoring rule, the low edge has label $\mathbb I$.
	Therefore, if this edge points to node $v_0$, and the high edge is $\ledge A{v_1}$, then the node $v$ represents $\ket{v} = \ket{0}\ket{v_0}+\ket{1}A\ket{v_1}$ with possibly $A=0$, so, if $\ket{v_0}\ne \vec{0}$, then $\ket v\ne \vec{0}$.
	An argument by induction now shows that no node in the reduced \limdd represents the all-zero vector.

	Therefore, each node has at least one prefix $x$ such that $\follow x{\ledge Rr}\ne \vec 0$.
	We conclude that the equation above contains at least $m$ non-zero contributions.
	Hence $m\leq \comprank{R\ket r}$, at any level $0\leq \ell\leq n$.
\end{proof}

\begin{lemma}
	\label{thm:intermediate-state-has-small-comprank}
	Each intermediate state in the circuit in \autoref{fig:circuit-prepare-W}  (with $n=2^c$) has $\comprank{\ket\psi}\leq n$.
\end{lemma}
\begin{proof}
	The initial state is $\ket{\psi_0}=\ket0^{\otimes n}$, which is a computational basis state, so $\comprank{\psi_0}=1$.
	The first $\log n$ gates are Hadamard gates, which produce the state 
	\begin{align}
	\ket{\psi_{\log n}}=H^{\otimes \log n}\otimes \mathbb I^{n-\log n}\ket 0 = \ket+^{\otimes \log n}\otimes \ket 0^{\otimes n-\log n}=\frac 1{\sqrt n}\sum_{x=0}^{n-1}\ket x\ket 0^{\otimes n-\log n}
	\end{align}
	This is a superposition of $n$ computational basis states, so we have $\comprank{\ket{\psi_{\log n}}}=n$.
	All subsequent gates are controlled-$X$ gates; these gates permute the computational basis states, but they do not increase their number.
\end{proof}

\begin{lemma}
	\label{thm:intermediate-states-have-small-limdd}
	The reduced \limdd of each intermediate state in the circuit in \autoref{fig:circuit-prepare-W} has polynomial size.
\end{lemma}
\begin{proof}
	By \autoref{thm:limdd-width-bounded-by-comprank}, the width of a \limdd representing $\ket\phi$ is at most $\chi_{\textnormal{comp}}(\ket\phi)$ at any level.
	Since there are $n$ levels, the total size is at most $n\chi_{\textnormal{comp}}(\ket\phi)$.
	By \autoref{thm:intermediate-state-has-small-comprank}, the intermediate states in question have polynomial $\chi_{\textnormal{comp}}$, so the result follows.
\end{proof}

\begin{lemma}
	\label{thm:gates-have-small-limdd}
	The \limdd of each gate in the circuit in \autoref{fig:circuit-prepare-W}  (with $n=2^c$) has polynomial size.
\end{lemma}
\begin{proof}
Each gate acts on at most $k=\log n+1$ qubits.
Therefore, the width of any level of the \limdd is at most $4^k=4n^2$.
The height of the \limdd is $n$ by definition, so the \limdd has at most $4n^3$ nodes.
\end{proof}

The \textsc{ApplyGate} procedure handles the Hadamard gates efficiently, since they apply a single-qubit gate to a product state.
The difficult part is to show that the same holds for the controlled-$X$ gates.
To this end, we show a general result for the speed of \limdd operations (\autoref{thm:applygate-polynomial-calls}).
Although this worst-case upper bound is tight, it is exponentially far removed from the best case, e.g., in the case of Clifford circuits, in which case the intermediate states can have exponential $\chi_\text{sub}$, yet the \limdd simulation is polynomial-time, as shown in \autoref{sec:clifford-polytime}.

\begin{lemma}
	\label{thm:applygate-polynomial-calls}
	The number of recursive calls made by subroutine \textsc{ApplyGate}$(U,\ket \psi)$, is at most $n\subrank U\subrank{\ket\psi}$, for any gate $U$ and any state $\ket\psi$.
\end{lemma}
\begin{proof}
	Inspecting \autoref{alg:apply-gate-limdd-limdd}, we see that every call to $\textsc{ApplyGate}(U,\ket\psi)$ produces four new recursive calls, namely $\textsc{ApplyGate}(\follow{rc}U,\follow c{\ket\psi})$ for $r,c\in \{0,1\}$.
	Therefore, the set of parameters in all recursive calls of $\textsc{ApplyGate}(U,\ket\psi)$ is precisely the set of tuples $(\follow {rc}{U},\follow {c}{\ket\psi})$, with $r,c\in \{0,1\}^\ell$ with $\ell=0\ldots n$.
	The terms $\follow {rc}U$ and $\follow c{\ket\psi}$ are precisely the subfunctions of $U$ and $\ket\psi$, and since there are at most $\subrank U$ and $\subrank {\ket\psi}$ of these, the total number of distinct parameters passed to \textsc{ApplyGate} in recursive calls at level $\ell$, is at most $\subrank{U,\ell} \cdot \subrank{\ket\psi,\ell}\leq \subrank U\cdot \subrank{\ket{\psi}}$.
	Summing over the $n$ levels of the diagram, we see that there are at most $n\subrank U\subrank{\ket\psi}$ distinct recursive calls in total.
	As detailed in \autoref{sec:applygate}, the \textsc{ApplyGate} algorithm caches its inputs in such a way that it will achieve a cache hit on a call $\textsc{ApplyGate}(U',\ket{\psi'})$ when it has previously been called with parameters $U,\ket\psi$ such that $U=U'$ and $\ket\psi=\ket{\psi'}$.
	Therefore, the total number of recursive calls that is made, is equal to the number of \emph{distinct} calls, and the result follows.
\end{proof}

In our case, both $\subrank U$ and $\subrank{ \ket\psi}$ are polynomial, so a polynomial number of recursive calls to \textsc{ApplyGate} is made.
We now show that also the \textsc{Add} subroutine makes only a small number of recursive calls every time it is called from \textsc{ApplyGate}.
First, \autoref{thm:add-calls-bounded-by-subfunction-rank} shows expresses a worst-case upper bound on the number of recursive calls to \textsc{Add} in terms of $\chi_{sub}$.
Then \autoref{thm:add-polynomial-subfunction-rank} uses this result to show that, in our circuit, the number of recursive calls is polynomial in $n$.

\begin{lemma}
	\label{thm:add-calls-bounded-by-subfunction-rank}
	The number of recursive calls made by the subroutine $\textsc{Add}(\ket \alpha,\ket \beta)$ is at most $n\chi_{sub}(\ket \alpha)\cdot \chi_{sub}(\ket\beta)$, if $\ket\alpha,\ket\beta$ are $n$-qubit states.
\end{lemma}
\begin{proof}
	Inspecting \autoref{alg:add-limdds}, every call to $\textsc{Add}(\ket\alpha,\ket\beta)$ produces two new recursive calls, namely $\textsc{Add}(\follow 0{\ket\alpha},\follow 0{\ket\beta})$ and $\textsc{Add}(\follow 1{\ket\alpha},\follow 1{\ket\beta})$.
	It follows that the set of parameters on $n-\ell$ qubits with which \textsc{Add} is called is the set of tuples $(\follow x{\ket\alpha},\follow x{\ket\beta})$, for $x\in \{0,1\}^\ell$.
	This corresponds precisely to the set of subfunctions of $\alpha$ and $\beta$ induced by length-$\ell$ strings, of which there are $\subrankspecific{\ket\alpha}{\ell}$ and $\subrankspecific{\ket\beta}{\ell}$, respectively.
	Because the results of previous computations are cached, as explained in \autoref{sec:applygate}, the total number of recursive calls is the number of \emph{distinct} recursive calls.
	Therefore, we get the upper bound of $\subrank{\ket\alpha}\cdot \subrank{\ket\beta}$ for each level of the \limdd.
	Since the \limdd has $n$ levels, the upper bound $n\subrank{\ket\alpha}\cdot \subrank{\ket\beta}$ follows.
\end{proof}

\begin{lemma}
	\label{thm:add-polynomial-subfunction-rank}
	The calls to \textsc{Add}$(\ket\alpha,\ket\beta)$ that are made by the recursive calls to $\textsc{ApplyGate}(U_t,\ket{\psi_t})$, satisfy $\subrank{\ket\alpha},\subrank{\ket\beta}=\text{poly}(n)$.
\end{lemma}
\begin{proof}
	We have established that the recursive calls to \textsc{ApplyGate} are all called with parameters of the form $\textsc{ApplyGate}(\follow {r,c}{U_t}, \follow c{\ket{\psi_t}})$ for some $r,c\in\{0,1\}^\ell$.
	Inspecting \autoref{alg:apply-gate-limdd-limdd}, we see that, within such a call, each call to $\textsc{Add}(\ket\alpha,\ket\beta)$ has parameters which are both of the form $\ket\alpha,\ket\beta=\textsc{ApplyGate}(\follow{rx,cy}{U_t}, \follow {cy}{\ket{\psi_t}})$ for some $x,y\in\{0,1\}$; therefore, the parameters $\ket\alpha,\ket\beta$ are of the form $\ket\alpha,\ket\beta = \follow{r,c}{U_t}\cdot \follow r{\ket{\psi_t}}$.
	Here $\follow {cy}{\ket{\psi_t}}$ is a quantum state on $n-(\ell+1)$ qubits.
	
	The computational basis rank of a state is clearly non-increasing under taking subfunctions; that is, for any string $x$, it holds that, $\comprank{\follow {x}{\ket{\psi}}}\leq \comprank{\ket{\psi}}$.
	In particular, we have $\comprank{\follow{cy}{\ket{\psi_t}}}\leq \comprank{\ket{\psi_t}}=\mathcal O(n)$.
	The matrix $\follow{rx,cy}{U_t}$ is a subfunction of a permutation gate, and applying such a matrix to a vector cannot increase its computational basis rank, so we have
	\begin{align}
	\subrank{\ket\alpha}=& \subrank{\follow{rx,cy}{U_t}\cdot \follow{cy}{\ket{\psi_t}}} \\
	\leq & \comprank{\follow{rx,cy}{U_t}\cdot \follow{cy}{\ket{\psi_t}}} \leq \comprank{\follow{cy}{\ket{\psi_t}}} \\
	\leq & \comprank{\ket{\psi_t}}=\mathcal O(n)
	\end{align}
	This proves the lemma.
\end{proof}

\begin{lemma}
	\label{thm:applygate-polynomial-time}
	Each call to $\textsc{ApplyGate}(U_t,\ket{\psi_t})$ runs in polynomial time, for any gate $U_t$ in the circuit in \autoref{fig:circuit-prepare-W} (with $n=2^c$).
\end{lemma}
\begin{proof}
	If $U_t$ is a Hadamard gate, then \limdds can apply this in polynomial time by \autoref{thm:hadamard-stabilizer-polytime}, since $\ket{\psi_t}$ is a stabilizer state.
	Otherwise, $U_t$ is one of the controlled-$X$ gates.
	In this case there are a polynomial number of recursive calls to $\textsc{ApplyGate}$, by \autoref{thm:applygate-polynomial-calls}.
	Each recursive call to $\textsc{ApplyGate}$ makes two calls to \textsc{Add}$(\ket\alpha,\ket\beta)$, where both $\alpha$ and $\beta$ are states with polynomial subfunction rank, by \autoref{thm:add-polynomial-subfunction-rank}.
	By \autoref{thm:add-calls-bounded-by-subfunction-rank}, these calls to \textsc{Add} all complete in time polynomial in the subfunction rank of its arguments.
\end{proof}

\begin{corollary}
\label{cor:w}
The circuit in \autoref{fig:circuit-prepare-W} (with $n=2^c$) can be simulated by \limdds in polynomial time.
\end{corollary}
\section{Numerical search for the stabilizer rank of Dicke states}
\label{sec:stabrank_search}

Given the separation between the Clifford + T simulator ---a specific stabilizer-rank based simulator--- and Pauli-\limdds, it would be highly interesting to theoretically compare Pauli-\limdds and general stabilizer-rank simulation.
However, proving an exponential separation would require us to find a family of states for which we can prove its stabilizer rank scales exponentially, which is a major open problem.
Instead, we here take the first steps towards a numerical comparison by choosing a family of circuits which Pauli-\limdds can efficiently simulate and using Bravyi et al.'s heuristic algorithm for searching the stabilizer rank of the circuits' output states~\cite{bravyi2016trading}.
If the stabilizer rank is very high (specifically, if it grows superpolynomially in the number of qubits), then we have achieved the goal of showing a separation.
We cannot use $W$ states for showing this separation because the $n$-qubit $W$ state $\ket{W_n}$ has linear stabilizer rank, since it is a superposition of only $n$ computational basis states.
Instead we focus on their generalization, Dicke states $\dicke{n}{w}$, which are equal superpositions of all $n$-qubit computational-basis status with Hamming weight $w$ (note $\ket{W_n} = \dicke{n}{1}$),
\begin{align}
\label{eq:dicke}
\dicke{n}{w} = \frac{1}{\sqrt{\binom nw}}\sum_{x:|x|=w} \ket x
\end{align}

We implemented the algorithm by Bravyi et al.: see \cite{githubrepo} for our open-source implementation.
Unfortunately, the algorithm's runtime grows significantly in practice, which we believe is due to the fact that it acts on sets of quantum state vectors, which are exponentially large in the number of qubits.
Our implementation allowed us to go to at most $9$ qubits using the SURF supercomputing cluster.
We believe this is a limitation of the algorithm and not of our implementation, since Bravyi et al. do not report beyond $6$ qubits while Calpin uses the same algorithm and reaches at most $10$ qubits~\cite{calpin2020exploring}.
Table~\ref{table:stabilizer-rank-search} shows the heuristically found stabilizer ranks of Dicke states with our implementation.
Although we observe the maximum found rank over $w$ to grow quickly in $n$, the feasible regime (i.e. up to $9$ qubits) is too small to draw a firm conclusion on the stabilizer ranks' scaling.

Since our heuristic algorithm finds only an upper bound on the stabilizer rank, and not a lower bound, by construction we cannot guarantee any statement on the scaling of the rank itself.
However, our approach could have found only stabilizer decompositions of very low rank, thereby providing evidence that Dicke states have very slowly growing rank, meaning that stabilizer-rank methods can efficiently simulate circuits which output Dicke states.
This is not what we observe; at the very least we can say that, if Dicke states have low stabilizer rank, then the current state-of-the-art method by Bravyi et al. does not succeed in finding the corresponding decomposition.
Further research is needed for a conclusive answer.

We now explain the heuristic algorithm by Bravyi et al. \cite{bravyi2016trading}, which has been explained in more detail in \cite{calpin2020exploring}.
The algorithm follows a simulated annealing approach: on input $n,w $ and $\chi$, it ~performs a random walk through sets of $\chi$ stabilizer states.
It starts with a random set $V$ of $\chi$ stabilizer states on $n$ qubits.
In a single `step', the algorithm picks one of these states $\ket{\psi} \in V$ at random, together with a random $n$-qubit Pauli operator $P$, and replaces the state $\ket{\psi}$ with $\ket{\psi'} := c(\id + P)\ket{\psi}$ with $c$ a normalization constant (or repeats if $\ket{\psi'} = 0$), yielding a new set $V'$.
The step is accepted with certainty if $F_V < F_{V'}$, where $F_V:= |\dickebra nw \Pi_V \dicke{n}{w}|$ with $\Pi_V$ the projector on the subspace of the $n$-qubit Hilbert space spanned by the stabilizer states in $V$.
Otherwise, it is accepted with probability $\exp(-\beta (F_{V'} - F_V))$, where $\beta$ should be interpreted as the inverse temperature.
The algorithm terminates if it finds $F_V = 1$, implying that $\dicke{n}{w}$ can be written as linear combination of $V$, outputting the number $\chi$ as (an upper bound on) the stabilizer rank of $\ket\psi$.
For a fixed $\chi$, we use identical values to Bravyi et al. \cite{bravyi2016trading} and vary $\beta$ from $1$ to $4000$ in $100$ steps, performing $1000$ steps at each value of $\beta$.

\begin{table}[t!]
    \caption{
        Heuristically-found upper bounds on the stabilizer rank $\chi$ of Dicke states $\dicke nw$ (eq.~\eqref{eq:dicke}) using the heuristic algorithm from Bravyi et al. \cite{bravyi2016trading}, see text in \autoref{sec:stabrank_search} for details.
        We investigated up to $9$ qubits using the SURF supercomputing cluster (approximately the same as the number of qubits reached in the literature as described in the text).
        Empty cells indicate non-existing or not-investigated states.
        In particular, we have not investigated $w> \lfloor \frac{n}{2}\rfloor$ since $\dicke{n}{w}$ and $\dicke{n}{n-w}$ have identical stabilizer rank because $X^{\otimes n} \dicke{n}{w} = \dicke{n}{n-w}$.
        For $\dicke{8}{3}$ and $\dicke{9}{4}$, we have run the heuristic algorithm to find sets of stabilizers up to size $11$ (theoretical upper bound) and $10$, respectively, but the algorithm has not found sets in which these two Dicke states could be decomposed.
        We emphasize that the algorithm is heuristic, so even if there exists a stabilizer decomposition of a given rank, the algorithm might not find it.
    }
        \label{table:stabilizer-rank-search}
    \centering
    \setlength{\tabcolsep}{12pt}
    \begin{tabular}{|r|ccccc|}
        \hline
        & \multicolumn{5}{c|}{\textbf{Hamming weight $w$}}\\\hline
        \textbf{\#qubits $n$} & \textbf{0} & \textbf{1} & \textbf{2} & \textbf{3} & \textbf{4} \\\hline
        \textbf{1} & 1 &   &   &   &   \\
        \textbf{2} & 1 & 1 &   &   &   \\
        \textbf{3} & 1 & 2 &   &   &   \\
        \textbf{4} & 1 & 2 & 2 &   &   \\
        \textbf{5} & 1 & 3 & 2 &   &   \\
        \textbf{6} & 1 & 3 & 4 & 2 &   \\
        \textbf{7} & 1 & 4 & 7 & 4 &   \\
        \textbf{8} & 1 & 4 & 8 & $\leq 11$  & 5 \\
        \textbf{9} &   &   &   &   & $>10?$  \\\hline
    \end{tabular}
\end{table}

\end{document}